\documentclass[11pt]{article}
\usepackage[utf8]{inputenc}

\usepackage{fullpage}
\usepackage{amsmath}
\usepackage{amssymb}
\usepackage{amsthm}
\usepackage{hyperref}
\usepackage{bbm}
\usepackage[dvipsnames]{xcolor}
\usepackage{caption}
\usepackage{subcaption}
\usepackage{float}
\usepackage{graphicx}
\usepackage{algorithmic}
\usepackage[linesnumbered,ruled,vlined]{algorithm2e}
\usepackage{todonotes} 

\usepackage{enumerate}
\usepackage[normalem]{ulem}

\def\balpha{\boldsymbol{\alpha}}
\def\bbeta{\boldsymbol{\beta}}

\def\bmu{\boldsymbol{\mu}}

\def\blambda{\boldsymbol{\lambda}}

\newcommand\cF{\mathcal F}

\newcommand\cK{\mathcal K}
\newcommand\cL{\mathcal L}

\newcommand\cP{\mathcal P}

\newcommand\bX{\boldsymbol{X}}

\newcommand\bsX{\boldsymbol{X}}
\newcommand\bsY{\boldsymbol{Y}}
\newcommand\bsZ{\boldsymbol{Z}}

\newcommand\EE{\mathbb E}

\newcommand\PP{\mathbb P}

\newcommand\RR{\mathbb R}

\DeclareMathOperator*{\argmin}{arg\,min}
\newcommand{\rMKV}{\rm{MFC}}
\newcommand{\rMFG}{\rm{MFG}}
\definecolor{burgundy}{rgb}{0.5,0.0, 0.13}

\hypersetup{colorlinks,
citecolor=burgundy,
citebordercolor=burgundy,
linkcolor=burgundy,
urlcolor=MidnightBlue
}

\newtheorem{theorem}{Theorem}[section]
\newtheorem{remark}[theorem]{Remark}
\newtheorem{definition}[theorem]{Definition}

\newtheorem{proposition}[theorem]{Proposition}

\newtheorem{assumption}[theorem]{Assumption}

\newtheorem{corollary}[theorem]{Corollary}

\title{From Nash Equilibrium to Social Optimum and vice versa:\\ a Mean Field Perspective}

\author{Ren\'e Carmona
\footnote{ORFE, Bendheim Center for Finance, Princeton University,
Princeton, NJ 08544, USA 
  (\href{mailto:rcarmona@princeton.edu}{rcarmona@princeton.edu}).}
\and G\"ok\c ce Dayan{\i}kl{\i}\footnote{Department of Statistics, University of Illinois at Urbana-Champaign, 
  Champaign, IL 61820, USA 
  (\href{mailto:gokced@illinois.edu}{gokced@illinois.edu}).}
\and Fran{\c c}ois Delarue
\footnote{Laboratoire J.A. Dieudonn\'e, Universit\'e de Nice Sophia-Antipolis, Parc Valrose, 06108 Nice Cedex 02, France 
  (\href{mailto:delarue@unice.fr}{delarue@unice.fr}).}
\and Mathieu Lauri\`ere
    \footnote{Shanghai Frontiers Science Center of Artificial Intelligence and Deep Learning; NYU-ECNU Institute of Mathematical Sciences at NYU Shanghai; NYU Shanghai, 567 West Yangsi Road, Shanghai, 200126, People’s Republic of China, 
  (\href{mailto:ml5197@nyu.edu}{ml5197@nyu.edu}).}
}

\date{}

\begin{document}

\maketitle

\abstract{
Mean field games (MFG) and mean field control (MFC) problems have been introduced to study large populations of strategic players. They correspond respectively to non-cooperative or cooperative scenarios, where the aim is to find the Nash equilibrium and social optimum. These frameworks provide approximate solutions to situations with a finite number of players and have found a wide range of applications, from economics to biology and machine learning. In this paper, we study how the players can pass from a non-cooperative to a cooperative regime, and vice versa. The first direction is reminiscent of mechanism design, in which the game's definition is modified so that non-cooperative players reach an outcome similar to a cooperative scenario. The second direction studies how players that are initially cooperative gradually deviate from a social optimum to reach a Nash equilibrium when they decide to optimize their individual cost similar to the free rider phenomenon. To formalize these connections, we introduce two new classes of games which lie between MFG and MFC: $\lambda$-interpolated mean field games, in which the cost of an individual player is a $\lambda$-interpolation of the MFG and the MFC costs, and $p$-partial mean field games, in which a proportion $p$ of the population deviates from the social optimum by playing the game non-cooperatively. We conclude the paper by providing an algorithm for myopic players to learn a $p$-partial mean field equilibrium, and we illustrate it on a stylized model.

}


\vskip 3mm
\noindent\emph{\textbf{Acknowledgments.}}
{François Delarue acknowledges the financial support of the European Research Council (ERC) under the European Union’s Horizon 2020 research and innovation programme (ELISA project, Grant agreement No. 101054746).}

{\small\tableofcontents}

\section{Introduction}
\label{sec:intro}

In large multi-agent systems, the agents may have non-cooperative, cooperative or a mixture of non-cooperative and cooperative interactions depending on the application area. The non-cooperative interactions are generally analyzed with the notion of Nash equilibrium, and the cooperative interactions are usually analyzed through the notion of social optimum. In the former, agents optimize their individual cost while in the latter, they jointly optimize an average cost for the whole population. Such problems have been studied in the framework of game theory and the agents are usually called players. As the number of players increases, exact solution to such games becomes intractable. Mean field approximations provide a framework to find approximate equilibria or social optima in large population games with symmetric and homogeneous players, and the quality of approximation improves as the number of players increases. In this direction, 
\textit{mean field games} and \textit{mean field control problems}\footnote{Mean field control is also called control of McKean-Vlasov dynamics and can also be thought as a setup where a social planner controls the population and prescribes behavior to the players in the society in order to minimize the players' costs.} have been introduced to approximate non-cooperative and cooperative settings respectively; see e.g.,~\cite{lasry2007mean,huang2006large,cdl_mfgVSmfc} and the monographs~\cite{Bensoussan_Book,CarmonaDelarue_book_I}. They have both attracted a growing interest in economics, control theory, applied mathematics, and machine learning communities. In this paper, we explore further the connections between mean field game (MFG for short) and mean field control (MFC for short) and discuss new equilibria notions where the populations have a mixture of cooperative and non-cooperative players. 

In game theory and in applications related to \textit{mechanism design} or \textit{policy making} for a large number of players, it is generally assumed 
that the players in the population optimize their own objectives while taking into account the interactions with other players in a non-cooperative way. This requires finding Nash equilibria. The system is in a Nash equilibrium when there is no player who can be rewarded by modifying her control unilaterally. This makes any Nash equilibrium stable by nature i.e., the players do not have any incentive to deviate from their equilibrium behavior. However, Nash equilibria are also known to lack efficiency in the sense that the players can jointly find a better global (social) outcome. This may occur either if the players are altruistic or if they follow the controls recommended by a social planner, whose objective is to minimize the social cost. A notable illustration of this shortcoming is given by the famous Braess paradox (e.g., \cite{Roughgarden}). Inefficiency of a Nash equilibrium when it is compared to a social optimum is quantified by the Price of Anarchy (PoA), introduced under this name in \cite{KoutsoupiasPapadimitriou} and precisely computed for specific static routing games, see e.g.,~\cite{RoughgardenTardos}. The concept of PoA was extended to (deterministic) differential games in \cite{BasarZhu}, then to mean field game models of linear quadratic type~\cite{carmona2019price} and to models representing congestion in crowd motion~\cite[Section 4.4]{lauriere2021numerical}.
This motivates us to study the ways by which mean field populations of players can evolve from Nash equilibrium to social optimum, and vice versa.

In the first half of this paper, we propose ways to incentivize a mean field population of non-cooperative players, in the spirit of the theory of mechanism design.  We address the following question: \textit{Can we incentivize individual players in an MFG in such a way that they have the same outcome as in the social optimum?} This is achieved without changing the non-cooperative nature of the players i.e., they still find a Nash equilibrium. The incentivization is obtained simply by modifying the costs they incur, which can also be interpreted as a kind of penalization. 
In fact, we address this question from two different angles.
In the first approach, we incentivize players into a Nash equilibrium which has the same equilibrium cost as the one obtained under the initial social optimum. In the second approach, we focus on incentivizing players in a way such that they end up behaving (in terms of their controls and actual states) exactly as if they were adopting the optimal control identified by a social planner optimizing the original social cost.
In both cases, the incentives are designed in such a way that the players remain in a Nash equilibrium. 
Furthermore, in many applications, it is not possible to suddenly perturb the players' costs and one wishes to change the cost using a continuous deformation. This leads us to propose a new type of games which we call \textit{$\lambda$-interpolated mean field games} in which each player's cost is a mixture of individual and social components.  
This is motivated by the fact that Nash equilibria have the desirable property of being stable. More broadly, the question of incentivizing players to behave in a socially optimal way is motivated by the regulation of large systems of players, with applications such as financial systemic risk or carbon emissions. Of particular interest is the possible regulation of \textit{Tragedy of the Commons} type of problems~\cite{hardin1968tragedy}, in which players could exhaust a resource or destroy their environment if everyone behaves in a purely individualistic way, but could preserve a system and have a higher long-term reward if everyone behaves in a cooperative way. In such situations, a key question is to find an incentivization through the cost function of the players to change their behavior to increase the social welfare without changing players' individualistic decision making process.

In the second half of the paper, we explore the instability of social optima under unilateral deviations. While MFG Nash equilibria are stable in the sense that no player would be interested in unilateral deviations, MFC social optima are unstable since any individual player can be better off by deviating unilaterally. To reflect this instability, we introduce the notion of \textit{Price of Instability}, defined as the optimal decrease in the cost for a representative player who deviates from the control prescribed by the social planner. This measures how much a single player can be tempted to deviate. We then consider the case where there is a non-negligible proportion $p$ of rational players deviating and behaving non-cooperatively. We call this problem \textit{$p$-partial mean field games}. We then look at the case where non-cooperative players who are not aware of this proportion of $p$ gradually deviate from the social optimum by repeatedly adjusting their control in a myopic way. We propose a generic deviation process and give two main examples: fixed point algorithm and fictitious play algorithm. Finally, we discuss the connections between the limit of the iterative deviation process and the $p$-partial mean field game.

\subsection{Literature Review and Related Work}

\noindent{\textbf{Mean Field Games and Control and Their Connections.}} Mean field games were introduced to overcome the difficulty of finding a Nash equilibrium in games with large number of players by~\cite{lasry2006jeux,lasry2006jeux2}
and \cite{huang2006large} simultaneously and independently. In this approach, the number of players are taken to be infinite and they assumed to be insignificant, identical and interacting symmetrically. In this way, we can focus on a representative player and her interactions with the population through the population distribution. In the first works of mean field games, forward backward partial differential (Kolmogorov-Fokker-Planck and Hamilton-Jacobi-Bellman) equations were used in order to characterize the Nash equilibrium. Later, probabilistic approaches to characterize both MFG Nash equilibrium and MFC Social Optimum have been introduced. The details can be found in~\cite{CarmonaDelarue_book_I,CarmonaDelarue_book_II}. As introduced previously, the inefficiency of Nash equilibrium is quantified with the PoA notion which is defined in the mean field setup as the ratio of the expected cost of the individual player in the MFG Nash equilibrium to the expected cost of the individual player in the MFC social optimum and by definition, it is greater than or equal to 1. In~\cite{Deori_Margellos_Prandini_2017}, authors show that the Nash equilibrium is efficient (i.e., PoA is equal to 1) in a electric vehicle charging mean field game where the model is a potential game. In~\cite{Li_Zhang_Zhao_2018}, authors show that the solution of an MFG is the same with the solution of a social planner's (modified) optimization problem under certain conditions. This is different than the first part of our paper where we show that the incentivized MFG solution is the same with the social planner's original optimization problem.

\vskip6pt
\noindent{\textbf{Mechanism Design and Incentivization in Large Games.}}
Incentivization in the mean field game setup has been implemented mainly by adding a principal or regulator to the game with her own cost function that is different than the players' cost functions in the mean field population. In~\cite{elie_tale,Carmona_Wang_2021,Aurell_Carmona_Dayanikli_Lauriere_2022}, authors look at problems of a contract theory between a principal and a mean field population on non-cooperative players. In~\cite{hubert2022incentives}, a contract theory problem between a government and a population of fully cooperative players is considered. In~\cite{CDL_2022}, authors look at the Stackelberg mean field problems (both for cooperative and non-cooperative populations) with a motivation to regulate carbon emission levels. In~\cite{GM_2023}, authors discuss a single-level approach to solve Stackelberg mean field games. The main difference of Stackelberg MFG setup (or also contract theory setup) from the incentivization mechanisms explored in this paper is the lack of a principal with her own objectives who is trying to find a Stackelberg equilibrium with the players in the population. Instead in the first part of our paper, we focus on incentivizing the players to have results similar to the social optimum of the original problem.

Instead of having a principal, designing an incentivization mechanism (i.e., taxation for the players) to prevent problems such as tragedy of the commons has been studied in~\cite{Wiszniewska_2001} for a dynamic game with a continuum of nonatomic players without the mean field game formulation.

\vskip6pt

\noindent{\textbf{Mixed Populations in Mean Field Models.}}
In our paper, we introduce two new game models: $\lambda$-interpolated mean field game and $p$-partial mean field game. In $\lambda$-interpolated mean field games, each individual player has a problem which is the mixture of MFG and MFC. A related setup is given in \cite{Angiuli_Detering_Fouque_Lauriere_Lin_2023} where the authors introduce an extension of MFGs in which each player is solving an MFC problem. Such games can be viewed as the limiting situation for a competition between a large number of large coalitions. In~\cite{Barreiro_Gomez_Duncan_Tembine_2020}, the authors introduce co-opetitive linear quadratic mean field games in which players may take into account the other players' cost and rewards in a positive or negative way while making their decisions. More recently, \cite{guo2023mesob} introduced a bi-level optimization problem to balance equilibrium and social optimum. Last, the literature on MFGs also covers multi-population models in which the players of each population are either non-cooperative~\cite{cirant2015multi,achdou2017mean} or cooperative~\cite{djehiche2017mean,barreiro2021mean}. These settings are often referred to as multi-population MFG and mean field type games respectively; see e.g.~\cite{bensoussan2018mean} for a comparison of these settings. However, in these settings, each player is of a given type (cooperative or non-cooperative) and does not change.

\subsection{Contributions and Paper Structure}

In this paper, we investigate the connections between two similar looking but actually very different problems: mean field game and mean field control. Our contributions are both conceptual and theoretical.
First, we introduce ways to incentivize people through their cost functions to respond in the ways social planner prescribes while they are still behaving non-cooperatively and we introduce $\lambda$-interpolated mean field games. We further contribute some theoretical results such as the existence and uniqueness of the $\lambda$-interpolated mean field equilibrium and the continuity of the equilibrium with respect to the parameter $\lambda$. Second, we quantify the instability of social optimum by introducing the \textit{Price of Instability} notion and prove a lower bound on this quantity. Third, we propose a new equilibrium notion, $p$-partial mean field equilibrium, where a $p$ proportion of people deviate from the social planner's prescribed behavior and establish existence and uniqueness results. We further discuss the continuity of the cost functions of deviating and non-deviating players with respect to the parameter $p$ and show that the cost of the deviating players is lower than the original social optimum cost, which leads to the well-known free rider phenomenon.  Finally, we introduce a generic iterative deviation algorithm where players are myopic and give two instances of this algorithm. We prove that the iterative process with the myopic players converges to the complete information scenario, i.e., $p$-partial mean field equilibrium.

In Section~\ref{sec:model}, we introduce our mean field model and describe the mean field game and mean field control (i.e, social planner's optimization) problems. In Section~\ref{sec:MFGtoMFC}, we investigate incentives that make non-cooperative players end up with the social planner's optimization outcomes. This can be thought as moving from the mean field game to mean field control direction. In Section~\ref{subsec:MFGtoMFC_value}, we find the incentives (i.e., perturbations in the cost function of players) to match the value of the players in the (incentivized) MFG to the optimal value of the (original) MFC problem. In Section~\ref{subsec:MFGtoMFC_control}, we find the incentives, again by perturbing the cost function, to change the behavior of the non-cooperative players to behave in the same way as in the social optimum of the original problem.

In Section~\ref{sec:MFCtoMFG}, we are going to focus on 
the instability problem of the social optimum (i.e, MFC solution). In Section~\ref{subsec:MFCtoMFG_completeinfo}, we will look at the case where the players have complete information i.e., players know the proportion of the non-cooperative (i.e., not following the social planner) players. In the complete information setup, we first look at the case where a single player deviates and we introduce the Price of Instability notion in Section~\ref{subsubsec:MFCtoMFG_completeinfo_PoI} and then we look at the case where $p$ proportion of players deviates and we introduce the $p$-partial mean field equilibrium in Section~\ref{subsubsec:MFCtoMFG_completeinfo_pmixed}. Finally, in Section~\ref{subsec:MFCtoMFG_iterative}, we let players to decide in a myopic way by assuming they do not know the proportion of the non-cooperative people and iteratively updating their behavior. We give a generic algorithm for this iterative deviations in Section~\ref{subsubsec:MFCtoMFG_iterative_generic} and give two special cases of this generic iterative deviation algorithm, namely fixed point iterative deviations in Section~\ref{subsubsec:MFCtoMFG_iterative_fixedpoint} and fictitious play iterative deviations in Section~\ref{subsubsec:MFCtoMFG_iterative_fictitious}. We further show that the behavior of the myopic players converges to the behaviors of the players in the corresponding $p$-partial mean field game in Section~\ref{subsubsec:MFCtoMFG_iterative_fixedpoint}.

\section{The Model}

\label{sec:model}

\subsection{Notation}
\label{sec:notation}

\textbf{Basic setting. } We first start with introducing our notation and the model. Let $T > 0$ be a finite time horizon. We assume that the problem is set on a complete filtered probability space $(\Omega, \mathcal{F}, \mathbb{F} = (\mathcal{F}_t)_{t\in[0,T]}, \mathbb{P})$ supporting
a $d$-dimensional Wiener process $W = (W_t)_{t\in[0,T]}$. The representative player has a continuous $\RR^d$-valued state process $X = (X_t)_{t\in[0,T]}$. The space of probability measures on $\mathbb{R}^d$ with a second moment is denoted by $\cP_2(\mathbb{R}^d)$  and endowed with the 2-Wasserstein distance (see e.g.,~\cite[Chapter 5]{CarmonaDelarue_book_I}) defined as:
$$
    W_2(\mu,\mu') = \inf_{\pi \in \Pi(\mu,\mu')} \left( \int_{\RR^d \times \RR^d} |x-x'|^2 \pi(dx,dx')\right)^{1/2},
$$
where $\Pi(\mu,\mu')$ is the set of all probability measures on $\RR^d\times\RR^d$ with a second moment.
The collection of square-integrable and $(\mathcal{F}_t)_{t\in[0,T]}$-progressively measurable action processes $\balpha = (\alpha_t)_{t\in[0,T]}$ where $\alpha_t \in \RR^d$ is denoted by $\mathbb{A}$.
In the sequel, if $X$ is a random variable, we will denote by $\mathcal{L}(X)$ the law of $X$.

\vskip6pt
\noindent\textbf{Model. } In the rest of the paper, we consider the following model. Let $\mu_0$ be an initial distribution. The state process $\bsX^{\balpha}$ for a representative player has the following dynamics:
\begin{equation}
\label{fo:state}
\begin{split}
    dX_t^{\balpha}&=\alpha_t dt + \sigma dW_t, \quad t \in [0,T]
    \\
    X_0^{\balpha}& \sim \mu_0,
\end{split}
\end{equation}
where $\sigma \in \RR$ is a non-zero constant. For each fixed flow of probability measures $\bmu=(\mu_t)_{0\le t\le T}$, the representative player has the cost:
\begin{equation}
\label{fo:J_of_alpha}
    J^{\bmu}(\balpha) = J(\balpha; \bmu)=\mathbb{E}\Bigl[\int_0^Tf(t,X_t^{\balpha},\mu_t,\alpha_t)dt + g(X_T^{\balpha},\mu_T)\Bigr],
\end{equation}
where $f : [0,T]\times \RR^d \times \cP_2(\RR^d) \times \RR^d \to \RR$ is the running cost of the representative player that depends on her state $X_t^{\balpha}$, the population distribution $\mu_t$ of states  and her control $\alpha_t$, and $g : \RR^d \times \cP_2(\RR^d) \to \RR$ is the terminal cost of the representative player that depends on her terminal state $X_T^{\balpha}$ and the population distribution $\mu_T$ of the states at the terminal time. When the control $\balpha$ is clear from the context, we will omit the superscript on $\bsX$. 

Although many of the ideas developed in the rest of the paper could be extended to more complex models, for the sake of simplicity, in the sequel we consider that the running cost function $f$ is of the form:
\begin{equation}
\label{fo:running_cost}
f(t,x,\mu,\alpha)=\frac12|\alpha|^2 + f_0(x,\mu)
\end{equation}
where $|\cdot|$ denotes the Euclidean norm. 

The assumptions on $f_0$ are stated below in Assumption~\ref{assumption:initial-model}. In particular,  
we will assume $f_0$ to be differentiable  with respect to $x$ and differentiable with respect to $\mu$ for both notions of derivatives
that we recall at the end of this subsection, namely, in the Lions sense and also in the functional sense (on the entire vector space of measures
and we then restrict 
the derivative to the space of probability measures). 
This avoids any ambiguity in the definition of the derivative when the measure argument is just taken in the space of probability measures (in which case the derivative is well defined up to an additive constant only).

The (reduced) Hamiltonian $H: [0,T] \times \RR^d \times \cP_2(\RR^d) \times \RR^d \times \RR^d  \to \RR$ is defined by:
\begin{equation}
\label{eq:reduced-H-def}
    H(t,x,\mu, y, \alpha) = \alpha \cdot y +\frac{1}{2}|\alpha|^2 + f_0 (x,\mu).
\end{equation}
We will denote by $\hat\alpha: [0,T] \times \RR^d \times \cP_2(\RR^d) \times \RR^d \to \RR^d$ its minimizer, which is:
\begin{equation}
\label{eq:hat-alpha-def}
    \hat\alpha(t,x,\mu,y) \in \argmin_{\alpha} H(t, x, \mu, y, \alpha) = -y.
\end{equation}

\begin{remark}
    We emphasize that, in the sequel, when we write $f$, $H$ or $\hat\alpha$, we always mean the functions defined in~\eqref{fo:running_cost}, \eqref{eq:reduced-H-def} and~\eqref{eq:hat-alpha-def}. Furthermore, in this model the drift is simply the control. Several of the results could easily be extended beyond this special setting, but we refrain to do so for the sake of simplicity of presentation. 
\end{remark}

\noindent\textbf{Differentiation with respect to a measure argument. } 
We will consider two notions of derivatives with respect to measures. The first notion is the Lions derivative, which we denote by $\partial_\mu$. We recall here the definition and refer to e.g.,~\cite{CardaliaguetDelarueLasryLions} or~\cite[Chapter 5]{CarmonaDelarue_book_I} for more details. A function $U: \cP_2(\RR^d) \to \RR$ is differentiable if there exists a map $\partial_\mu U: \cP_2(\RR^d) \times \RR^d \to \RR$ such that for any $\mu,\mu' \in \cP_2(\RR^d)$, 
$$
    \lim_{s \to 0^+} \frac{U((1-s)m + sm') - U(m)}{s} = \int_{\RR^d} \partial_\mu U(\mu, x') d(\mu' - \mu)(x').
$$
We say that $U$ is $\mathcal{C}^1$ if $\partial_\mu U$ is continuous. 
If $X$ is an $\RR^d$ random variable and $\varphi: \RR^d \times \RR^d \to \RR$, then the notation:
$
    \tilde{\EE}[\varphi(X, \tilde X)]
$
means that the expectation is taken (only) over $\tilde{X}$, which is an independent copy of $X$.

The second notion is the flat or functional derivative. We borrow the definition from~\cite[Definition 5.43]{CarmonaDelarue_book_I}. A function $U: \cP_2(\RR^d) \to \RR$ is said to have a functional (or flat or linear) derivative if there exists a function 
$$
    \frac{\delta U}{\delta m}: \cP_2(\RR^d) \times \RR^d \ni (\mu, x) \mapsto \frac{\delta U}{\delta m}(\mu)(x) \in \RR,
$$
continuous for the product topology, such that for any subset $\cK \subseteq \cP_2(\RR^d)$, the function $\RR^d \ni x \mapsto [\delta U/ \delta m](\mu)(x)$ is at most of quadratic growth in $x$ uniformly in $\mu \in \cK$, and such that for all $\mu$ and $\mu'$ in $\cP_2(\RR^d)$, it holds: 
$$
    U(\mu') - U(\mu) = 
    \int_0^1 \int_{\RR^d} \frac{\delta U}{\delta m}(t\mu' + (1-t)\mu)(x) d[\mu' - \mu](x) dt.
$$
Notice that $\frac{\delta U}{\delta m}$ is not defined uniquely but it is unique up to an additive constant. For the sake of definiteness, unless otherwise specified, we will assume that the following normalization condition is satisfied: $\int\frac{\delta U}{\delta m}(\mu)(x) \mu(dx) = 0$. This selects a unique function for $\frac{\delta U}{\delta m}$. However the choice of this constant ($0$ or another value) is not going to impact our analysis. 
The two notions of derivatives are related by the formula:
\begin{equation}
    \label{eq:connection-partialmu-deltam}
    \partial_x\frac{\delta U}{\delta m}(\mu)(x)=\partial_\mu U(\mu)(x), \qquad (\mu, x) \in \cP_2(\RR^d) \times \RR^d.
\end{equation}

We refer to~\cite[Chapter 5]{CarmonaDelarue_book_I}  and~\cite{CardaliaguetDelarueLasryLions}
for more details on the connection between the two notions of derivatives.

We introduce several sets of assumptions on $f_0$ and $g$. 
The first one is

\begin{assumption}
\label{assumption:initial-model}
{ \ }

\begin{enumerate}[(i)]
    \item The functions $f_0$ and $g$ are continuously differentiable with respect to $x$ and differentiable with respect to $\mu$ (in the sense of $\partial_\mu$). Furthermore, we assume that, for any $(x,\mu) \in {\mathbb R}^d \times {\mathcal P}_2({\mathbb R}^d)$, there exists a version of 
    $v \mapsto \partial_\mu f_0(x,\mu)(v)$ (resp. 
    $v \mapsto \partial_\mu g(x,\mu)(v)$) such that the mapping $(x,\mu,v) \mapsto \partial_\mu f_0(x,\mu)(v)$
    (resp. $(x,\mu,v) \mapsto \partial_\mu g(x,\mu)(v)$)
    is continuous. 
    \item The derivatives  $\partial_x f_0$ are $\partial_x g$ are Lipschitz continuous
    (the space ${\mathcal P}_2({\mathbb R}^d)$ being equipped with $W_2$)
    and the derivatives $\partial_\mu f_0$ and $\partial_\mu g$
    are Lipschitz continuous, in the following sense. 
    For all $x, x^\prime \in \RR^d,\ \alpha, \alpha^\prime \in \RR^d$, $\mu, \mu^\prime \in \cP_2 (\RR^d)$, and any $\RR^d$-valued random variables $X$ and $X^\prime$ having $\mu$ and $\mu^\prime$ as distributions respectively, we have:
    \begin{equation*}
        \begin{aligned}
            &\EE\Bigl[ \Big| \partial_\mu f_0 (x^\prime, \mu^\prime)(X^\prime) - \partial_\mu f_0 (x, \mu)(X)\Big|^2 \Bigr] \leq L\Big[|x^\prime - x|^2 + \EE \Big[|X^\prime - X|^2\Big]\Big]\\
            &\EE \Bigl[ \Big| \partial_\mu g (x^\prime, \mu^\prime)(X^\prime) - \partial_\mu g (x, \mu)(X)\Big|^2 \Bigr] \leq L\Big[|x^\prime - x|^2 + \EE \Big[|X^\prime - X|^2\Big]\Big]
        \end{aligned}
        \end{equation*}
 \end{enumerate}
 \end{assumption}

The following remarks are in order. 
\begin{enumerate}[(a)]
\item Following~\cite[Proposition 5.51]{CarmonaDelarue_book_I}, 
the joint continuity assumption 
in the item $i$ of~Assumption~\ref{assumption:initial-model} says that the functions also admit linear functional derivatives $\delta f_0/\delta m$ and 
$\delta g/\delta m$
and that 
the relationship \eqref{eq:connection-partialmu-deltam}
holds. 
\item
Thanks to the Lipschitz property of the functions, 
all the functions $f_0$, $\partial_x f_0$, $g$ and $\partial_x g$ are locally bounded.
In particular, 
there exists a constant $L$ (possibly different from the one in Asssumption \ref{assumption:initial-model}) such that, for any $R\geq 0$ and any $(x,\mu)$ such that $|x|\leq R$ and $M_2(\mu) \leq R$, $|\partial_x f_0(x, \mu)|,\ |\partial_x g(x, \mu)|$ are bounded by $L(1+R)$ where $M_2(\mu)^2$ is the second moment of a measure $\mu \in \cP_2(\RR^d)$, i.e.,  $M_2(\mu)^2 = \int_{\RR^d} |x|^2 d\mu(x)$. Furthermore, the $L^2(\RR^d, \mu; \RR^d)$ norms of $x^{\prime} \mapsto \partial_\mu f_0 (x, \mu)(x^{\prime})$, $x^{\prime} \mapsto \partial_\mu g (x, \mu)(x^{\prime})$ are bounded by $L(1+R)$. 
\item In particular, $f_0$ and $g$ are locally Lipschitz continuous. Namely,
for all $x, x^\prime \in \RR^d$, and $\mu, \mu^\prime \in \cP_2 (\RR^d)$, we have:
        \begin{equation*}
        \begin{aligned}
            &\Big| f_0 (x^\prime, \mu^\prime) - f_0 (x, \mu)\Big| + \Big|g (x^\prime, \mu^\prime) - g (x, \mu)\Big| \\
            &\hskip5mm \leq L\Big[1+|x^\prime| + |x| + M_2(\mu^\prime) + M_2(\mu)\Big]\times \Big[|x^\prime- x| + W_2(\mu^\prime, \mu)\Big].
        \end{aligned}
        \end{equation*}  
        \item Assumption 
        \ref{assumption:initial-model}
        subsumes assumption  {\bf Necessary SMP} (A1)-(A2) (p. 166)
        and 
        assumption {\bf (Pontryagin Optimality)} (A1)-(A2) (pp. 542-543)
        in \cite{CarmonaDelarue_book_I}, when adapted to the model~\eqref{fo:state}--\eqref{fo:running_cost}. 
        The first one provides a convenient form of the Pontryagin principle for MFGs and the second one 
        for MFCs.  
\end{enumerate}

Throughout the article,  Assumption \ref{assumption:initial-model}
is in force. 
We will sometimes complement it with one or other (or both) of the following 
hypotheses:
\begin{assumption}
\label{assumption:SMP:MFG}
At least one of the following two properties holds true:  
\begin{enumerate}[(i)]
\item The matrix $\sigma$ is non-degenerate and, 
for any $R>0$, 
the derivatives 
$\partial_x f_0$ are $\partial_x g$ are bounded in the $x$-argument, uniformly 
with respect to the entries $\mu$ satisfying $M_2(\mu) \leq R$.
\item
The functions $f_0$ and $g$ are convex in the argument $x$, and the functions 
$x \in {\mathbb R}^d \mapsto x \cdot \partial_x f_0(0,\delta_x)$ and 
$x \in {\mathbb R}^d \mapsto x \cdot \partial_x f_0(0,\delta_x)$
are greater than $-C(1+\vert x \vert)$ for some $C>0$. 
\end{enumerate}
\end{assumption}

Assumption~\ref{assumption:SMP:MFG}
is used next to guarantee that the 
stochastic Pontryagin principle, which we use repeatedly in the sequel, provides a sufficient condition of optimality, and then to derive the existence of an equilibrium to the MFG under study.
 
\begin{assumption}
\label{assumption:existence:uniqueness:MFC}
The functions $f_0$ and $g$ are convex in $(x,\mu)$, convexity 
with respect to the measure argument 
being understood in the displacement convex sense, namely
\begin{equation*} 
f_0 \Bigl( \lambda x + (1-\lambda) x' , {\mathcal L} \bigl( 
\lambda X + (1-\lambda) X'\bigr) 
\Bigr)
\leq \lambda 
f_0 \bigl(x , {\mathcal L} (X) \bigr) 
+ (1-\lambda) 
f_0 \bigl( x' , {\mathcal L}(  X') \bigr) 
\bigr), 
\end{equation*}
for any $\lambda \in [0,1]$, $x,x' \in {\mathbb R}^d$ and 
${\mathbb R}^d$-valued square integrable random 
variables $X$ and $X'$,
with respective laws ${\mathcal L}(X)$ and ${\mathcal L}(X')$, 
and similarly for $g$. 
\end{assumption}
In combination with 
Assumption 
\ref{assumption:initial-model}, Assumption 
\ref{assumption:existence:uniqueness:MFC} makes it possible to apply 
 Theorem~6.16 (p. 550) and Theorem 6.19 (p. 559) in \cite{CarmonaDelarue_book_I}, which guarantee existence and uniqueness of a minimizer 
 to the MFC considered in this article. We will come back to this point next.

\subsection{The Mean Field Game Problem}
\label{subsec:model_mfg}

We start with the standard definition of the MFG problem.

\begin{definition}
An MFG solution, also called an MFG (Nash) equilibrium, is a pair $(\balpha^{\rMFG},\bmu^{\rMFG})$ such that $\balpha^{\rMFG}$ minimizes $J^{\boldsymbol \mu^{\rMFG}}(\cdot)$ defined in~\eqref{fo:J_of_alpha}, and $\mu^{\rMFG}_t = \cL(X^{\rMFG}_t)$ for all $t \in [0,T]$, where $\bX^{\rMFG}$ is the solution to~\eqref{fo:state} controlled by $\balpha^{\rMFG}$. 
\end{definition}

  Solving an MFG thus amounts to solving a \emph{fixed point} problem. Below, we often use the stochastic Pontryagin principle to identify the MFG Nash equilibrium. Even though Pontryagin principle just provides in general a necessary but not sufficient condition of optimality, it is known that sufficiency holds under Assumptions~\ref{assumption:initial-model} and~\ref{assumption:SMP:MFG}.
  Indeed, 
  for a given flow $(\mu_t)_{t \in [0,T]}$ with values in 
  ${\mathcal P}_2({\mathbb R}^d)$, the 
  control problem 
  \eqref{fo:state}--\eqref{fo:J_of_alpha}--\eqref{fo:running_cost}
has at least one minimizer. This follows for instance from the same result as in the one we invoke below for MFC, 
see 
\cite[Theorems 2.2 \& 2.3]{lacker_2017} (but it is clear that the result was known before the publication 
of the latter work, see the bibliography
cited in \cite{lacker_2017}).
In turn, the stochastic Pontryagin principle provides a necessary condition for 
the minimizer, see for instance 
\cite[Theorem 3.27]{CarmonaDelarue_book_I}. 
In fact, under 
  Assumptions~\ref{assumption:initial-model} and~\ref{assumption:SMP:MFG}, 
the Pontryagin system
(whose form is obtained by replacing 
$\mathcal{L}(X_t)$ by $\mu_t$  in the 
system  \eqref{fo:mfg_FBSDE} below) 
is uniquely solvable
and thus provides a characterization of the minimizer: 
under item $(i)$ in Assumption~\ref{assumption:SMP:MFG}, which follows from 
\cite{delarue2002existence}; under 
$(ii)$ which follows from 
\cite[Theorem 3.17]{CarmonaDelarue_book_I}.
As a consequence, we deduce that, for the model given in Section~\ref{sec:model} with~\eqref{fo:running_cost}, the MFG equilibrium 
is characterized by the solution, denoted by $(\bsX^{\rMFG}, \bsY^{\rMFG}, \bsZ^{\rMFG})$, of the following FBSDE of the McKean-Vlasov (MKV) type (see e.g.,~\cite[Chapter 3.3.2]{CarmonaDelarue_book_I}):
\begin{equation}
\label{fo:mfg_FBSDE}
\begin{cases}
dX_t&= -Y_tdt +\sigma dW_t\\
dY_t&=-\partial_x f_0(X_t,\mathcal{L}(X_t)) dt + Z_t dW_t,\\
Y_T&=\partial_x g(X_T,\mathcal{L}(X_T)),
\end{cases}
\end{equation}
where the drift of $\bX$ is the optimal control, see~\eqref{eq:hat-alpha-def} evaluated at $y=Y_t$. 
When the population is at equilibrium, her equilibrium cost is obtained by using the equilibrium control as well and yields the following cost: 
\begin{equation}
\label{fo:J_MFG_of_alpha}
    J^{\mu^{\rm{MFG}}}(\balpha^{\rMFG})
    = \mathbb{E}\Bigl[\int_0^Tf(t,X^{\rMFG}_t,\mu^{\rMFG}_t,\alpha^{\rMFG}_t)dt + g(X^{\rMFG}_T,\mu^{\rMFG}_T)\Bigr]
\end{equation}
with $\mu^{\rMFG}_t=\mathcal{L}(X^{\rMFG}_t)$ with $X^{\rMFG}_t = X^{\balpha^{\rMFG}}_t$ for each $t\ge 0$.

Notice that, under item 
$(i)$ in Assumption~\ref{assumption:SMP:MFG}, existence of an MFG equilibrium follows from 
\cite[Theorem 4.32]{CarmonaDelarue_book_I}. 
Under item $(ii)$, it follows from 
\cite[Theorem 4.53]{CarmonaDelarue_book_I}.

\subsection{The Mean Field Control Problem}
\label{subsec:model_mfc}
As we will compare next the two problems, we now present MFC problem, also referred to as social planner's problem or McKean-Vlasov (MKV) control problem. 

\begin{definition}
\label{def:MFCpb}
An MFC solution also called mean field control optimum is a minimizer of the following cost, defined for each control $\boldsymbol \alpha=(\alpha_t)_{0\le t\le T}$ as:
\begin{equation}
\label{fo:J_MKV_of_alpha}
    J^{\rMKV}(\balpha)=\mathbb{E}\Bigl[\int_0^Tf(t,X_t^{\balpha},\mathcal{L}(X_t^{\balpha}),\alpha_t)dt + g(X_T^{\balpha},\mathcal{L}(X_T^{\balpha}) )\Bigr]
\end{equation}
where $\boldsymbol X^{\balpha}$ is determined by \eqref{fo:state} using the control $\balpha$.
\end{definition}
An MFC problem is thus an \emph{optimization} problem. Differently from the MFG, in the MFC problem players can be thought as playing cooperatively or being managed by a social planner to minimize the expected social cost. 
In the MFG setting, when the representative player changes her behavior, the population is assumed to stay the same since the representative player has an infinitesimal influence. However, in the MFC setting, we assume all the agents behave similarly (i.e., have similar statistical distributions) and when the representative player changes their control, everyone uses the new control; therefore, the population distribution also changes. We define the social cost (per individual player) as the quantity:
\begin{equation}
\label{fo:J_MKV}
    J^{*}:=\inf_{\balpha}J^{\rMKV}(\balpha)
\end{equation}
Notice that:
\begin{equation}
\label{fo:PoA}
    J^{*}\le J^{\bmu^{\rm{MFG}}}(\balpha^{\rMFG})
\end{equation}
for all MFG equilibria $\boldsymbol \alpha^{\rMFG}=(\alpha^{\rMFG}_t)_{0\le t\le T}$. In the literature, the relationship between the optimal MFC cost and the equilibrium MFG cost has been studied through the notions of PoA~\cite{graber2016linear,carmona2019price} and (in)efficiency~\cite{cardaliaguet2019efficiency}. For concrete examples and numerical illustrations we refer to~\cite{carmona2019price} and \cite[Section 2.6 ]{lauriere2021numerical} in linear-quadratic settings, \cite[Section 4.4 ]{lauriere2021numerical} in a crowd-motion example, \cite{elie2020contact,doncel2022mean} in SIR models, and \cite{narasimha2019mean} in wireless networks, to cite just a few.

\vskip 6pt
Under 
Assumption 
\ref{assumption:initial-model}, 
the MFC problem has an optimal control, which we denote by $\balpha^{\rMKV}$, i.e., 
\begin{equation}
\label{fo:alpha_MKV}
    \balpha^{\rMKV} \in \argmin_{\boldsymbol \alpha}J^{\rMKV}(\boldsymbol \alpha)
\end{equation} 
This follows from 
\cite[Theorems 2.2 \& 2.3]{lacker_2017}. 
The reader can easily check that Assumptions $A$ and $C$ in the latter reference are satisfied 
under Assumption~\ref{assumption:initial-model}. In particular, the convexity condition stated in 
Assumption $C$ follows from the fact the drift in 
\eqref{fo:state} is linear in $\alpha$
and the cost 
in 
\eqref{fo:running_cost}
is convex in $\alpha$. 
Theorem 6.14 in 
\cite{CarmonaDelarue_book_I} can be used to infer the form 
of the Pontryagin principle in this situation. 
Uniqueness of the minimizer is known to hold true under the 
additional convexity condition 
stated in Assumption \ref{assumption:existence:uniqueness:MFC}, in which case 
the symbol `$\in$' in 
\eqref{fo:alpha_MKV} can be replaced by `$=$', i.e., 
\begin{equation*}
    \balpha^{\rMKV} := \argmin_{\boldsymbol \alpha}J^{\rMKV}(\boldsymbol \alpha).
\end{equation*}
Accordingly, we denote by $\bmu^{\rMKV}$ the corresponding flow of marginal distributions of the optimally controlled state process:
\begin{equation}
\label{fo:mu_MKV}
\bmu^{\rMKV}=(\mu^{\rMKV}_t)_{0\le t\le T}\quad\text{with}\quad \mu^{\rMKV}_t=\cL(X^{\rMKV}_t)\quad\text{and}\quad dX^{\rMKV}_t=\alpha^{\rMKV}_t dt + \sigma dW_t.
\end{equation}
The optimal cost defined in~\eqref{fo:J_MKV} can then be expressed using the notation~\eqref{fo:J_of_alpha} as:
\[
    J^{*} := J^{\bmu^{\rMKV}}(\balpha^{\rMKV}).
\]
By Pontryagin's principle and because the (reduced) Hamiltonian is~\eqref{eq:reduced-H-def},  
we can regard the process 
$(Y_t^{\rMKV} := -\alpha^{\rMKV}_t)_{t \in [0,T]}$ as an adjoint process solving the Backward Stochastic Differential Equation (BSDE):
\begin{equation}
\label{fo:MKV_adjoint}
dY_t=- \Bigl(\partial_x f_0(X_t,\cL(X_t)) dt + \tilde\EE[\partial_\mu f_0(\tilde X_t,\cL(X_t))(X_t)] \Bigr) dt+ Z_t dW_t,
\end{equation}
with the terminal condition
\begin{equation}
\label{fo:MKV_terminal}
Y_T=\partial_x g(X_T,\cL(X_T)) + \tilde\EE[\partial_\mu g(\tilde X_T,\cL(X_T))(X_T)]
\end{equation}
where we recall that we use the notation $\partial_\mu$ for the Lions derivative and $\,\tilde\,$ for independent copies (see Section~\ref{sec:notation}). In summary, using an optimal control for the MFC problem, the processes $\bsX$ and $\bsY$ solve the following FBSDE system:
\begin{equation}
\label{fo:mkv_FBSDE}
\begin{cases}
dX_t&= -Y_tdt +\sigma dW_t\\
dY_t&=-\Bigl(\partial_x f_0(X_t,\cL(X_t)) dt + \tilde\EE[\partial_\mu f_0(\tilde X_t,\cL(X_t))(X_t)] \Bigr) dt+ Z_t dW_t,\\
Y_T&=\partial_x g(X_T,\cL(X_T)) + \tilde\EE[\partial_\mu g(\tilde X_T,\cL(X_T))(X_T)],
\end{cases}
\end{equation}
which is different from the FBSDE in \eqref{fo:mfg_FBSDE} that characterizes the MFG equilibrium. To stress that the two solutions are different, we will denote by $(\bsX^{\rMKV}, \bsY^{\rMKV}, \bsZ^{\rMKV})$ the solution to the above system.

\begin{remark}
\label{rem:convex}
As in the MFG case, Pontryagin principle states a necessary condition for the optimality and for sufficiency, further assumptions should be satisfied, which is guaranteed by Assumption~\ref{assumption:existence:uniqueness:MFC}. 
Under this condition, 
the system 
\eqref{fo:mkv_FBSDE}
is indeed uniquely solvable, 
see
\cite[Theorems 6.16 and 6.19]{CarmonaDelarue_book_I} and hence characterizes the (unique) MFC equilibrium. 
For instance, so is the case 
 if $f_0$ and $g$
is in the form: 
\begin{equation}
\label{eq:convex:case}
    f_0(x,\mu) = \int_{{\mathbb R}^d} \varphi_0(x-y) d\mu(y), 
    \quad 
    g(x,\mu) = \int_{{\mathbb R}^d} \psi(x-y) d\mu(y),
\end{equation}
where $\varphi_0$ and $\psi$ in the right-hand side are convex functions on ${\mathbb R}^d$ which are bounded below and with Lipschitz derivatives (in the $x$-argument).
\end{remark}

\section{From MFG to MFC: Incentivization to Reach Social Optimum}
\label{sec:MFGtoMFC}

In this section, we are going to focus on the question of incentivizing non-cooperative players to behave in a way that yields results similar to what they would get if they behaved cooperatively. This can be viewed as a kind of \textit{mechanism design} question. In fact there are two aspects to this question: \textit{matching the optimal value} and \textit{matching the optimal control}. 
In other words, we will modify the cost function of the players in a way that, at equilibrium in the new MFG, they have the same value or the same control as in the original MFC. 
We start by discussing how we can change the problem in a way that the non-cooperative players will end up with the same cost levels of the original problem if they behaved cooperatively. In other words, we will change the cost function of the players in a way that, while playing in a non-cooperative way, they end up with an equilibrium cost that is equal to the MFC cost in the original problem. 
Secondly, we discuss how we can change the problem in a way that the non-cooperative players will behave in the same way (i.e., will have the same controls) as if they were cooperative. 

\subsection{Matching the Mean Field Control Optimal Value} 
\label{subsec:MFGtoMFC_value}

In this subsection, we first aim at incentivizing the players (by modification of their cost functions) into a behavior which leads to the same equilibrium cost as the one obtained previously under the rule of the social planner. In order to do so, we work with the value functions of the MFC problem and of the MFG problem. 

\vskip 6pt
We first recall the definition of the value function of the MFC problem, see Definition~\ref{def:MFCpb} and \eqref{fo:J_MKV}. Let $\bsX^{{\rMKV}, t, \mu}$ be the solution of the following optimization problem:
for any $t \in [0,T]$ and $\xi \in L^2(\Omega,\cF_{t},\PP;\RR^d)$ if we denote by $\mu=\cL(\xi)$ the law of $\xi$, we define the value function $v(t,\mu)$ as
the quantity 
\begin{equation}
\label{fo:value_MKV}
v(t,\mu)=\inf_{\balpha_{\vert [t,T]}} \EE\bigg[
\int_t^{T} f\bigl(s,X_s^{\balpha},
{\mathcal L}(X_{s}^{\balpha}),\alpha_{s}\bigr) ds
+g\bigl(X_{T}^{\balpha},{\mathcal L}(X_{T}^{\balpha} ) \bigr)
\bigg],
\end{equation}
where the state process  $\bsX^{\balpha, t, \mu}$ satisfies the state dynamics \eqref{fo:state} on $[t,T]$ with initial condition $X_{t}^{\balpha} = \xi$.  We recall that $f$ is defined in~\eqref{fo:running_cost} to simplify the presentation.
Here, the infimum is taken over $\RR^d$-valued 
square-integrable 
$(\cF_{s})_{s \in [t,T]}$-progressively measurable processes 
$\balpha_{\vert [t,T]} =(\alpha_{s})_{ s \in [t,T]}$. It turns out that this infimum only depends upon $t$ and the distribution of $\xi$ and we denote it by $v(t,\mu)$ where $\mu = \cL(\xi)$, for existence results see e.g.,~\cite[Theorems 2.2 and 2.3]{lacker_2017}. 
Next, we define the extended value function $(t,x,\mu)\mapsto V(t,x,\mu)$ by conditioning on the initial state of the controlled diffusion, namely, 
for each $(t,x,\mu)\in [0,T]\times \RR^d\times {\mathcal P}_{2}(\RR^d)$, 
\begin{equation}
\label{fo:value_function_MKV}
\begin{split}
    V(t,x,\mu) &= \EE\bigg[
\int_t^T f\bigl(s,X_s^{{\rMKV}, t, \mu},\cL(X_s^{{\rMKV}, t, \mu}),{\alpha}^{{\rMKV}, t, \mu}_s\bigr)ds 
    \\
    &\qquad\qquad+g\bigl(X_T^{{\rMKV}, t, \mu},\cL(X_T^{{\rMKV}, t, \mu})\bigr)\,
\big| \, X_{t}^{{\rMKV}, t, \mu} = x\bigg],
\end{split}
\end{equation}
where $\balpha^{{\rMKV}, t, \mu}$ is the minimizer (which is assumed to be unique under Assumption~\ref{assumption:existence:uniqueness:MFC}) in \eqref{fo:value_MKV} under the constraint $\cL(X^{{\rMKV}, t, \mu}_t)=\mu$. The extended value function $V$ satisfies
\begin{equation}
\label{eq:connection:v:V}
v(t,\mu) = \int_{\RR^d} V(t,x,\mu) d\mu(x). 
\end{equation}

Our goal is to identify $V$, when it is smooth, as the value function of an MFG with modified cost functions. As we just highlighted, this implicitly requires the minimizer to be unique, as multiplicity of the minimizers are currently associated with singularities of the value function. In order to proceed with the identification of $V$, we introduce the following modified cost function:
\begin{equation}
\label{fo:f_tilde}
    \tilde f(t,x,\mu,\alpha)=\frac12 |\alpha|^2+\tilde f_0(t,x,\mu),
\end{equation}
where, using a dot $\cdot$ to denote the inner product in $\RR^d$,
\begin{equation}
\label{fo:f_0_tilde}
\begin{split}
    &\tilde f_0(t,x,\mu)= f_0(x,\mu)
    +\frac12 \int_{\RR^d}\int_{\RR^d} \partial_{\mu} V(t,x',\mu)(x)\cdot\partial_{\mu} V(t,x'',\mu)(x) d\mu(x')d\mu(x'') \\
    &\hspace{55pt}
    -\int_{\RR^d} \int_{\RR^d} \partial_{\mu} V(t,x',\mu)(\tilde x)
    \cdot\partial_{\mu} V(t,x,\mu )(\tilde x)  d\mu(x')d\mu(\tilde x).
\end{split}
\end{equation}

\begin{proposition}
\label{prop:incentive-value}
Assume that the derivatives of the function $(x,\mu) \mapsto V(t,x,\mu)$ with respect to $x$ and $\mu$, and the derivative of $x' \mapsto \partial_\mu V(t,x,\mu)(x')$ with respect to $x'$ are continuous and bounded. Assume also that $\partial_x V$ is Lipschitz continuous in 
the spatial and measure argument (w.r.t. the $2$-Wasserstein distance). Then, for any initial distribution, 
the MFG with costs $(\tilde f, g)$ has an equilibrium control 
${\balpha}^{\rMFG}$ that shares the same 
value function \eqref{fo:value_function_MKV} as the minimizer of the MFC problem with costs $(f,g)$, namely, for any 
$x \in {\mathbb R}^d$,
\begin{equation*}
 \EE\bigg[
\int_t^T f\bigl(s,X_s^{\rMFG},\cL(X_s^{\rMFG}),{\alpha}^{\rMFG}_s\bigr)ds 
+g\bigl(X_T^{\rMFG},\cL(X_T^{\rMFG})\bigr)\,
\big| \, X_{t}^{\rMFG} = x\bigg]
= 
V\bigl(t,x,\cL(X_t^{{\balpha}^*})\bigr). 
\end{equation*}
\end{proposition}

For further discussion on the smoothness of $V$, please refer to Appendix~\ref{app:MFGtoMFC_value}.

\begin{proof}[Proof of Proposition~\ref{prop:incentive-value}]
    Remember that the extended value function for MFC problem is given by \eqref{fo:value_function_MKV}. In fact, it turns out that the MFC optimizer is given as $\alpha^{{\rMKV}}_s=\bar\alpha(s,X^{\hat\balpha}_s,\cL(X^{\hat\balpha}_s))$ for some deterministic function $(t,x,\mu) \mapsto \bar\alpha(t,x,\mu)$ which is given by: 
\begin{equation}
\label{eq:def-bar-alpha-opt-orig-MFC}
\begin{aligned}
    \bar{\alpha}(t,x,\mu)&=
    \hat{\alpha}\Bigl(t,x,\mu,\partial_{x} V(t,x,\mu)
    + \int_{\RR^d} \partial_{\mu} V(t,x',\mu)(x) d\mu(x')\Bigr)\\
    &=
    -\partial_{x} V(t,x,\mu)
    - \int_{\RR^d} \partial_{\mu} V(t,x',\mu)(x) d\mu(x'),
\end{aligned}
\end{equation}
where we recall that $\hat\alpha$ is defined in~\eqref{eq:hat-alpha-def}, and where the extended value function solves the following master equation: 
\begin{equation}
\label{fo:mkv_master_PDE}
\begin{split}
    &\partial_{t}  V(t,x,\mu) 
    + 
    \bar\alpha(t,x,\mu)  \partial_{x} V(t,x,\mu)  
    \\
    &\hspace{15pt}
    + \frac{\sigma^2}2 \Delta V(t,x,\mu)  +\frac12 |\bar\alpha(t,x,\mu)\bigr)|^2+f_0(x,\mu)
    \\
    &\hspace{15pt}
    +\int_{\RR^d} \biggl[ \bar\alpha(t,x',\mu) \partial_{\mu} V(t,x,\mu )(x') 
    +
    \frac{\sigma^2}{2}{\rm trace} \Bigl(
    \partial_{x'}  \partial_{\mu} V(t,x,\mu)(x')
    \Bigr) \biggr]d\mu(x') = 0,
\end{split}
\end{equation}
for $(t,x,\mu) \in [0,T] \times \RR^d \times \cP_{2}(\RR^d)$, 
with the terminal condition $V(T,x,\mu) = g(x,\mu)$. 
The interested reader can refer to~\cite[Chapter 6]{CarmonaDelarue_book_I} and~\cite{cdll_book} for further discussion.

Plugging this relationship into 
\eqref{fo:mkv_master_PDE}, we obtain the full-fledged form of the master equation:
{\begin{equation}
\label{fo:master_planner}
\begin{split}
&\partial_{t}  V(t,x,\mu) + \frac{\sigma^2}2 \Delta V(t,x,\mu) - \frac12 |\partial_x V(t,x,\mu) |^2 
\\
&\hspace{5pt}
+\frac12 \int_{\RR^d}\int_{\RR^d} \partial_{\mu} V(t,x',\mu)(x)\partial_{\mu} V(t,x'',\mu)(x) d\mu(x')d\mu(x'') 
+f_0(x,\mu)
\\
&\hspace{5pt}
+\int_{\RR^d} \biggl[\biggl(
-\partial_x V(t,\tilde x,\mu) 
- \int_{\RR^d} \partial_{\mu} V(t,x',\mu)(\tilde x) d\mu(x')
\biggr) \cdot\partial_{\mu} V(t,x,\mu )(\tilde x) 
\\
&\hspace{65pt}
+\frac{\sigma^2}{2}{\rm trace} \Bigl(\partial_{\tilde{x}} \partial_{\mu} V(t,x,\mu)(\tilde x)
\Bigr) \biggr]d\mu(\tilde x) = 0,
\end{split}
\end{equation}
}
for $(t,x,\mu) \in [0,T] \times \RR^d \times \cP_{2}(\RR^d)$, 
with the terminal condition $V(T,x,\mu) = g(x,\mu)$. 

\vskip 6pt
Now, our goal is to identify $V$ as the value function of an MFG with perturbed cost functions. At this stage, we do just that by identifying the above master equation as the master equation of an MFG with a different cost function. In order to do so, we use the fact that the master equation of an MFG with controlled state equation \eqref{fo:state} and running cost function $\tilde f$ defined in~\eqref{fo:f_tilde} is given by:
\begin{equation}
\label{fo:master_MFG}
\begin{split}
\hskip -16pt
&\partial_{t} U(t,x,\mu) + \frac{\sigma^2}2 \Delta U(t,x,\mu)
- \frac12|\partial_{x} U(t,x,\mu))|^2  
\\
&\hspace{5pt} 
- \int_{\RR^d}
\partial_{x} {U}(t,x',\mu) \cdot  \partial_{\mu} U(t,x,\mu)(x')
 d\mu(x')
\\
&\hspace{5pt} +
\frac{\sigma^2}{2}
\int_{\RR^d}
 \text{\rm trace}
 \Bigl[
\partial_{x'} \partial_{\mu} U(t,x,\mu) (x')
\Bigr]
d\mu(x')
+ \tilde f_0(t,x,\mu) = 0,
\end{split}
\end{equation} 
with terminal condition $U(T,x,\mu)=g(x,\mu)$. 
See e.g.,~\cite{CardaliaguetDelarueLasryLions,CarmonaDelarue_book_II} for more details on the derivation of the master equation. 
By inspection, one sees that the choice of $\tilde{f}_0$ given by~\eqref{fo:f_0_tilde} 
does the trick in the sense that it turns the master equation \eqref{fo:master_planner} of the MFC into the master equation of an MFG.
So if the individual players are incentivized and face the running costs $\tilde f$ as given by formulas \eqref{fo:f_0_tilde} and \eqref{fo:f_tilde}, then in an MFG equilibrium, they will have the same cost as if they were implementing the optimal control of a social planner using the running cost function \eqref{fo:running_cost}.

\end{proof}

\begin{remark}
\label{rem:match-mfc-val-all-mu}
The reader may object that, after all, there are plenty of other ways to design a new MFG with a given optimal cost. For instance, we could just consider $\tilde f_0=0$
and $\tilde g=J^{*}$ (i.e., $\tilde g$  is a constant cost, equal to the optimal cost in the MFC optimization problem introduced in~\eqref{fo:J_MKV}). Indeed, the optimizer
in the resulting MFG is zero and the equilibrium cost is obviously equal to $J^{*}$. The very interest of our approach is that
not only the costs are the same but also 
the conditional remaining costs to an 
individual player inside the population are the same in the two approaches: the common value is $V(t,x,\mu)$ at time $t$ when the player is in state $x$ and the measure is in state $\mu$, which is very much stronger (and would be false in the latter trivial case when $\tilde f_0=0$
and $\tilde g=J^{*}$). 
Accordingly, we should say that our incentivization is Markovian. 
\end{remark}

We want to emphasize that although in the new MFG, the value for a representative player is the same as the MFC optimum (i.e., social optimum) corresponding to the original model, the players' dynamics are not the same in general. Indeed, remember that the players' dynamics are directly driven by the control, see~\eqref{fo:state}. 
In the MFC with running cost $f$, the social optimum is attained by using the control $\bar\alpha$ defined in~\eqref{eq:def-bar-alpha-opt-orig-MFC}: $\alpha^{{\rMKV}}_s=\bar\alpha(s,X^{{\rMKV}}_s,\cL(X^{{\rMKV}}_s))$. In the new MFG with running cost $\tilde{f}$ defined in~\eqref{fo:f_tilde}, the MFG equilibrium control is obtained by using the control given by the minimizer of the Hamiltonian:
\begin{equation*}
\tilde{H}(t,x,\mu,y,\alpha) =\alpha y + \frac12 |\alpha|^2 + \tilde{f}_0(x,\mu)
\end{equation*}
along the solution of the MFG. In this situation, the equilibrium control is given by:
$$
    \alpha^{\rMFG}_s = - \partial_x U(s,X^{\rMFG}_s,\cL(X^{\rMFG}_s))
$$
Now, for the sake of contradiction, assume that $\bar\balpha = \balpha^{\rMFG}$. Then $\cL(X^{\rMKV}_s) = \cL(X^{\rMFG}_s)$. Using~\eqref{eq:hat-alpha-def} and the fact that $V = U$, we deduce that the equality of the two controls implies
$$
    \int_{\RR^d} \partial_{\mu} V(t,x',\mu_s)(x) d\mu_s(x') = 0,
$$
where $\mu_s = \cL(X^{\rMKV}_s)$. This is not true in general, hence a contradiction. For an example where this term is non-zero, see \cite[Section 4.5.1, pp. 311--313]{CarmonaDelarue_book_II}.

\subsection{Matching the Mean Field Control Optimal Behavior} 
\label{subsec:MFGtoMFC_control}
The goal of this section is to show that one can incentivize the individual players in an MFG (by modifying their running and terminal cost functions) in such a way that, while they still remain in a Nash equilibrium, they end up behaving (in terms of their control and actual state) exactly as if they were adopting the MFC optimal control identified by a social planner optimizing the original MFC cost.

\subsubsection{Incentivization Mechanism in Mean Field Game}

Here again, we identify the MFG equilibrium via the FBSDE of McKean-Vlasov type derived from the stochastic Pontryagin principle. The starting point is the characterization of the behavior (in terms of control and state processes) of individual players adopting the prescriptions identified by the social planner's optimization. As explained in Section~\ref{subsec:model_mfc}, the state dynamics at the MFC optimum are given by the forward component of the solution of the FBSDE \eqref{fo:mkv_FBSDE}
while the MFC optimal control is given by $\alpha^{{\rMKV}}_t=-Y_t$ for $0\le t\le T$.

Now, we introduce new running and terminal cost functions for the MFG defined in Section~\ref{subsec:model_mfg} with the running cost~\eqref{fo:running_cost}. For each $\lambda\in[0,1]$, we define a new cost function $f_\lambda$ by:
\begin{equation}
\label{fo:f_lambda}
    f_\lambda(x,\mu)= f_0(x,\mu) +\lambda\tilde\EE\Bigl[\frac{\delta f_0}{\delta m}(\tilde X,\mu)(x)\Bigr], 
\quad (x,\mu) \in {\mathbb R}^d \times {\mathcal P}_2({\mathbb R}^d),
\end{equation}
where we recall the notations introduced in Section~\ref{sec:notation}.  

Now, we consider the MFG with controlled state dynamics  \eqref{fo:state}, a running cost function: 
$$
\frac12|\alpha|^2+f_\lambda(x,\mu),
$$
and a terminal cost function $g_\lambda(x,\mu)= g(x,\mu) +\lambda\tilde\EE[\frac{\delta g}{\delta m}(\tilde X,\mu)(x)]$.

\begin{proposition}
\label{prop:MKV:MFG:lambda=1}

Assume that the system 
\eqref{fo:mkv_FBSDE} is uniquely solvable and that $f_1$ and $g_1$ (together with $\sigma$) satisfy Assumption~\ref{assumption:SMP:MFG}. Then, for $\lambda=1$, the MFG associated with $(f_1,g_1)$ has the same solution as the original MFC problem (Definition~\ref{def:MFCpb}). 

\end{proposition}

Before presenting the proof, let us recall that, under 
Assumption 
\ref{assumption:existence:uniqueness:MFC}, the system \eqref{fo:mkv_FBSDE} is indeed uniquely solvable. 
As far as $f_1$ and $g_1$ are concerned, 
the assumption $(i)$ in Assumption~\ref{assumption:SMP:MFG} is satisfied if 
$f_0$ and $g$ are Lipschitz continuous in $(x,\mu)$ when ${\mathcal P}_2({\mathbb R}^d)$ is equipped with 
the 
Total Variation distance. Indeed, this forces the derivatives $\delta f_0/\delta m$ and 
$\delta g_0/\delta m$ to be bounded. 
The assumption $(ii)$ in Assumption~\ref{assumption:SMP:MFG} holds
when $f_0$ and $g$ satisfy 
\eqref{eq:convex:case} with $\varphi_0$ and $\psi$ convex. 

\begin{proof}[Proof of Proposition~\ref{prop:MKV:MFG:lambda=1}]
We first start by identifying equilibria in the MFG driven by $(f_\lambda, g_\lambda)$ for $\lambda\in[0,1]$ via the stochastic maximum principle with the new cost functions. The MFG equilibrium state dynamics are given by the forward component of the solution of the FBSDE:
\begin{equation}
\label{fo:mkv_fbsde_lambda}
    \begin{cases}
    dX_t&= -Y_tdt +\sigma dW_t\\
    dY_t&=-\partial_x f_\lambda(X_t,\cL(X_t)) dt + Z_t dW_t,\\
    Y_T&=\partial_x g_\lambda(X_T,\cL(X_T)),
    \end{cases}
\end{equation}
while since the dependence of the Hamiltonian with respect to control is the same as before, the equilibrium control is still given by $\alpha_t=-Y_t$ for $0\le t\le T$.
Clearly, this MFG equilibrium coincides with our original MFG equilibrium as characterized by \eqref{fo:mfg_FBSDE} when $\lambda=0$. Moreover, since the L-derivative and the functional derivatives are related by the formula~\eqref{eq:connection-partialmu-deltam} (with $U = f$ here), 
we see that the solution of this MFG as given by \eqref{fo:mkv_fbsde_lambda} coincides for $\lambda=1$ with the solution of the MKV FBSDE \eqref{fo:mkv_FBSDE}, which is known, 
under the standing assumption, to be the unique MFC minimizer. 

Now, it remains to observe from the additional assumption we made on $(f_1,g_1)$ that, for any equilibrium $(\mu_t)_{0 \le t \le T}$ of the MFG driven by $(f_1,g_1)$, the cost functional in the environment $(\mu_t)_{0 \le t \le T}$ has a unique minimizer, which is uniquely characterized by the stochastic 
maximum principle, see for instance \cite[Chapter 3]{CarmonaDelarue_book_I}.
Therefore,
there is a one-to-one mapping between the equilibria of the
mean field game driven by $(f_1,g_1)$
and the solutions 
to  
\eqref{fo:mkv_FBSDE}. 
\end{proof}

\begin{remark}
\label{rem:interpr_controlincent}
The interpretation of the statement is as follows: Without the intervention of the social planner and without forcing all the players to use the control identified by someone else, we can get to the same state trajectorya and behavior by letting the individual players settle in a Nash equilibrium for an MFG with perturbed cost functions, the perturbations having the interpretation of incentives.

Also, it should be clear that the game that we have at 
$\lambda=1$ is a potential game. In short, 
$f_1$ coincides with the linear functional derivative of 
the function 
$\mu \mapsto  \int_{{\mathbb R}^d} f_0(x,\mu) d \mu(x),$ 
and similarly for $g_1$.
\end{remark}

\subsubsection{Another Interpretation: $\lambda$-Interpolated Mean Field Games}
\label{sec:lambda-interp-def}

The  procedure 
introduced in the above subsection provides an interpolation (parameterized by $\lambda$) between the solution(s) of a given MFG and the solution(s) of a given MFC with the same state equation and running and terminal cost functions. Our goal in this section is to study this interpolation. We will start by giving another interpretation
of the FBSDE system 
\eqref{fo:mkv_fbsde_lambda}.

We introduce, for a generic flow 
${\boldsymbol \mu}:=(\mu_t)_{0 \le t \le T}$ from 
$[0,T]$ to ${\mathcal P}_2({\mathbb R^d})$, 
the cost $J^{\lambda,{\rm MF}}
( \balpha ; {\boldsymbol \mu} )$:
\begin{equation}
\label{eq:def-J-lambda-MF}
\begin{split}
J^{\lambda,{\rm MF}}
\bigl( \balpha ; {\boldsymbol \mu} \bigr)
:=& 
(1-\lambda) 
J\bigl( {\balpha};{\boldsymbol \mu}\bigr) 
+
\lambda J^{{\rMKV}} 
\bigl( {\balpha}\bigr)
\\
=& 
{\mathbb E}
\biggl[ \frac12 \int_0^T 
\vert \alpha_t \vert^2 dt 
+ \int_0^T
\Bigl[
(1-\lambda) f_0 \bigl(X_t^{\balpha},\mu_t\bigr) 
+
\lambda f_0 \bigl( X_t^{\balpha},{\mathcal L}(X_t^{\balpha}) 
\bigr) 
\Bigr] dt
\biggr]
\\
&\hspace{15pt}
+
{\mathbb E} 
\Bigl[
(1-\lambda) 
g
\bigl( X_T^{\balpha},
\mu_T
\bigr) 
+ 
\lambda g
\bigl( X_T^{\balpha}, 
{\mathcal L}(X_T^{\balpha})
\bigr) 
\Bigr]. 
\end{split}
\end{equation}

\begin{definition}
For a given $\lambda \in [0,1]$, 
we say that 
a (square-integrable)
control $\balpha^{\blambda}$
induces
a $\lambda$-interpolated 
mean field equilibrium if
 $\balpha^{\blambda}$ solves the minimization problem
\begin{equation*}
\inf_{\balpha} 
J^{\lambda,{\rm MF}}
\bigl( \balpha ; {\boldsymbol \mu}^{\lambda} \bigr),
\end{equation*}
where
${\boldsymbol \mu}^{\lambda}:=
(\mu_t^{\lambda})_{0 \le t \le T}={\mathcal L}(X_t^{\lambda})$, 
for $t \in [0,T]$ where $\bX^{\lambda}$ is the state process solving~\eqref{fo:state} controlled by $\balpha^{\lambda}$.

\end{definition}

Of course, 
a $0$-interpolated mean field equilibrium is an MFG equilibrium and a $1$-interpolated mean field equilibrium is an MFC optimum. 
This problem falls in the category of \emph{mean field control games} introduced in~\cite{angiuli2022reinforcement,angiuli2022reinforcementmixedmfcg}. Such games are an extension of MFGs where, given the population distribution flow, each player solves an MFC problem instead of a standard stochastic control problem. In this case too, the cost depends on the population distribution as well as on the individual player's distribution. In fact, we can write:
\[
    J^{\lambda,{\rm MF}}\bigl( \balpha ; {\boldsymbol \mu} \bigr)
    := {\mathbb E} \biggl[ \frac12 \int_0^T \vert \alpha_t \vert^2 dt + \int_0^T f_0 \bigl(X_t^{\balpha}, \mu_t, {\mathcal L}(X_t^{\balpha}) \bigr) dt
    + g_\lambda\bigl( X_T^{\balpha}, \mu_T, {\mathcal L}(X_T^{\balpha})\bigr)\biggr],
\]
with $f_\lambda(x, \mu, \tilde\mu) = (1-\lambda) f_0 \bigl(x,\mu\bigr) 
+ \lambda f_0 \bigl( x, \tilde\mu \bigr)$ and $g_\lambda(x, \mu, \tilde\mu) = (1-\lambda) g \bigl(x,\mu\bigr) 
+ \lambda g \bigl( x, \tilde\mu \bigr)$, where $\mu$ and $\tilde\mu$ play the role of the distributions called respectively \emph{global} and \emph{local} in~\cite{angiuli2022reinforcement,angiuli2022reinforcementmixedmfcg}.

\begin{proposition}
\label{prop:lambda_interpolated_FBSDEcharacterization}
Let $\lambda \in [0,1]$. Assume that 
$\balpha^{\lambda}$ is a 
$\lambda$-interpolated mean field equilibrium. Then, 
the pair of processes 
$(X_t,Y_t)_{0 \leq t \leq T}=(X_t^{\lambda},-\alpha^{\lambda}_t)_{0 \leq t \leq T}$
solves
the system 
\eqref{fo:mkv_fbsde_lambda}.
\end{proposition}

\begin{proof}[Proof of Proposition~\ref{prop:lambda_interpolated_FBSDEcharacterization}]
The proof is a standard application of the stochastic maximum principle for 
controlled McKean-Vlasov
processes. 
When the environment 
$\bmu$ is fixed, 
the minimizers of 
$J^{\lambda,{\rm MF}}(\cdot \, ; 
{\boldsymbol \mu})$
solves the system 
\eqref{fo:mkv_fbsde_lambda}
but with $f_\lambda(X_t,{\mathcal L}(X_t))$ replaced by 
$(1-\lambda) f_0(X_t,\mu_t)+
\lambda f_0(X_t,{\mathcal L}(X_t))$ and similarly for 
$g_\lambda$. Under the fixed point condition ${\boldsymbol \mu}=({\mathcal L}(X_t))_{0 \le t \le T}$, we get 
\eqref{fo:mkv_fbsde_lambda}. 
\end{proof}

Intuitively, this interpolation can be viewed as a situation in which an \textit{invisible hand} tunes the cost function: starting from the MFG setting, the cost gradually incorporates more and more of the MFC setting. We can imagine that this process is done iteratively: at each step, the invisible hand changes the previous cost (by increasing $\lambda$), making it a bit more similar to the MFC cost; then the players compute a Nash equilibrium for this modified game, before moving to the next step. 

In the remaining, we will discuss the continuity property of the $\lambda$-interpolated mean field equilibrium and the effect of monotonicity assumption for the cost functions. For the finite player interpretation of $\lambda$-interpolated mean field equilibrium, please refer to Appendix~\ref{app:MFGtoMFC_control}.

\subsubsection{Continuous Deformation from Mean Field Game to Mean Field Control}

An interesting question in practice is whether we can find, for each 
$\lambda \in [0,1]$, a $\lambda$-interpolated mean field equilibrium such that the resulting equilibria are continuous w.r.t. the parameter $\lambda$: intuitively, this would mean that there exists a continuous deformation connecting an MFG equilibrium and an MFC optimum. We can easily provide 
some intuitive
applications: when playing  
games repeatedly, this would allow to drive continuously
the equilibrium state of the population to a social optimum. Continuity implies that there is no big jump in the successive costs paid by the players during these iterations. From a modelling point of view, the invisible hand may want to change the rules in a soft manner that does not perturb the collectivity too abruptly.
This question is however difficult to solve in full generality. Indeed, without any further structural conditions guaranteeing uniqueness, 
existing compactness methods used in MFG theory to prove existence of an equilibrium are certainly not sufficient to prove the existence of such a 
continuous deformation.
An additional continuous selection would be 
necessary, but this would be outside the technical scope of this paper. 
For this reason, we stick below to richer cases
for which continuous deformations can be shown to exist by simpler arguments. 

Generally speaking, if uniqueness of $\lambda$-interpolated mean field equilibria holds true for 
any $\lambda \in [0,1]$ and if the resulting 
equilibria have finite second-order moments, uniformly in time and in $\lambda$, we may expect to prove continuity by a standard compactness argument. As an example of uniqueness condition, 
we have, in line with Remark 
\ref{rem:convex}, the following statement:
\begin{proposition}
\label{prop:convex}
Assume that 
 $f_0$ and $g$
are given as\footnote{For simplicity, we use the above definition of $f_0$, although in this case, $\frac{\delta f_0}{\delta m}$ may fail to satisfy the normalization condition mentioned in Section~\ref{sec:notation} to define the derivative in a unique way. However, we could use a different convention, and the same ideas could be applied with small changes to ensure that the normalization condition is satisfied.}
\begin{equation*}
    f_0(x,\mu) = \int_{{\mathbb R}^d} \varphi_0(x-y) d\mu(y), 
    \quad 
    g(x,\mu) = \int_{{\mathbb R}^d} \psi(x-y) d\mu(y),
\end{equation*}
where $\varphi_0$ and $\psi$ on the right-hand sides are  even
convex functions on ${\mathbb R}^d$ with Lipschitz continuous derivatives.
Then, for any $\lambda \in [0,1]$, there exists 
a unique $\lambda$-interpolated mean field equilibrium control 
$\balpha^\lambda$
and the mapping 
\begin{equation*}
    [0,1] \ni \lambda 
    \mapsto (X_t^{\lambda})_{0 \le t \le T}
\end{equation*}
is continuous when distances between two different values of the right-hand side are understood for the norm  
$\| {\boldsymbol X} \|_{{\mathbb S}_2} := \sup_{0 \le t \le T} {\mathbb E}[ \vert X_t \vert^2]^{1/2}$. 
\end{proposition}

\begin{proof}[Proof of Proposition~\ref{prop:convex}]
In this setting, we even have more than what is said in the statement. Indeed,  
we  can observe that (recalling Section~\ref{sec:notation}): 
\begin{equation*}
    \tilde{\mathbb E} 
    \bigl[ \frac{\delta f_0}{\delta m}\bigl( \tilde X,\mu \bigr)(x)\bigr]=
    {\mathbb E} 
    \bigl[ \varphi_0 \bigl( X-x) \bigr]
    = f_0(x,\mu),
\end{equation*}
where $X$ is a random variable 
with $\mu$ as its distribution. Therefore, 
$f_\lambda(x,\mu)=(1+\lambda) f_0(x,\mu)$
and the system 
\eqref{fo:mkv_fbsde_lambda}
is not only the 
mean field forward-backward system that arises
when applying the 
Pontryagin principle
to describe the MFG equilibria driven 
by the two 
costs $f_\lambda$ and $g_\lambda$,  but it is also 
the Pontryagin system 
deriving from the MFC problem 
associated with 
the runnning cost
\begin{equation*}
    \mu \mapsto \frac{1+\lambda}2 \int_{{\mathbb R}^d} 
    \int_{{\mathbb R}^d}
    \varphi_0(x-y) d \mu(x) d \mu(y),
\end{equation*}
and the terminal cost 
\begin{equation*}
    \mu \mapsto 
    \frac{1+\lambda}2 
    \int_{{\mathbb R}^d}
    \int_{{\mathbb R}^d} 
    \psi(x-y) d \mu(x) 
    d\mu(y). 
\end{equation*}
Following Assumptions~\ref{assumption:initial-model}  and \ref{assumption:existence:uniqueness:MFC} (see also Remark
\ref{rem:convex}), this MFC problem has a unique solution, 
which is fully characterized by the maximum principle. 
Hence there exists at most one $\lambda$-interpolated mean field equilibrium. As for existence, the point is to check that the Pontryagin system 
\eqref{fo:mkv_fbsde_lambda}
is, in this setting, a sufficient condition. This follows again from the convexity conditions on the coefficients, 
 which imply that Assumption~\ref{assumption:SMP:MFG} is in force (the properties of $f_\lambda$ and $g_\lambda$ are deduced from the ones of $f_0$ and $g$).

As for the continuity w.r.t. $\lambda$, 
we take again benefit of the convex structure of the coefficients in order to bypass any compactness arguments, which would be more technical. 
Indeed, 
the
coefficients in the system 
\eqref{fo:mkv_fbsde_lambda} satisfy the following monotonicity condition: 
\begin{equation}
\label{eq:monotonicity}
    \begin{split}
    {\mathbb E}
    \Bigl[
    \Bigl\langle
    \partial_x f_\lambda\bigl( X, {\mathcal L}(X)\bigr) 
    -
    \partial_x f_\lambda 
    \bigl( X',{\mathcal L}(X')\bigr), X-X' 
    \Bigr\rangle 
    \Bigr] \geq 0, 
    \end{split}
\end{equation}
for any two random variables $X,X' \in L^2(\Omega,{\mathcal A},{\mathbb P};{\mathbb R}^d)$. 
The point is then to prove stability by expanding the inner product \begin{equation*}
    {\mathbb E}
    \bigl[ \bigl\langle X_t^\lambda - X_t^{\lambda'}, Y_t^\lambda - Y_t^{\lambda'} \bigr\rangle \bigr],
\end{equation*}
which is the usual strategy for forward-backward systems with monotone coefficients. 

We have
\begin{equation*}
    \begin{split}
    &\frac{d}{dt}
    {\mathbb E}
    \bigl[ \bigl\langle X_t^\lambda - X_t^{\lambda'}, Y_t^\lambda - 
    Y_t^{\lambda'} 
    \bigr\rangle \bigr]
    \\
    &= - {\mathbb E}
    \Bigl[ \bigl\vert Y_t^\lambda - 
    Y_t^{\lambda'} \bigr\vert^2\Bigr]
    - 
    {\mathbb E}
    \Bigl[
    \Bigl\langle
    \partial_x f_{\lambda'}\bigl( X_t^\lambda, {\mathcal L}(X_t^{\lambda})\bigr) 
    -
    \partial_x f_{\lambda'} 
    \bigl( X_t^{\lambda'},{\mathcal L}(X_t^{\lambda'})\bigr), X_t^\lambda-X_t^{\lambda'}
    \Bigr\rangle 
    \Bigr]
    \\
    &\hspace{15pt} - 
    {\mathbb E}
    \Bigl[
    \Bigl\langle
    \partial_x f_{\lambda}\bigl( X_t^\lambda, {\mathcal L}(X_t^{\lambda})\bigr) 
    -
    \partial_x f_{\lambda'} 
    \bigl( X_t^{\lambda},{\mathcal L}(X_t^{\lambda})\bigr), X_t^\lambda-X_t^{\lambda'} 
    \Bigr\rangle 
    \Bigr],
    \end{split}
\end{equation*}
and then, expressing $\partial_x f_{\lambda} - \partial_x f_{\lambda'}$ in terms of $(\lambda-\lambda')$ and 
performing a similar expansion for the boundary condition, we get 
(using 
\eqref{eq:monotonicity}):
\begin{equation*}
    \begin{split}
    {\mathbb E}
    \int_0^T \vert Y_t^{\lambda} 
    - Y_t^{\lambda'}
    \vert^2 dt 
    \leq 
    C_\lambda
    \vert \lambda - \lambda' \vert 
    \sup_{0 \leq t \leq T} {\mathbb E}\bigl[ \vert X_t^{\lambda} 
    - X_t^{\lambda'} \vert^2 \bigr]^{1/2},
    \end{split}
\end{equation*}
where the constant 
$C_\lambda$ depends on the second order moments of $(X_t^\lambda)_{0 \le t \le T}$. Using the fact that $|X_t^{\lambda} 
    - X_t^{\lambda'}| \le \int_0^t|Y_s^{\lambda} 
    - Y_s^{\lambda'}|ds$ and plugging this in the right-hand side above, we deduce that 
\begin{equation*}
    \begin{split}
    {\mathbb E}
    \int_0^T \vert Y_t^{\lambda} 
    - Y_t^{\lambda'}
    \vert^2 dt 
    \leq 
    C_\lambda
    \vert \lambda - \lambda' \vert^2,
    \end{split}
\end{equation*}
for a possibly new value of $C_\lambda$ and then
\begin{equation*}
    \sup_{0 \leq t \leq T} {\mathbb E}\bigl[ \vert X_t^{\lambda} 
    - X_t^{\lambda'} \vert^2 \bigr]^{1/2}
    \leq C_\lambda \vert \lambda - \lambda' \vert,
\end{equation*}
which completes the proof. 
\end{proof}

\begin{remark}
In fact, the statement could be generalized to 
a cost
$f_0$ (and similarly for $g$), with Lipschitz continuous derivatives in $(x,\mu)$, such that:
\begin{enumerate}
\item The function  
\begin{equation*}
    \mu \in {\mathcal P}_2({\mathbb R}^d) \mapsto \int_{{\mathbb R}^d} f_0(x,\mu) d \mu(x) 
\end{equation*}
is displacement convex, i.e., 
\begin{equation*}
    X \in L^2(\Omega,{\mathcal A},{\mathbb P};{\mathbb R}^d) \mapsto 
    {\mathbb E}
    \bigl[ f_0\bigl(X,{\mathcal L}(X)\bigr) \bigr]
\end{equation*}
is convex;
\item
For any $m \in {\mathcal P}_2({\mathbb R}^d)$, the function 
${\mathbb R}^d \ni x \mapsto f_0(x,m)$
is convex, which implies in particular that the function 
\begin{equation*}
\mu \in {\mathcal P}_2({\mathbb R}^d) \mapsto \int_{{\mathbb R}^d} f_0(x,m) d \mu(x) 
\end{equation*}
is displacement convex i.e., 
\begin{equation*}
    X \in L^2(\Omega,{\mathcal A},{\mathbb P};{\mathbb R}^d) \mapsto 
    {\mathbb E}
    \bigl[ f_0\bigl(X,m\bigr) \bigr]
\end{equation*}
is convex;
\item The function $f_0$ satisfies the  
monotonicity property
\begin{equation*}
    {\mathbb E}
    \Bigl[ 
    \bigl\langle \partial_x f_0\bigl( X, {\mathcal L}(X) \bigr)
    -
    \partial_x f_0\bigl( Y, {\mathcal L}(Y) \bigr), X- Y
    \bigr\rangle
    \Bigr]
    \geq 0,
\end{equation*}
for any two random variables $X,Y \in L^2(\Omega,{\mathcal A},{\mathbb P};{\mathbb R}^d)$. 
\end{enumerate}
Briefly, items (1) 
and (2) say that, for a given 
$\lambda \in [0,1]$ and a given flow ${\boldsymbol \mu}$, there is a unique minimizer to the 
cost $J^{\lambda,\rm MF}(\cdot; {\boldsymbol \mu})$. Items (1) and (3) 
imply
\eqref{eq:monotonicity}
and thus 
say that there is a unique solution 
to the system 
\eqref{fo:mkv_fbsde_lambda}, which is in fact the unique minimizer to  
$J^{\lambda,\rm MF}(\cdot; {\boldsymbol \mu})$, 
when ${\boldsymbol \mu} = (\mu_t)_{0 \le t \le T}$ is 
given by $\mu_t = {\mathcal L}(X_t)$, 
for $(X_t)_{0 \le t \le T}$ the forward component in the solution of 
\eqref{fo:mkv_fbsde_lambda}.
For instance, if $f_0(x,\mu)$ is given, one may add 
$A \vert x \vert^2$ to it to enforce the above sets 
of conditions. 
\end{remark}

\subsubsection{Monotone Interactions}

We now address the case when the two cost functions 
$f_0$ and $g$ satisfy the Lasry-Lions monotonicity condition:
\begin{equation}
    \label{eq:def:monotonicity}
    \forall m, m' \in {\mathcal P}({\mathbb R}^d), \quad 
    \int_{{\mathbb R}^d} \Bigl( f_0(x,m') - f_0(x,m) \Bigr) d \bigl( m' - m \bigr) (x) \geq 0, 
\end{equation}
and similarly for $g$. 
We then have the following statement: 
\begin{proposition}
\label{prop:lambda_monotonicity}
Let $0 \leq \lambda < \lambda ' \leq 1$. Assume that ${\boldsymbol \alpha}^{\lambda}$
(together with the flow ${\boldsymbol \mu}^\lambda$) 
and ${\boldsymbol \alpha}^{\lambda'}$
(together with the flow ${\boldsymbol \mu}^{\lambda'}$)
are two interpolated mean field equilibria, respectively 
with $\lambda$ and $\lambda'$ as parameters. 
Then, necessarily
\begin{equation*}
    J^{\lambda', \rm MF}\bigl( {\boldsymbol \alpha}^{\lambda'};{\boldsymbol \mu}^{\lambda'} \bigr) \leq 
    J^{\lambda, \rm MF}\bigl( {\boldsymbol \alpha}^{\lambda};{\boldsymbol \mu}^{\lambda} \bigr).
\end{equation*}
Moreover,  if $\lambda' <1$
and one of the two functions $f_0$ or $g$ is strictly monotone, then the inequality is strict. 
\end{proposition}
\begin{proof}
By optimality of each of the two equilibria (w.r.t. the 
corresponding parameter), we have the following two inequalities: 
\begin{equation*}
    \begin{split}
    &J^{\lambda,{\rm MF}}\bigl( \balpha^{\lambda};\bmu^{\lambda}) \leq 
    J^{\lambda,{\rm MF}}\bigl( \balpha^{\lambda'};\bmu^{\lambda}),
    \\
    &J^{\lambda',{\rm MF}}\bigl( \balpha^{\lambda'};\bmu^{\lambda'}) \leq 
    J^{\lambda',{\rm MF}}\bigl( \balpha^{\lambda};\bmu^{\lambda'}).
    \end{split}
\end{equation*}
Therefore,
multiplying the first inequality by 
$1-\lambda'$
and the second one by 
$1-\lambda$
and then 
adding the resulting two inequalities, 
we get
\begin{equation*}
    \begin{split}
    \bigl( \lambda' (1-\lambda) - \lambda(1-\lambda')\bigr) J^{\rMKV}\bigl( 
    \balpha^{\lambda'} \bigr) 
    &\leq 
    \bigl( \lambda'(1-\lambda) - \lambda(1-\lambda')\bigr) J^{\rMKV}\bigl( 
    \balpha^{\lambda} \bigr) 
    \\
    &\hspace{15pt} +
     (1-\lambda ')
      (1-\lambda ) \bigl( J\bigl( \balpha^{\lambda};\bmu^{\lambda'}) - 
    J\bigl( \balpha^{\lambda'};\bmu^{\lambda'})\bigr) 
    \\
    &\hspace{15pt} + 
    (1-\lambda)  (1-\lambda ')  \bigl( J\bigl( \balpha^{\lambda'};\bmu^{\lambda}) - 
    J\bigl( \balpha^{\lambda};\bmu^{\lambda})\bigr).
    \end{split}
\end{equation*}
Using monotonicity to show that $J\bigl( \balpha^{\lambda};\bmu^{\lambda'}) - 
J\bigl( \balpha^{\lambda'};\bmu^{\lambda'}) - ( J\bigl( \balpha^{\lambda};\bmu^{\lambda}) - J\bigl( \balpha^{\lambda'};\bmu^{\lambda}) ) \le 0$, we obtain:
\begin{equation}
\label{eq:interpolated:comparison}
    \begin{split}
    \bigl( \lambda' - \lambda\bigr) J^{\rMKV}\bigl( 
    \balpha^{\lambda'} \bigr) 
    &\leq 
    \bigl( \lambda' - \lambda\bigr) J^{\rMKV}\bigl( 
    \balpha^{\lambda} \bigr),
    \end{split}
\end{equation}
which can be reformulated as 
\begin{equation*}
    J^{\lambda',\rm MF}\bigl( 
    \balpha^{\lambda'};
    {\boldsymbol \mu}^{\lambda'} \bigr) =
    J^{\rMKV}\bigl( 
    \balpha^{\lambda'} \bigr) 
    \leq 
    J^{\rMKV}\bigl( 
    \balpha^{\lambda} \bigr) 
    =
    J^{\lambda ,\rm MF}\bigl( 
    \balpha^{\lambda};
    {\boldsymbol \mu}^{\lambda}\bigr). 
\end{equation*}
This proves the first claim. 
The second claim is shown by noticing that
the inequality 
\eqref{eq:interpolated:comparison}
becomes strict when 
$\lambda' <1$ and one of the two functions, $f_0$ or $g$, is strictly monotone. 
\end{proof}

We draw several conclusions from the previous lemma, in the monotone setting:
\begin{enumerate}
    \item First, the equilibrium cost is decreasing w.r.t 
    $\lambda$;
    \item Second, interpolated mean field equilibria are at most unique when 
    $\lambda <1$ and one of the two functions, $f_0$ or $g$ satisfies the Lasry-Lions monotonicity condition strictly.
\end{enumerate}

\begin{remark}
    As stated in the second observation above when $\lambda=1$ the uniqueness is not guaranteed. The deformation process can help us select one equilibrium at $\lambda=1$ by studying the $\lambda$-interpolated mean field equilibrium as $\lambda\rightarrow 1$.
\end{remark}

\color{black}

\section{From MFC to MFG: Deviation from Social Optimum}
\label{sec:MFCtoMFG}

The goal of this section is to discuss what happens when individual players that are previously controlled by the social planner are let to deviate from the social planner's prescribed behavior. First, we introduce the Price of Instability (PoI) notion to quantify the instability of the MFC optimum by letting one infinitesimal player deviate from the social optimum. Then we introduce the notion of $p$-partial mean field equilibrium where a proportion $p$ of the population is allowed to deviate. From here, we introduce a generic algorithm to describe the iterative deviations, where the game is played over and over again with a varying proportion of players deviating from the social optimum.

\subsection{Price of Instability of a Social Optimum and $p$-Partial Mean Field Game}
\label{subsec:MFCtoMFG_completeinfo}

\subsubsection{Single Player Deviation: PoI of a Social Optimum}
\label{subsubsec:MFCtoMFG_completeinfo_PoI}
As we argued in Section \ref{sec:model} (recall formula \eqref{fo:PoA}),  the social cost of any MFG equilibrium is higher than the social cost (per individual) incurred when the players all agree to use the (common) control identified by a social planner when computed via the solution of the MFC problem. Letting players minimize their individual costs, even when they reach a Nash equilibrium, comes at a cost when compared to a centrally coordinated optimization. Various forms of quantification of this difference in cost have been proposed, as we mentioned in the introduction, a popular one being known as the PoA.
In this section, we argue that while less costly, the MFC optimum is less stable.
In this section, we quantify this level of instability. 

\vskip 6pt
In order to do so, we evaluate how much an individual player's average cost can be lowered if she decides to deviate unilaterally from the MFC optimal control $\balpha^{\rMKV}$ identified by the social planner. If she follows the control identified by the social planner, the individual player's cost is
\begin{equation}
J^* := J^{\rMKV}(\boldsymbol{\alpha}^{\rMKV})=\EE\Bigl[\int_0^T \Bigl(\frac12|\alpha^{\rMKV}_t|^2+f_0(X^{\rMKV}_t,\mu^{\rMKV}_t)\Bigr)dt + g(X^{\rMKV}_T,\mu^{\rMKV}_T)\Bigr].    
\end{equation}

On the other hand, if she allows herself to deviate from this control, still evolving in the same environment, her cost can become
\begin{equation}
\label{eq:def-hat-J0}
    \hat J_0 := J^{\bmu^{\rMKV}}(\hat\balpha) = \EE\Bigl[\int_0^T \Bigl(\frac12|\hat\alpha_t|^2+f_0(\hat X_t,\mu^{\rMKV}_t)\Bigr)dt + g(\hat X_T,\mu^{\rMKV}_T)\Bigr]
\end{equation}
where $d\hat X_t = \hat\alpha_t dt +\sigma dW_t$ and
\begin{equation}
\begin{aligned}
    \hat \balpha = \argmin_{\balpha}\ &\EE\Bigl[\int_0^T \Bigl(\frac12|\alpha_t|^2+f_0(X^{\balpha}_t,\mu^{\rMKV}_t)\Bigr)dt + g(X^{\balpha}_T,\mu^{\rMKV}_T)\Bigr]\\[2mm]
    & \text{s.t.   } dX^{\balpha}_t = \alpha_t dt + \sigma dW_t.
\end{aligned}
\end{equation}

\begin{definition}
\label{de:PoI}
The Price of Instability (PoI)  is defined as the quantity:
\begin{equation}
\label{fo:PoI}
{\rm PoI}= J^* - \hat J_0
\end{equation}
where $J^*$ is the cost of the mean field control problem, and $ \hat J_0$, defined in~\eqref{eq:def-hat-J0}, is the cost associated to the optimal control in the classical control problem in the environment given by the marginal distributions of the optimal state in the MFC (i.e., social planner's optimization) problem. 
\end{definition}

By construction, $\text{PoI}\ge 0$, and an interesting question is to understand when the PoI is (strictly) positive. We start with the following remark.
\begin{remark}\label{rem:PoI0}
    If ${\rm PoI}=0$,  ${\balpha}^{\rMKV}$ is an MFG equilibrium control. If the MFG equilibrium is unique, then ${\rm PoA} = 1$. However, in principle, it is possible to have PoI arbitrarily close to $0$ and PoA arbitrarily large. This would correspond to a situation where the social optimum is ``almost stable'' under unilateral deviations, but when all the agents take unilateral deviations, the situation gradually evolves toward a Nash equilibrium with a much higher average cost. 
\end{remark}

Before
we state the next result, we recall
from 
\cite[Theorems 2.2 \& 2.3]{lacker_2017}
that any optimal control can be assumed to be in  
a feedback form, i.e., $({\alpha}_t^{\rMKV})_{0 \le t \le T}= (\alpha(t,X_t^{\rMKV}))_{0 \le t \le T}$, in the sense
that we can always find 
a control in feedback form for which the marginal laws 
of the optimal trajectory are the same. We then have the following result:

\begin{proposition}
\label{prop:PoI}
Assume that 
${\boldsymbol \alpha}^{\rMKV}$ 
can be written in a feedback form, 
with a bounded feedback function 
$\alpha$ that is bounded in time and Lipschitz continuous in $x$. 

\begin{enumerate}
    \item If ${\rm PoI} = 0$, then 
it must hold, 
\begin{equation}
\label{eq:PoI:1}
    \int_{{\mathbb R}^d}
    \partial_\mu f_0\bigl(x,\mu_t^{\rMKV}\bigr)(y) 
    d \mu_t^{\rMKV}(x) = 0, \qquad y \in {\mathbb R}^d,  t \in [0,T]
\end{equation}
and
\begin{equation}
\label{eq:PoI:2}
    \int_{{\mathbb R}^d}
    \partial_\mu g\bigl(x,\mu_T^{\rMKV}\bigr)(y) 
    d \mu_T^{\rMKV}(x) = 0, \qquad y \in {\mathbb R}^d.
\end{equation}
    \item If $f_0$ and $g$ are twice continuously differentiable, then ${\rm PoI} \ge \frac{1}{4C} \EE \int_0^T |Y_t|^2 dt$, where $C$ is a constant which depends only on the model's coefficients and 
    \begin{equation}
    \label{eq:PoI-dfdg-defY}
    Y_t = {\mathbb E} 
    \biggl\{
    \widetilde{\mathbb E}
    \biggl[ 
    \partial_\mu 
        g \bigl( \widetilde{X}_T^{\rMKV}, \mu_T^{\rMKV} \bigr)
        \bigl( X_T^{\rMKV}\bigr)
        +
    \int_t^T
    \partial_\mu f_0
    \bigl( 
    \widetilde{X}_s^{\rMKV}, 
    \mu_s^{\rMKV}
    \bigr)\bigl(X_s^{\rMKV} \bigr) ds 
    \biggr] \, 
        \Big\vert \,
        {\mathcal F}_t \biggr\}. 
\end{equation}
\end{enumerate}
\end{proposition}

\begin{remark}
The following remarks are in order:
\begin{enumerate}
\item The first point in Proposition 
\ref{prop:PoI} says that cases of ${\rm PoI}=0$ are untypical, as conditions \eqref{eq:PoI:1} 
and \eqref{eq:PoI:2} are certainly difficult to satisfy. 
For instance, when $f_0$ is given by a convolution of the form 
\begin{equation*}
    f_0(x,\mu) = \int_{{\mathbb R}^d} \varphi_0(x-y) d\mu(y), 
    \quad x \in {\mathbb R}^d, 
    \ \mu \in {\mathcal P}({\mathbb R}^d), 
\end{equation*}
where $f_0$ in the right-hand side is a real-valued function on ${\mathbb R}^d$,
condition~\eqref{eq:PoI:1} rewrites as
\begin{equation*}
    \int_{{\mathbb R}^d}
    \nabla \varphi_0\bigl(x-y) d \mu_t^{\rMKV}(x) = 0, \qquad y \in {\mathbb R}^d, t\in[0,T], 
\end{equation*}
where $\nabla \varphi_0$ is the gradient of the function $\varphi_0$, and it is continuous and has at most linear growth by assumption.
Hence condition~\eqref{eq:PoI:1} cannot be true 
for a subset of initial conditions 
$\mu_0^{\rMKV}$ that has $\delta_0$ as accumulation point, unless $\varphi_0$ is a constant function. 

We may think of other conditions precluding
\eqref{eq:PoI:1}
or
\eqref{eq:PoI:2}. For instance, when the trace of $\partial_y \partial_{\mu} f_0(x,m,y)$
(resp. 
$\partial_y \partial_{\mu} g(x,m,y)$) 
is positive, which is  the case 
in the 
example addressed in 
Proposition
\ref{prop:convex}
provided the convexity condition therein holds in a strong sense. 
\item As far as 
equation \eqref{eq:PoI-dfdg-defY} is concerned, 
we notice that, thanks to the standing assumption on 
the optimal feedback, the optimal trajectory is progressively-measurable 
with respect to the Brownian filtration (possibly augmented with the 
$\sigma$-algebra generated by the initial condition). In particular $(Y_t)_{t \in [0,T]}$
has continuous trajectories and can be represented in the form of a backward SDE.
\end{enumerate}
\end{remark}

\begin{proof}[Proof of Proposition~\ref{prop:PoI}]
We start by noticing that,
for a new bounded control process $\bbeta$
and for any $\varepsilon \in \RR^{{d}},$ 
\begin{equation*}
\begin{split}
    &\frac{d}{d\varepsilon}
    \Bigl[ 
    J \bigl( {\balpha}^{\rMKV} + \varepsilon {\bbeta} ; 
    \mu^{
    {\balpha}^{\rMKV} + \varepsilon
    \bbeta}
    \bigr)
    -
    J \bigl( {\balpha}^{\rMKV} + \varepsilon {\bbeta} ; \mu^{{\rMKV} }\bigr) 
    \Bigr]_{\Big\vert \varepsilon=0}
    \\
    &={\mathbb E}
    \widetilde {\mathbb E}
    \biggl[
    \int_0^T
    \partial_\mu f_0
    \bigl( 
    X_t^{\rMKV}, 
    \mu_t^{\rMKV}
    \bigr)\bigl( \widetilde X_t^{\rMKV}\bigr) \cdot \partial \widetilde X_t^{\rMKV}
    dt 
    \\
&\hspace{30pt}    +
    \partial_\mu 
        g \bigl( X_T^{\rMKV}, \mu_T^{\rMKV} \bigr) 
        \bigl( \widetilde X_T^{\rMKV}\bigr)
        \cdot \partial \widetilde X_T^{\rMKV}
        \biggr],
\end{split}
\end{equation*}
where 
\begin{equation*}
    \partial \tilde{X}_t^{\rMKV} = \int_0^t \tilde\beta_s ds, \qquad t\in[0,T],
\end{equation*}
with the same convention as before that $\tilde\bbeta$ is an independent copy of $\bbeta$ (see Section~\ref{sec:notation}).
We can indeed compute the above derivative following the steps of the proof of the stochastic maximum principle (see \cite[Chapters 3 \& 6]{CarmonaDelarue_book_I} for example).

Since $\bbeta$ is an arbitrary progressively measurable process, we can invoke 
Fubini's theorem to rewrite the above identity as 
\begin{equation}
\label{eq:proof-PoI-dJ-dfdg}
\begin{split}
    &\frac{d}{d\varepsilon}
    \Bigl[ 
    J \bigl( {\balpha}^{\rMKV} + \varepsilon {\bbeta} ; 
    \mu^{
    {\balpha}^{\rMKV} + \varepsilon
    \bbeta}
    \bigr)
    -
    J \bigl( {\balpha}^{\rMKV} + \varepsilon {\bbeta} ; \mu^{{\rMKV} }\bigr) 
    \Bigr]_{\Big\vert \varepsilon=0}
    \\
    &={\mathbb E}
    \widetilde {\mathbb E}
    \biggl[
    \int_0^T
    \beta_t 
    \biggl( 
        \partial_\mu 
        g \bigl( \widetilde X_T^{\rMKV}, \mu_T^{\rMKV} \bigr) 
        \bigl( X_T^{\rMKV}\bigr)
+
    \int_t^T \partial_\mu f_0
    \bigl( 
    \widetilde X_s^{\rMKV}, 
    \mu_s^{\rMKV}
    \bigr)\bigl( X_s^{\rMKV}\bigr) 
    ds
  \biggr)        \biggr]
    \\
    &= \EE \int_0^T Y_s \cdot \beta_s ds,
\end{split}
\end{equation}
with 
$(Y_t)_{t \in [0,T]}$ as in 
\eqref{eq:PoI-dfdg-defY}.

Next, we prove each point in the statement. 
\vskip 4pt

\noindent{\bf First point. }
Following Remark~\ref{rem:PoI0}, assuming that ${\rm PoI}=0$, we deduce that 
${\balpha}^{\rMKV}$
is an MFG equilibrium.
Hence, 
for a new control process $\bbeta$
and for any $\varepsilon \in \RR^{{d}},$ using the fact that ${\balpha}^{\rMKV}$ is both an optimal  control 
and an equilibrium control, we deduce:
\begin{equation*}
\begin{split}
    &J \bigl( {\balpha}^{\rMKV} + \varepsilon {\bbeta} ; 
    {\boldsymbol \mu}^{
    {\balpha}^{\rMKV} + \varepsilon
    \bbeta}
    \bigr) \geq 
    J \bigl( {\balpha}^{\rMKV} ; 
    {\boldsymbol \mu}^{\rMKV}
    \bigr),
    \\
    &J \bigl( {\balpha}^{\rMKV} + \varepsilon {\bbeta} ; 
    {\boldsymbol \mu}^{
    {\rMKV} }
    \bigr) \geq 
    J \bigl( {\balpha}^{\rMKV} ; 
    {\boldsymbol \mu}^{\rMKV}
    \bigr).
\end{split}
\end{equation*}
Since ${\balpha}^{\rMKV}$ is an optimizer for the MFC problem, $\frac{d}{d\varepsilon} J \bigl( {\balpha}^{\rMKV} + \varepsilon {\bbeta} ; \mu^{{\balpha}^{\rMKV} + \varepsilon\bbeta}\bigr) = 0$. Moreover, it is also an MFG equilibrium, and hence it is a best response for a player facing $\mu^{\rMKV}$; therefore $\frac{d}{d\varepsilon} J \bigl( {\balpha}^{\rMKV} + \varepsilon {\bbeta} ; \mu^{{\balpha}^{\rMKV}}\bigr) = 0$.
Hence, we get: 
\begin{equation*}
    \frac{d}{d\varepsilon}
    \Bigl[ 
    J \bigl( {\balpha}^{\rMKV} + \varepsilon {\bbeta} ; 
    \mu^{
    {\balpha}^{\rMKV} + \varepsilon
    \bbeta}
    \bigr)
    -
    J \bigl( {\balpha}^{\rMKV} + \varepsilon {\bbeta} ; \mu^{{\rMKV} }\bigr) 
    \Bigr]_{\Big\vert \varepsilon=0}
    =0.
\end{equation*}
By the preliminary analysis, we obtain
that 
$(Y_t)_{t \in [0,T]} \equiv 0$, from which we deduce that 
\begin{equation*}
\tilde{\mathbb E} 
\Bigl[
\partial_\mu g\bigl(\widetilde X_T^{\rMKV},\mu_T^{\rMKV}\bigr)\bigl( 
X_T^{\rMKV} \bigr) \Bigr] = Y_T=0,
\end{equation*}
with probability 1, and similarly, for any $t \in [0,T]$, 
\begin{equation*}
\widetilde{\mathbb E} \Bigl[
\partial_\mu f\bigl(\widetilde X_t^{\rMKV},\mu_t^{\rMKV}\bigr)\bigl( X_t^{\rMKV} \bigr) \Bigr]=0, 
\end{equation*}
with probability 1. 

Now, 
since the feedback function 
$\alpha$ is at most of linear growth in 
$x$, we deduce that, for any 
$t>0$, $\sup_{0 \leq s \leq t} \vert X_s\vert
\leq C_t( 1+
\sup_{0 \leq s \leq t} \vert B_s\vert)$, 
for a deterministic constant $C_t$ depending on 
$t$. By Girsanov theorem, we deduce that
the marginal laws of 
$\tilde{X}^{\rMKV}$ have a full support for small positive time
and hence for any positive time by Markov property. 
This completes the proof of \eqref{eq:PoI:2} and also of \eqref{eq:PoI:1}
when $t>0$. Identity \eqref{eq:PoI:1} is obtained for $t=0$ by letting $t$ 
tend to $0$, using a continuity argument.
\vskip6pt
\noindent{\bf Second point. } 
As before, since ${\balpha}^{\rMKV}$ is an optimizer for the MFC problem, we have $\frac{d}{d\varepsilon} J \bigl( {\balpha}^{\rMKV} + \varepsilon {\bbeta} ; \mu^{{\balpha}^{\rMKV} + \varepsilon\bbeta}\bigr) = 0$. 
Going back to~\eqref{eq:proof-PoI-dJ-dfdg}, we obtain:
\begin{equation*}
\begin{split}
    &\frac{d}{d\varepsilon}
    \Bigl[ 
    J \bigl( {\balpha}^{\rMKV} + \varepsilon {\bbeta} ; 
    \mu^{
    {\balpha}^{\rMKV} + \varepsilon
    \bbeta}
    \bigr)
    -
    J \bigl( {\balpha}^{\rMKV} + \varepsilon {\bbeta} ; \mu^{{\rMKV} }\bigr) 
    \Bigr]_{\Big\vert \varepsilon=0}
    = \EE \int_0^T Y_s \cdot \beta_s ds.
\end{split}
\end{equation*}

Hence:
\begin{equation*}
\begin{split}
    &\frac{d}{d\varepsilon}
    \Bigl[ 
    J \bigl( {\balpha}^{\rMKV} + \varepsilon {\bbeta} ; \mu^{{\rMKV} }\bigr) 
    \Bigr]_{\Big\vert \varepsilon=0}
    = -\EE \int_0^T Y_s \cdot \beta_s ds.
\end{split}
\end{equation*}
Moreover, 
returning 
to 
\eqref{eq:proof-PoI-dJ-dfdg}
and writing the analogue of the formula but at some $\varepsilon >0$, 
one can see (using the Lipschitz property of the derivatives of $f_0$ and $g$)
that, for $\varepsilon \in [0,1]$
\begin{equation*} 
    \Bigl\vert \frac{d}{d \varepsilon} 
    \Bigl[ J\bigl({\boldsymbol \alpha}^{\rm MFC} + \varepsilon {\boldsymbol \beta};\mu^{{\rm MFC}} \bigr)\Bigr]
    -
    \frac{d}{d\varepsilon}
    \Bigl[ 
    J \bigl( {\balpha}^{\rMKV} + \varepsilon {\bbeta} ; \mu^{{\rMKV} }\bigr) 
    \Bigr]_{\Big\vert \varepsilon=0}
    \Bigr\vert \leq C {\mathbb E} \int_0^T \vert \beta_t \vert^2 dt. 
\end{equation*} 
Thus, 
\begin{equation*} 
    \begin{split}
    &J\bigl({\boldsymbol \alpha}^{\rm MFC} +  {\boldsymbol \beta};\mu^{{\rm MFC}} \bigr)
    \leq J\bigl({\boldsymbol \alpha}^{\rm MFC} ;\mu^{{\rm MFC}} \bigr)
     -
    {\mathbb E} \int_0^T Y_t \cdot \beta_t dt
    +
    C {\mathbb E} \int_0^T \vert \beta_t \vert^2 dt.
    \end{split}
\end{equation*} 
We then minimize over ${\boldsymbol \beta}$ and get, for $\beta_t=\tfrac1{2C} Y_t$ for all $t \in [0,T]$,
\begin{equation*} 
\begin{split}
    &J\bigl({\boldsymbol \alpha}^{\rm MFC} +  {\boldsymbol \beta};\mu^{{\rm MFC}}
    \bigr)
    \leq J\bigl({\boldsymbol \alpha}^{\rm MFC} ;\mu^{{\rm MFC}}
    \bigr)
     - \tfrac1{4C} 
    {\mathbb E} \int_0^T  \vert Y_t \vert^2 dt,
    \end{split} 
\end{equation*} 
which says that 
$
    {\rm PoI} 
    \geq \tfrac1{4C} 
    {\mathbb E} \int_0^T  \vert Y_t \vert^2 dt. 
$
\color{black}  
\end{proof}

\subsubsection{Subpopulation Deviation: $p$-Partial Mean Field Games}
\label{subsubsec:MFCtoMFG_completeinfo_pmixed}

The idea is now to look at a mixed population where, instead of a single player deviating, a fixed proportion, say $p$, of the population is allowed to deviate from the social planner's optimal prescription to follow their own best interests. In this new situation, we assume that a proportion $(1-p)$ of the population still follows the control originally prescribed by the social planner (MFC) when everyone is assumed to be cooperative. The remaining proportion $p$ of players deviates to minimize their individual cost, and the population reaches an equilibrium, which we call a $p$-partial mean field equilibrium. This concept is different from the notion of $\lambda$-interpolated mean field equilibrium introduced in Section~\ref{sec:lambda-interp-def}, as will become clear below.

\paragraph{Definition and First Results}

Throughout this subsection, we consider the same dynamics as in 
\eqref{fo:state}, with the same cost as in 
\eqref{fo:J_of_alpha} and 
\eqref{fo:running_cost}. We start with the following definition, which corresponds to a game with a population consisting of a proportion $p$ of non-cooperative players and a proportion $(1-p)$ of players who continue to use the original control prescribed by the social planner:

\begin{definition}[$p$-Partial Mean Field Equilibrium]\label{def:p-partial-eq}
We call $(\hat{\boldsymbol \alpha}^p, p \hat {\boldsymbol \mu}^p +(1-p) {\boldsymbol \mu}^{\rMKV})$ a $p$-partial mean field equilibrium if
\begin{itemize}
    \item[i.] ${\boldsymbol \mu}^{\rMKV}$ is the flow of distributions resulting from the (original) social planner's optimization;
    \item[ii.] $\hat\mu^p_t=\cL(X^{\hat\balpha^p}_t)$ for all $t\in[0,T]$;
    \item[iii.]$\hat\balpha^p$ is the minimizer of $J(\balpha; p\hat{\boldsymbol \mu}^p+(1-p){\boldsymbol \mu}^{\rMKV})$ given fixed population distribution flow $p\hat{\boldsymbol \mu}^p+(1-p){\boldsymbol \mu}^{\rMKV}$. 
\end{itemize}
\end{definition}
In this model, given any fixed $p \in [0,1]$, given a $p$-partial mean field equilibrium $(\hat\balpha^p; p{\boldsymbol \mu}^{\hat\balpha^p}+(1-p){\boldsymbol \mu}^{\rMKV}\big)$, we define the two costs, $\hat J_p$ and $J^*_p$, as follows: 
\begin{equation}
    \begin{aligned}
        \hat J_p &:= J\big(\hat\balpha^p; p{\boldsymbol \mu}^{\hat\balpha^p}+(1-p){\boldsymbol \mu}^{\rMKV}\big)\\[2mm]
        J^*_p &:=J\big(\balpha^{\rMKV}; p{\boldsymbol \mu}^{\hat\balpha^p}+(1-p){\boldsymbol \mu}^{\rMKV}\big).
    \end{aligned}
\end{equation}
Intuitively, the population consists of two sub-populations: the players who keep using the control prescribed by the social planner in the original MFC, and the players who are optimizing their individual cost in a non-cooperative way. The first sub-population represents a proportion $p$ of the total, and the second one a proportion $(1-p)$.    
The quantity $J^*_p$ is the cost to an individual player that still follows the social planner, while $\hat{J}_p$ is the cost to a non-cooperative player. Notice that this problem is different from the $\lambda$-interpolated mean field game problem introduced in Section~\ref{sec:lambda-interp-def}, where each player's cost involves a mixture of individual and collective distributions. In a $p$-partial mean field equilibrium, each player is either following the MFC optimal control or optimizing her individual cost. Furthermore, note that the $(1-p)$ proportion of players following the social planner's recommendation use a control which is socially optimal only in the case where $p=0$. Otherwise, their control is sub-optimal. 
By definition, the cost of MFC, $J^*$, is equal to $J^*_0$ and the cost of the representative player in a regular MFG, $\hat J$, is equal to $\hat J_1$. 
\begin{remark}
\label{rem:hat:star:hat}
    Notice that, $J^*_0 - \hat J_0$ is equal to the PoI since $J^*_0$ is the cost of MFC and $\hat J_0$ is the cost of a single player using their best response given that everyone else still follows the control prescribed by the social planner. Moreover, 
\begin{equation*}
\hat{J}_0 \leq J^* = J^*_0 \leq \hat{J}_1.
\end{equation*}
\end{remark}

It must be clear to the reader that a $p$-partial mean field equilibrium is an equilibrium of a standard MFG whose coefficients depend on 
$\bmu^{\rMKV}$. Indeed, one can consider
the MFG corresponding to the running cost and the terminal cost (instead of $f_0$ and $g$ respectively):
\[
 (x,\mu) \mapsto f_0^{p}( x, p \mu+(1-p) \mu^{\rMKV} ), \qquad 
 (x,\mu) \mapsto g^{p} \bigl( x, p \mu + (1-p) \mu^{\rMKV} ),
\]
where $\mu^{\rMKV}$ is fixed and known from the players who are optimizing their individual cost. If we denote by $(\hat\balpha,\hat\bmu)$ the MFG equilibrium resulting from these costs functions, then the $p$-partial mean field equilibrium as defined in Definition~\ref{def:p-partial-eq} is: $(\hat{\boldsymbol \alpha}^p, p \hat {\boldsymbol \mu}^p +(1-p) {\boldsymbol \mu}^{\rMKV}) = (\hat\balpha, p \hat {\bmu} +(1-p) {\boldsymbol \mu}^{\rMKV})$.

In this regard, existence of a solution can be proved by means of standard fixed point arguments (without uniqueness, see for instance \cite[Section 4.4.1]{CarmonaDelarue_book_I}).

\begin{proposition} 
\label{prop:existence:p-mixed:eq}
In addition to Assumption \ref{assumption:initial-model}, assume 
that Assumption~\ref{assumption:SMP:MFG}
is in force. 
Then, 
for any $p \in [0,1]$, there exists at least one 
$p$-partial mean field equilibrium. 

Moreover, the $p$-partial mean field equilibrium 
is unique if $p$ is small enough (i.e., $p<\varepsilon_0$ for some $\varepsilon_0>0$ independent of the choice of 
the initial condition of $\mu^{\rMKV}$). 
\end{proposition}

\begin{proof}[Proof of Proposition~\ref{prop:existence:p-mixed:eq}]
Existence in the general case follows from existence results for MFGs. Under item 
$(i)$ in Assumption~\ref{assumption:SMP:MFG}, this follows from 
\cite[Theorem 4.32]{CarmonaDelarue_book_I}. 
Under item $(ii)$, it follows from 
\cite[Theorem 4.53]{CarmonaDelarue_book_I}.

Uniqueness (for $p$ small enough) is proven by means of a standard contraction argument. Assume indeed that 
$\hat{\boldsymbol \alpha}^p$ and $\hat{\boldsymbol \beta}^p$
are two $p$-partial mean field equilibria and denote by 
${\boldsymbol \mu}^{{\hat{\boldsymbol \alpha}}^p}$
and 
${\boldsymbol \mu}^{{\hat{\boldsymbol \beta}}^p}$
the corresponding flows of probability measures. 
Then, using the convex of structure of the Lagrangian, we obtain 
\begin{equation*}
\begin{split}
    \frac12 {\mathbb E}
    \int_0^T \vert \hat{\alpha}_t^{p}
    - \hat{\beta}_t^{p} \vert^2 
    dt &\leq J\bigl( \hat {\boldsymbol \alpha}^{p};
    p {\boldsymbol \mu}^{\hat{\boldsymbol \beta}^p}
    + 
    (1-p) {\boldsymbol \mu}^{\rMKV}
    \bigr) 
    - 
     J\bigl( \hat {\boldsymbol \beta}^{p};
    p{\boldsymbol \mu}^{\hat{\boldsymbol \beta}^p}
    + 
    (1-p) {\boldsymbol \mu}^{\rMKV}
    \bigr)  
    \\
    \frac12 {\mathbb E}
    \int_0^T \vert \hat{\alpha}_t^{p}
    - \hat{\beta}_t^{p} \vert^2 
    dt
    &\leq J\bigl( \hat{{\boldsymbol \beta}^{p}};
    p {\boldsymbol \mu}^{\hat{\boldsymbol \alpha}^p}
    + 
    (1-p) {\boldsymbol \mu}^{\rMKV}
    \bigr)
    - J\bigl( \hat{{\boldsymbol \alpha}^{p}};
    p {\boldsymbol \mu}^{\hat{\boldsymbol \alpha}^p}
    + 
    (1-p) {\boldsymbol \mu}^{\rMKV}
    \bigr). 
\end{split}
\end{equation*}
Therefore, adding the two lines, we obtain 
\begin{align}
    &{\mathbb E}
    \int_0^T \vert \hat{\alpha}_t^{p}
    - \hat{\beta}_t^{p} \vert^2 
    dt
    \label{eq:monotonicity:1}
    \\
    &\leq 
    p \int_{{\mathbb R}^d}
    \Bigl[
    g \bigl( x, 
    p \mu_T^{\hat{\boldsymbol \beta}^p}
    + 
    (1-p) \mu_T^{\rMKV}
    \bigr) 
    -
    g\bigl( x, 
    p \mu_T^{\hat{\boldsymbol \alpha}^p}
    + 
    (1-p) \mu_T^{\rMKV}
    \bigr) 
    \Bigr] 
    d \bigl(  \mu_T^{\hat{\boldsymbol \alpha}^p}
    - 
    \mu_T^{\hat{\boldsymbol \beta}^p}
    \bigr)(x) \nonumber
    \\
    &\hspace{15pt} + p \int_{0}^T \int_{{\mathbb R}^d}
    \Bigl[
    f_0\bigl( x, 
    p \mu_t^{\hat{\boldsymbol \beta}^p}
    + 
    (1-p) \mu_t^{\rMKV}
    \bigr) 
    -
    f_0\bigl( x, 
    p \mu_t^{\hat{\boldsymbol \alpha}^p}
    + 
    (1-p) \mu_t^{\rMKV}
    \bigr) 
    \Bigr] 
    d \bigl(  \mu_t^{\hat{\boldsymbol \alpha}^p}
    - 
    \mu_t^{\hat{\boldsymbol \beta}^p}
    \bigr)(x) dt. \nonumber
\end{align}
Here, 
we call $\varepsilon$ a 
Bernoulli($p$) random
variable, that is independent of the 
$\sigma$-field ${\mathcal F}_T$
(extending the probability space, we can
always find such a variable). Then, 
\begin{align}
& \int_{{\mathbb R}^d}
    \Bigl[
    g \bigl( x, 
    p \mu_T^{\hat{\boldsymbol \beta}^p}
    + 
    (1-p) \mu_T^{\rMKV}
    \bigr) 
    -
    g\bigl( x, 
    p \mu_T^{\hat{\boldsymbol \alpha}^p}
    + 
    (1-p) \mu_T^{\rMKV}
    \bigr) 
    \Bigr] 
    d \bigl(  \mu_T^{\hat{\boldsymbol \alpha}^p}
    - 
    \mu_T^{\hat{\boldsymbol \beta}^p}
    \bigr)(x) 
    \nonumber
    \\
   &= \int_{{\mathbb R}^d}  
   \Bigl[ g\Bigl(x , 
    {\mathcal L} 
    \bigl(
    \varepsilon X_T^{\hat{\boldsymbol \alpha}^p}
    + (1-\varepsilon) X_T^{\rm MFC} \bigr) \Bigr) 
    -
    g\Bigl( x , {\mathcal L} \Bigl( \varepsilon X_T^{\hat{\boldsymbol \alpha}^p}
    + (1-\varepsilon) X_T^{\rm MFC} \Bigr)\Bigr) \Bigr]    d \bigl(  \mu_T^{\hat{\boldsymbol \alpha}^p}
    - 
    \mu_T^{\hat{\boldsymbol \beta}^p}
    \bigr)(x)\nonumber
       \\
    &= {\mathbb E}
    \int_0^1
    \Bigl[ \partial_\mu g\Bigl( x , {\mathcal L} \bigl(
     \lambda \varepsilon X_T^{\hat{\boldsymbol \alpha}^p}
     +
  (1-   \lambda) \varepsilon X_T^{\hat{\boldsymbol \beta}^p}
    + (1-\varepsilon) X_T^{\rm MFC} \bigr), \nonumber
    \\
&\hspace{30pt}      \lambda \varepsilon X_T^{\hat{\boldsymbol \alpha}^p}
     +
  (1-   \lambda) \varepsilon X_T^{\hat{\boldsymbol \beta}^p}
    + (1-\varepsilon) X_T^{\rm MFC}\Bigr)
  \cdot  
\Bigl( 
 \varepsilon X_T^{\hat{\boldsymbol \alpha}^p}
 -
  \varepsilon X_T^{\hat{\boldsymbol \beta}^p}
  \Bigr)
    \Bigr] 
      d \bigl(  \mu_T^{\hat{\boldsymbol \alpha}^p}
    - 
    \mu_T^{\hat{\boldsymbol \beta}^p}
    \bigr)(x)\biggr\} d \lambda \nonumber
    \\
    &= {\mathbb E} \widetilde{\mathbb E} 
    \Bigl[ 
    \partial_\mu g\Bigl( \widetilde X_T^{\hat{\boldsymbol \alpha}^p} , {\mathcal L} \bigl(
     \lambda \varepsilon X_T^{\hat{\boldsymbol \alpha}^p}
     +
  (1-   \lambda) \varepsilon X_T^{\hat{\boldsymbol \beta}^p}
    + (1-\varepsilon) X_T^{\rm MFC} \bigr), \label{eq:second:order:g}
    \\
&\hspace{30pt}      \lambda \varepsilon X_T^{\hat{\boldsymbol \alpha}^p}
     +
  (1-   \lambda) \varepsilon X_T^{\hat{\boldsymbol \beta}^p}
    + (1-\varepsilon) X_T^{\rm MFC}\Bigr)
  \cdot  
\Bigl( 
 \varepsilon X_T^{\hat{\boldsymbol \alpha}^p}
 -
  \varepsilon X_T^{\hat{\boldsymbol \beta}^p}
  \Bigr)
\Bigr] \nonumber
\\
&\hspace{15pt}  - {\mathbb E} \widetilde{\mathbb E} 
    \Bigl[ 
    \partial_\mu g\Bigl( \widetilde X_T^{\hat{\boldsymbol \beta}^p} , {\mathcal L} \bigl(
     \lambda \varepsilon X_T^{\hat{\boldsymbol \alpha}^p}
     +
  (1-   \lambda) \varepsilon X_T^{\hat{\boldsymbol \beta}^p}
    + (1-\varepsilon) X_T^{\rm MFC} \bigr), \nonumber
    \\
&\hspace{30pt}      \lambda \varepsilon X_T^{\hat{\boldsymbol \alpha}^p}
     +
  (1-   \lambda) \varepsilon X_T^{\hat{\boldsymbol \beta}^p}
    + (1-\varepsilon) X_T^{\rm MFC}\Bigr)
  \cdot  
\Bigl( 
 \varepsilon X_T^{\hat{\boldsymbol \alpha}^p}
 -
  \varepsilon X_T^{\hat{\boldsymbol \beta}^p}
  \Bigr)
\Bigr]. \nonumber
   \end{align}
   Proceeding similarly with $f_0$ and using item $(ii)$ in 
   Assumption 
   \ref{assumption:initial-model}, we obtain 
\begin{equation*}
\begin{split}
    {\mathbb E}
    \int_0^T \vert \hat{\alpha}_t^{p}
    - \hat{\beta}_t^{p} \vert^2 
    dt
    &\leq C_T p {\mathbb E} \bigl[ \sup_{0 \leq t \leq T} \vert X_t^{\hat{\balpha}^p} - X_t^{\hat\bbeta_t^p} \vert^2\bigr],
\end{split}
\end{equation*}
where the constant $C_T$ depends on $T$ but is independent of $p$. 
Observing that 
\begin{equation*}
    {\mathbb E} \bigl[ \sup_{0 \leq t \leq T} \vert X_t^{\hat\balpha^p} - X_t^{\hat\bbeta_t^p} \vert^2\bigr]
    \leq 
    C_T {\mathbb E}
    \int_0^T \vert \hat{\alpha}_t^{p}
    - \hat{\beta}_t^{p} \vert^2 
    dt,
\end{equation*}
we complete the proof. 
\end{proof}

Assuming in the rest of this subsection that the conclusion of 
Proposition 
\ref{prop:existence:p-mixed:eq} holds true,
we now provide sufficient conditions to get bounds on the individual costs. The main objective is to compare 
$\hat{J}_p$ for different values of $p$, it being understood that, for a given $p \in (0,1]$, 
$p$-partial mean field equilibria may not be unique.
When $p=0$, there is exactly one $0$-partial mean field equilibrium 
(under the same assumptions on
$f_0$ and $g$ as in the statement of Proposition \ref{prop:existence:p-mixed:eq}): this equilibrium is just the 
minimizer of $J(\cdot \, ; \mu^{\rMKV})$. 

The following concavity condition says that the $0$-partial mean field equilibrium provides the best (individual) cost 
(among all the $p$-partial mean field equilibria): 
\begin{proposition}
\label{prop:relationship_costs}
Assume that $f_0$ and $g$ are 
displacement
concave in the measure argument (with the same meaning for displacement concavity as in 
Assumption 
\ref{assumption:existence:uniqueness:MFC}). Let $p \in [0,1]$ and let  
$(\hat{{\boldsymbol \alpha}}^{p},p\hat{\boldsymbol\mu}^{p}+(1-p)\boldsymbol{\mu}^{\rm MFC})$ be a $p$-partial mean field equilibrium. Denote  $\hat{J}_p = J(\hat{{\boldsymbol \alpha}}^{p},p\hat{\boldsymbol\mu}^{p}+(1-p)\boldsymbol{\mu}^{\rm MFC})$. Then:
\begin{equation*}
\hat{J}_0
\leq 
(1-p) \hat{J}_0 + p J^{*}
\leq \hat{J}_p \leq J^{*}_p. 
\end{equation*}
\end{proposition}

\begin{proof}
We write 
\begin{equation*}
\begin{split}
\hat{J}_p&=J\bigl(\hat{\boldsymbol \alpha}^{p};p \hat{\boldsymbol \mu}^p + (1-p) {\boldsymbol \mu}^{\rMKV} \bigr)
\\
&= \frac12 {\mathbb E}
\int_0^T \vert \hat{\alpha}_t^{p} \vert^2 dt 
+ 
{\mathbb E}
\biggl[ 
\int_0^T f_0\bigl(X_t^{\hat{\balpha}^p},
p \tilde{\mu}_t^{p}
+
(1-p) \mu_t^{\rMKV}  
\bigr) dt
+ g \bigl(X_T^{\hat{\balpha}^p},
p 
\tilde{\mu}_T^{p}
+
(1-p) \mu_T^{\rMKV} 
\bigr) 
\biggr]
\\
&\geq \frac12 
{\mathbb E}
\int_0^T \vert \hat{\alpha}_t^{p} \vert^2 dt 
+
 p
 {\mathbb E}
\biggl[ 
\int_0^T f_0\bigl(X_t^{\hat{\balpha}^p}, 
\tilde{\mu}^p_t
\bigr) dt
+ g\bigl(X_T^{\hat{\balpha}^p},\tilde{\mu}^p_T \bigr)\biggr]
\\
&\hspace{15pt}
+ 
(1-p) {\mathbb E}
\biggl[ 
\int_0^T f_0\bigl(X_t^{\hat{\balpha}^p},  \mu_t^{\rMKV}  \bigr) dt
+ g \bigl(X_T^{\hat{\balpha}^p},
\mu_T^{\rMKV}
\bigr)
\biggr]
\\
&\geq p J^* + 
(1-p) J\bigl(\hat{\boldsymbol \alpha}^{0};{\boldsymbol \mu}^{\rMKV}\bigr)\\
&= p J^* + 
(1-p) \hat{J}_0,
\end{split}
\end{equation*}
where the first inequality comes from the concavity of $f_0$ and $g$ in the measure argument and the last inequality comes from the fact that $J^*$ is the social optimum and $\hat{\balpha}^0$ is the minimizer of $J(\cdot; \bmu^{\rMKV})$.
Since $J^* \geq \hat{J}_0 = J(\hat{\boldsymbol \alpha}^{0};{\boldsymbol \mu}^{\rMKV})$
and $J^{*}_p  \geq \hat{J}_p $, 
the proof of the inequality is completed.
\end{proof}

\begin{remark}
The interpretation of this result is as follows: deviating (i.e., non-cooperative) players will have a lower cost
when there are only a few of them, meaning their proportion is zero. 
It would be interesting to push the analysis further and to see how the cost of a deviating player evolves, when the 
number of players is finite (say $N$) but large, and the number of deviating players is $N^\beta$ for some 
$\beta \in (0,1)$. 
\end{remark}

\paragraph{Monotone Interactions}

In this subsection, we further assume (in addition to the conclusion of Proposition 
\ref{prop:existence:p-mixed:eq}) that the two running and terminal costs $f_0$ and 
$g$ are monotone in the sense of Lasry and Lions, 
see \eqref{eq:def:monotonicity}. The purpose is to address some of the properties of the $p$-partial mean field equilibria in this setting. 

\begin{proposition}
Let ${\boldsymbol \mu}^{\rMKV}$ be a flow of distributions resulting from the MFC problem (i.e., social planner's optimization). Then, for 
any $p \in [0,1]$, there exists a unique $p$-partial mean field equilibrium. 
\end{proposition}

\begin{proof}
Recalling that a $p$-partial mean field equilibrium is a solution to the MFG driven by the two functions 
$(x,m) \mapsto f_0( x, (1-p) \mu^{\rMKV} + p m)$ 
and
$(x,m) \mapsto g \bigl( x, (1-p) \mu^{\rMKV} + pm)$, 
which are monotone, we invoke Lasry-Lions uniqueness criterion to guarantee that 
the $p$-partial mean field equilibrium is indeed unique. See e.g.~\cite{lasry2007mean} and ~\cite[Section 3.4]{CarmonaDelarue_book_I} for standard proofs.
\end{proof}

Monotonicity is known to supply mean field games with a strong form of stability. To wit, we have 
the following Lipschitz property: 

\begin{proposition}
\label{prop:lipschitz:property:p}
Assume that $f_0$ and $g$ satisfy the Lasry-Lions monotonicity condition. Then, 
the path 
$$
 [0,1] \ni p \mapsto \bigl(X_t^{\hat{\balpha}^p} \bigr)_{0 \leq t \leq T}
 \in L^2\bigl( \Omega;{\mathcal C}([0,T];{\mathbb R}^d)\bigr)
$$
is Lipschitz continuous, 
where the space ${\mathcal C}([0,T];{\mathbb R}^d)$ in the right-hand side is equipped with the supremum norm. 

In particular, the two cost mappings  
$$[0,1] \ni p \mapsto \hat{J}_p, \quad 
[0,1] \ni p \mapsto J_p^*$$
are also continuous. 
\end{proposition}
For the proof of Proposition~\ref{prop:lipschitz:property:p}, please refer to Appendix~\ref{app:p_mixed}.

\begin{remark}
For sure,
the reader may wonder about stronger regularity properties of the
mapping $p \mapsto (X_t^{\hat{\balpha}^p})_{0 \leq t \leq T}$, like differentiability.
We can reasonably guess that such properties indeed hold true when the coefficients 
$f_0$ and $g$ are sufficiently smooth (provided that the latter two satisfy the Lasry-Lions monotonicity condition). 
The proof would consist in a linearization procedure, very similar to the ones
used in \cite{CardaliaguetDelarueLasryLions,CarmonaDelarue_book_II,CCD} for the analysis of the master equation to 
monotone mean field games. Since the details would lead to a substantial increase in the length of the article, we prefer not to deal with the issue carefully, but we emphasize that we do not see a major obstacle. 
\end{remark}

The following corollary is important for explaining the free rider phenomenon.
\begin{corollary}
\label{cor:intermediate:value}
Under the assumption of 
Proposition 
\ref{prop:lipschitz:property:p}, there exists 
$p^* \in [0,1]$ such that 
$\hat{J}_{p^*} = J^*$.
\end{corollary}

\begin{proof}
The proof follows from the intermediate value theorem
and the inequality $\hat{J}_0 \leq J^* \leq \hat{J}_1$.
\end{proof}

Here is an obvious application of Corollary 
\ref{cor:intermediate:value}: If $\hat{J}_0 < J^*$
(i.e., the PoI is strictly positive), then
we can 
choose $p^*$ as the smallest value of 
$p$ when $\hat{J}_{p^*}=J^*$. For $p \in [0,p^*)$, individualists (i.e., non-cooperative players) take advantage of the cooperative players. We provide some examples of this \textit{free ride} phenomenon in the next subsection. 
Of course, it should be clear that the monotonicity condition here is used to guarantee continuity of the cost w.r.t. the parameter $p$. There may be cases when the cost is continuous without monotonicity (e.g., continuity is certainly true when $p$ is small, as equilibria can be constructed by Picard approximations).

\paragraph{Example}

\label{subsec:example_terminal_mean_interaction}
We compute and plot the variations of $\hat J_p$ and $J^*_p$ with respect to  $p$ in order to illustrate the relationships between mean field control, mean field game, and $p$-partial mean field equilibria in a toy model with one-dimensional state and control. Our model is similar to the examples treated in~\cite{cdl_mfgVSmfc}, \cite[p. 278]{CarmonaDelarue_book_I} and~\cite[Chapter 6]{Bensoussan_Book}
and as follows.
We consider the following dynamics
\begin{equation*}
    dX_{t} = \alpha_{t} dt + dB_{t},
\end{equation*}
and the following cost:
\begin{equation*}
J^{\boldsymbol{\mu}}(\boldsymbol{\alpha}) = {\mathbb E}
\biggl[ \int_{0}^T \frac12 ( X_{t}^2 + \alpha_{t}^2) dt + \frac12 (X_{T} - q \bar \mu_{T})^2
\biggr],
\end{equation*}
where we put a bar on top of a measure or a random variable to denote its mean. In particular $\bar\mu_T$ denotes the mean of $\mu_T$ and $\bar X_{T}$ is the expectation of the random variable $X_{T}$. Throughout this example, $X_{0}$ is some fixed $x_{0} \in \mathbb{R}$
and $q$ is a constant parameter. We chose a linear quadratic model in order to be able to compute the equilibria explicitly, and as simple a mean field interaction as possible by having it enter only the terminal cost. The proofs of Propositions~\ref{prop:ode_mfg} and \ref{prop:ode_mfc} are similar to the proofs in~\cite{Bensoussan_Book} and given in Appendices~\ref{app:proof_ode_mfg} and~\ref{app:proof_ode_mfc} for the sake of completeness. The proof of Proposition~\ref{prop:ode_p_mixed_mfg} can be found
in Appendix~\ref{app:proof_ode_p_mixed_mfg}.

\begin{proposition}
\label{prop:ode_mfg}
In the MFG problem, the equilibrium control is given as $\alpha^{\rMFG}_t = -(X_t+r_t)$ and the cost is given as  $\hat J_1 = \frac{1}{2}x_0^2 + r_0x_0+s_0$, where 
$(\boldsymbol{r},\boldsymbol{s}, \boldsymbol{\bar X})$ solve the following forward-backward ordinary differential equation (FBODE) system: 
\begin{equation}
\label{eq:ode_mfg}
\left\{
    \begin{aligned}
        \dot{r}_t    &= r_t && r_T=-q\bar X_T\\
        \dot{s}_t    &= \dfrac{1}{2} (r_t^2-1) && s_T=\dfrac{1}{2}q^2\bar X_T^2\\
        \dot{\bar X}_t &= -(\bar X_t+r_t) \qquad\qquad&&\bar X_0 = x_0.
    \end{aligned}
\right.
\end{equation}

Notice that the solution satisfies, for all $t \in [0,T]$, $r_t = -q \bar{X}_T \exp(t-T)$, and $(1-qT)\bar{X}_T = e^{-T} \bar{X}_0$. This system has a unique solution if $qT \neq 1$. 

\end{proposition}

Notice that the condition $qT \neq 1$ is satisfied in particular if $q<0$, which matches the requirements for the usual monotonicity condition. Note that this is similar to the conditions discussed in~\cite{cdl_mfgVSmfc} and~\cite[p. 278]{CarmonaDelarue_book_I}.

\begin{proposition}
\label{prop:ode_mfc}
In the MFC problem, the optimal control is given as $\alpha^{\rMKV}_t = -(X_t+r_t)$ and the cost is given as $J^*_0 = \frac{1}{2}x_0^2 + r_0x_0+s_0+(1-q)q\bar X_T^2$, where $(\boldsymbol{r},\boldsymbol{s}, \boldsymbol{\bar X})$ solve the following FBODE system:
\begin{equation}
\label{eq:ode_example_mfc}
\left\{
    \begin{aligned}
        \dot{r}_t    &= r_t && r_T=-2q\bar X_T+q^2\bar X_T\\
        \dot{s}_t    &= \dfrac{1}{2} (r_t^2-1) && s_T=\dfrac{1}{2}q^2\bar X_T^2\\
        \dot{\bar X}_t &= -(\bar X_t+r_t) \qquad\qquad&&\bar X_0 = x_0.
    \end{aligned}
\right.
\end{equation}

Notice that the solution satisfies, for all $t \in [0,T]$, $r_t = (-2q+q^2) \bar X_T \exp(t-T)$, and $(1-(2q-q^2)T)\bar{X}_T = e^{-T} \bar{X}_0$.  This system has a unique solution if $(2q-q^2)T \neq 1$. 

\end{proposition}

Here again, notice that the condition $(2q-q^2)T \neq 1$ is satisfied in particular if $q<0$, which matches the requirements for the usual monotonicity condition.

\begin{proposition}
\label{prop:ode_p_mixed_mfg}
In the $p$-partial mean field game, 
the equilibrium control is given as $\hat{\alpha}^{p}_t = -(X_t+r_t)$ and the cost is given as $\hat J_p = \frac{1}{2}x_0^2 + r_0x_0+s_0$, where $(\boldsymbol{r},\boldsymbol{s}, \boldsymbol{\bar X})$ implicitly depend on $p$ and solve the FBODE system:
\begin{equation}
\label{eq:ode_p_mixed_mfg}
\left\{
    \begin{aligned}
        \dot{r}_t    &=  r_t && r_T=-qp\bar X_T - q(1-p) \bar X_T^{\rMKV}\\
        \dot{s}_t    &= \dfrac{1}{2} (r_t^2-1) && s_T=\dfrac{1}{2}\big(q p \bar X_T + q(1-p)\bar X_T^{\rMKV}\big)^2\\
        \dot{\bar X}_t &= -(\bar X_t+r_t) \qquad\qquad&&\bar X_0 = x_0.
    \end{aligned}
\right.
\end{equation}
\end{proposition}
\begin{remark}
Consistently with Remark 
\ref{rem:hat:star:hat},
we notice that, when $p=1$, the $p$-partial mean field equilibrium characterized by \eqref{eq:ode_p_mixed_mfg} corresponds to the mean field game equilibrium characterized by \eqref{eq:ode_mfg}.  Note that when $p=0$, the solution to~\eqref{eq:ode_p_mixed_mfg} does not coincide with the solution to the MFC, \eqref{eq:ode_example_mfc}. This is because, when $p=0$, \eqref{eq:ode_p_mixed_mfg} characterizes the optimal control for an infinitesimal non-cooperative player while the rest of the population is playing the MFC optimal control.
\end{remark}

\begin{figure}
    \centering
    \includegraphics[width=\linewidth]{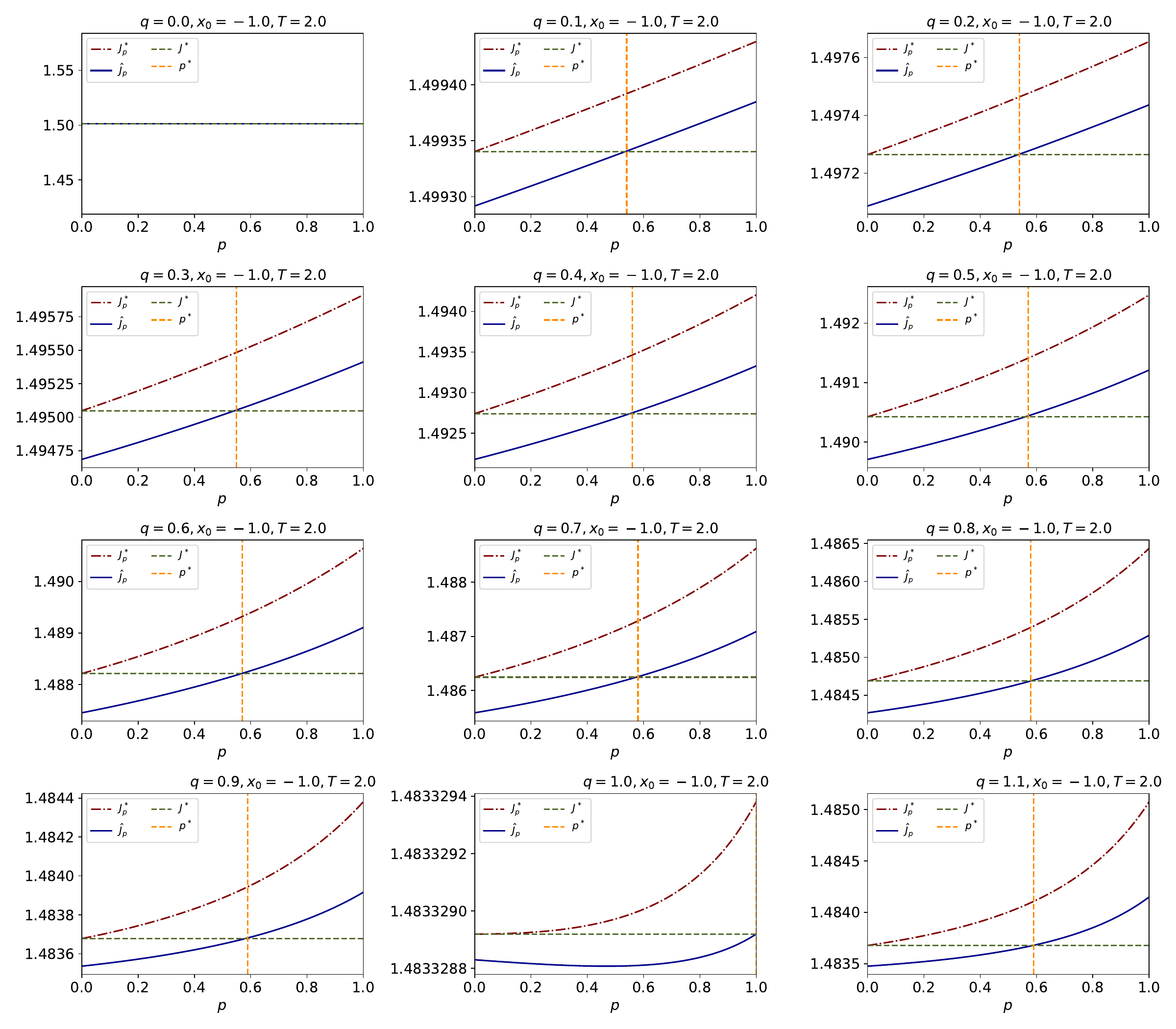}
    \caption{ Results for $q \ge 0$: Costs of the non-cooperative people and the social planner followers are shown with (full) blue and (dashed-dotted) red lines, respectively. The (dashed) green line represents the cost of the MFC i.e., social planner's optimization. The vertical (dashed) orange line represents the value of $p^*$.}
    \label{fig:p-mixed_cost_comp_over_p}
\end{figure}
 In Figures~\ref{fig:p-mixed_cost_comp_over_p} and~\ref{fig:p-mixed_cost_comp_over_p_negQ}, we see that $J_0^*$ indeed corresponds to  the mean field control cost as we expected, i.e., $J_0^*= J^*$. Furthermore, we see that $\hat J_0< J^{*}$. The PoI (see~\ref{de:PoI}) can be seen as the difference between the red and blue curves at $p=0$. We also see that as $p$ increases the cost  for continuing to follow the control prescribed by the social planner ($J^*_p$) increases. Moreover, we observe that the cost of individualists ($\hat J_p$) is lower than the original social planner cost $J^*$ (i.e., everyone were to follow the social planner) until $p^*$, which is the value of $p$ where $\hat J_p = J^*$ (the value of $p$ for which the blue and green lines cross, represented by a vertical dashed line). This is an instance of \textit{free ride}\footnote{In social sciences, the free rider problem is defined as the market failure that occurs when those who benefit from common pool resources underpay. In our setup, it refers to the fact that the cost of deviating players being less than the cooperative setup.} meaning that individualists can take advantage of cooperative players, but this advantage diminishes as the proportion of individualists increases. 
 Figures~\ref{fig:p-mized_p_starVSq} and~\ref{fig:p-mized_p_starVSq_negposQ} give the variations of $p^*$ as a function of $q$. 
\begin{figure}
    \begin{subfigure}{0.32\textwidth}
        \includegraphics[width=1\linewidth]{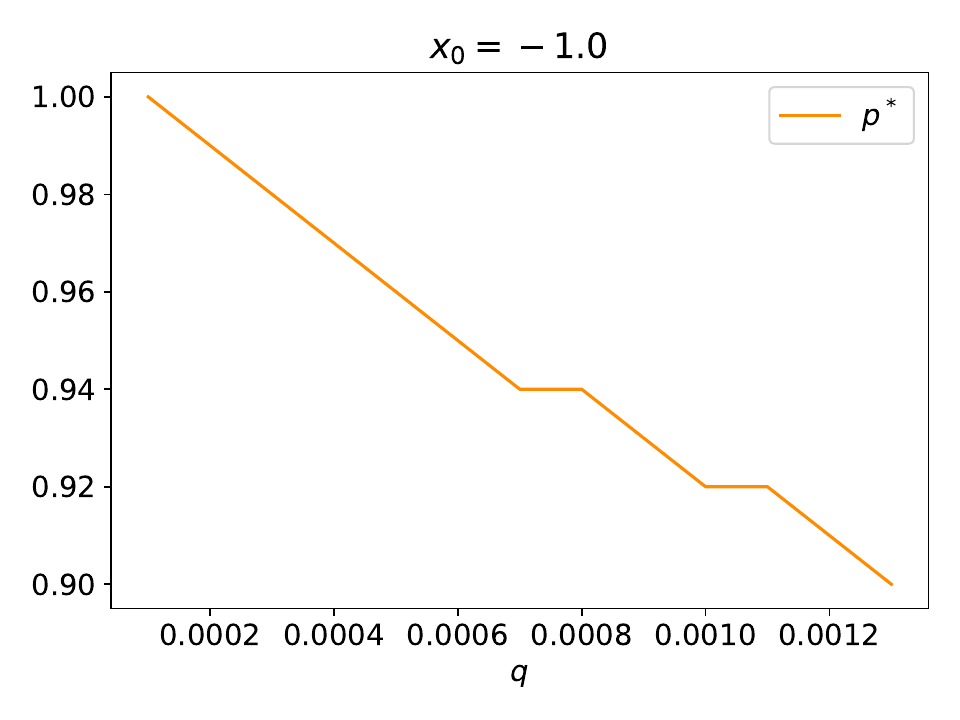}
    \end{subfigure}
    \begin{subfigure}{0.32\textwidth}
        \includegraphics[width=1\linewidth]{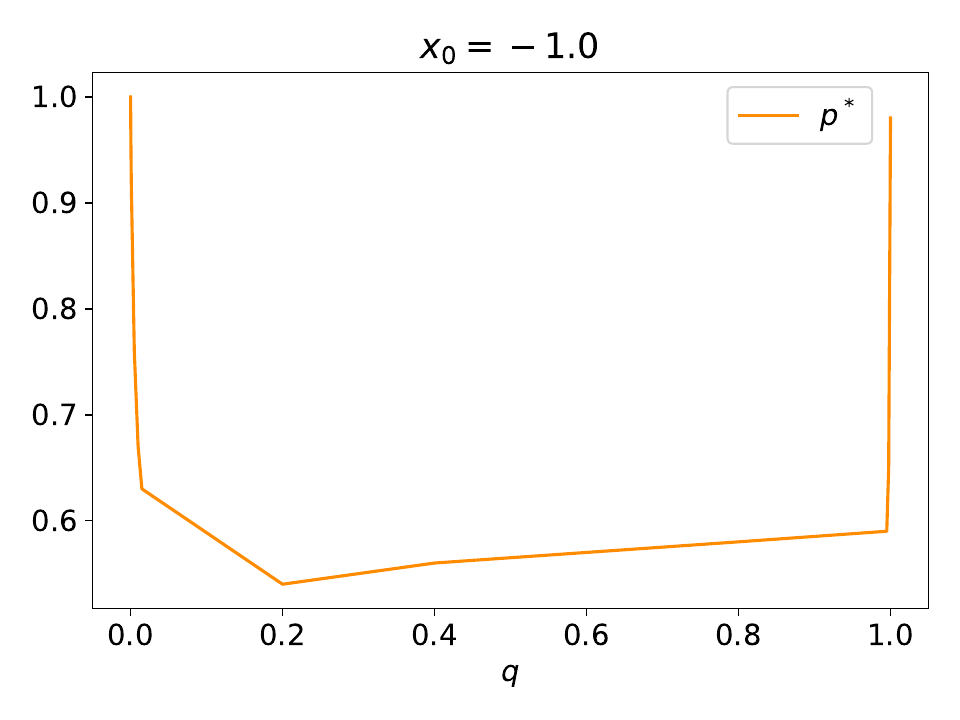}
    \end{subfigure}
    \begin{subfigure}{0.32\textwidth}
        \includegraphics[width=1\linewidth]{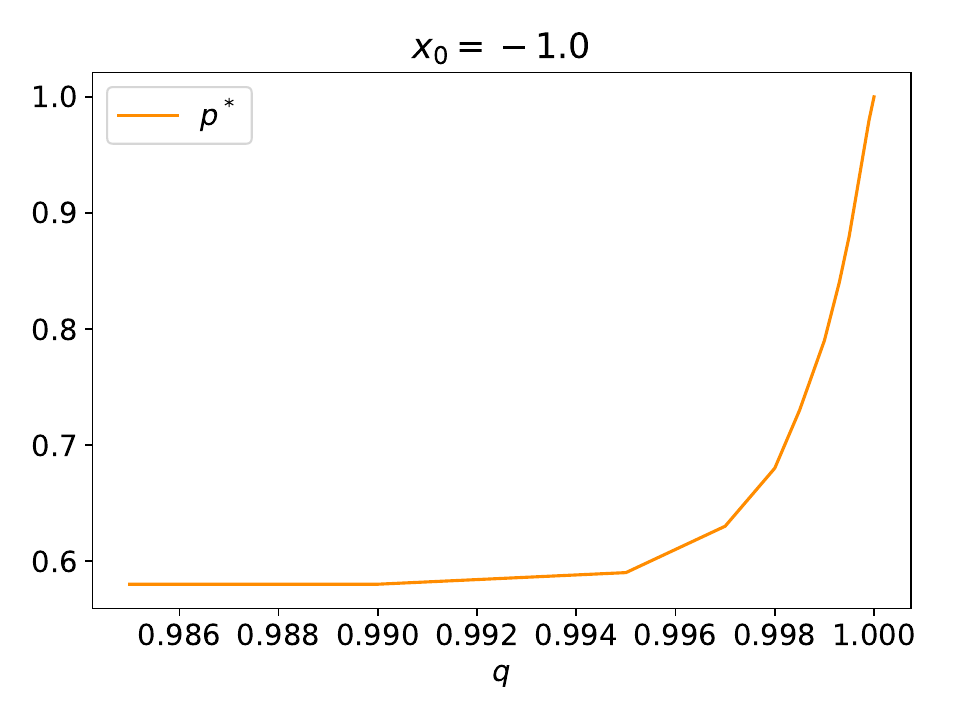}
    \end{subfigure}
    \caption{$p^*$ with respect to different $q$ choices.}
    \label{fig:p-mized_p_starVSq}
\end{figure}

\begin{figure}
    \centering
    \includegraphics[width=\linewidth]{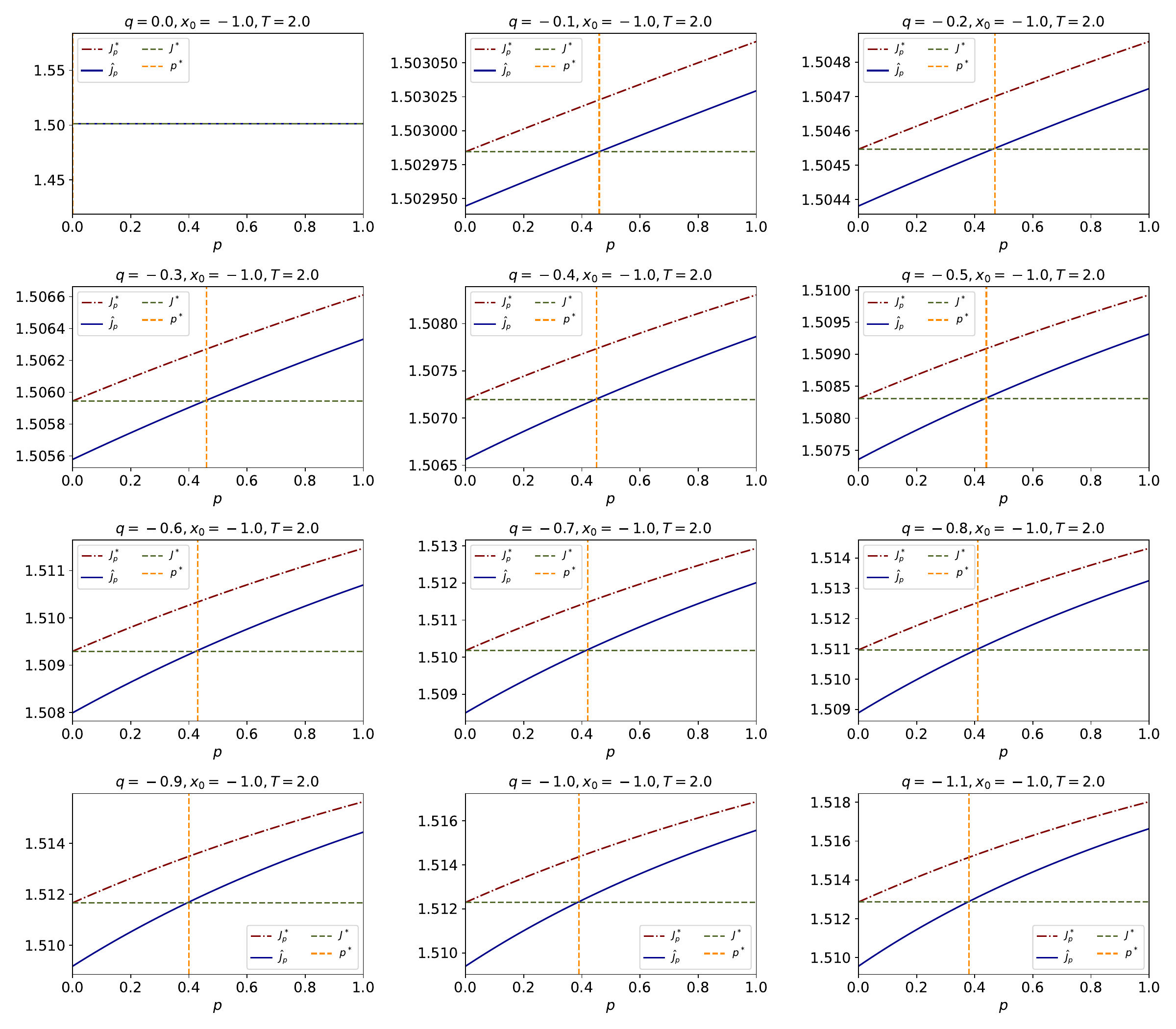}
    \caption{ Results for $q \le 0$: Costs of the non-cooperative people and the social planner followers are shown with (full) blue and (dashed-dotted) red lines, respectively. The (dashed) green line represents the cost of the MFC i.e., social planner's optimization. The vertical (dashed) orange line represents the value of $p^*$.}
    \label{fig:p-mixed_cost_comp_over_p_negQ}
\end{figure}

\begin{figure}
\center
    \begin{subfigure}{0.32\textwidth}
        \includegraphics[width=1\linewidth]{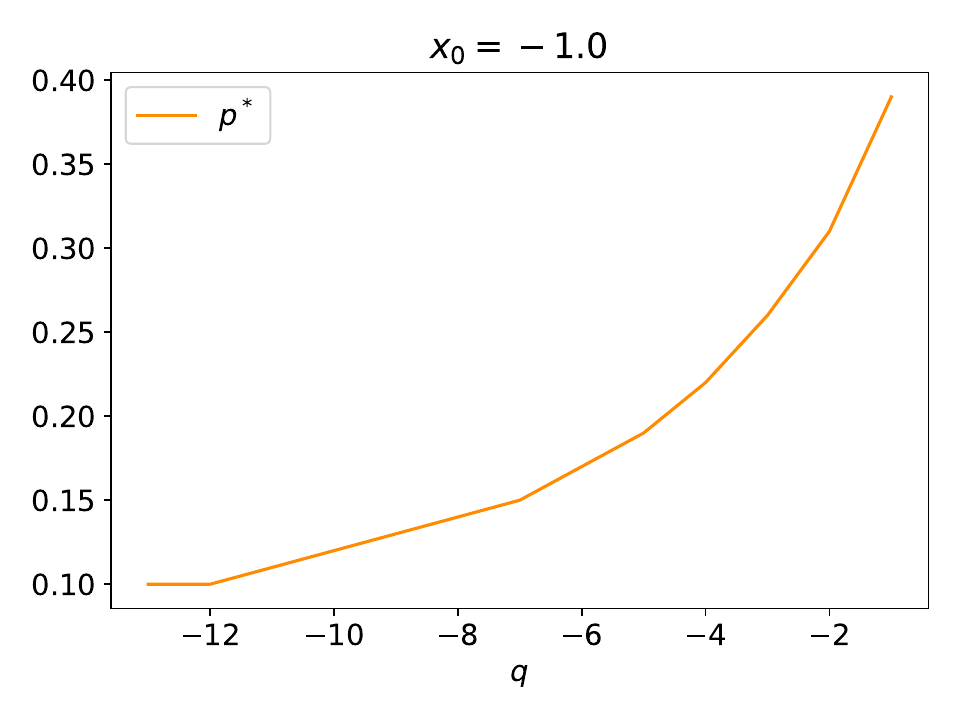}
    \end{subfigure}
    \begin{subfigure}{0.32\textwidth}
        \includegraphics[width=1\linewidth]{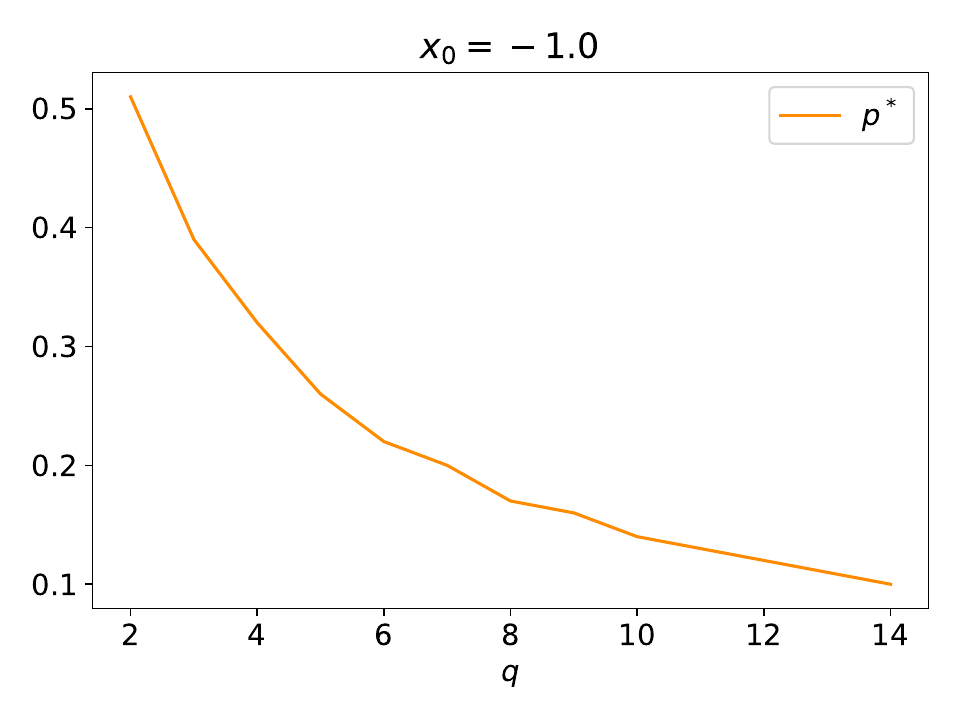}
    \end{subfigure}
    \caption{$p^*$ with respect to different $q$ choices.}
    \label{fig:p-mized_p_starVSq_negposQ}
\end{figure}

\subsection{Deviation Iterations}
\label{subsec:MFCtoMFG_iterative}

We now consider the following situation: a large population of players are behaving in a socially optimal way, but some players start behaving purely in their own interest, while other people do not change their behavior. Players who are deviating from the social optimum can be allowed to repeatedly update their behavior (i.e., their control), which leads to an iterative procedure. We assume that the players are myopic in the sense that they do not know the proportion of \textit{non-cooperative} players in the whole population, and they do not try to anticipate the population distribution in future iterations. They simply compute a best response to the distribution they currently see, which is composed of non-cooperative and cooperative players. In this way, a sequence of distributions and a sequence of controls are generated. 

At each iteration, we need to determine which players change their control and which ones keep the same control as in the previous iteration. There is an important parameter to determine, namely, the proportion of players that will keep playing the control obtained at a given iteration. We start with a generic algorithm and then present two instances: Picard fixed point and fictitious play.

\subsubsection{Generic Algorithm}
\label{subsubsec:MFCtoMFG_iterative_generic}

We present in Algorithm~\ref{algo:generic-iter} a generic iterative procedure. We take a general initial condition, but we can think of $\balpha^0 = \balpha^{\rMKV}$ to recover the setting we are interested in. Here, the sequence $(\underline{q}^{k} = (q^{n,j})_{j=0,\dots,n})_{n = 0, \dots, N}$ quantifies the proportions of players playing each of the past controls. At iteration $n$, $\underline{q}^n = (q^{n,j})_j$ represents a distribution over $\{1,\dots,n\}$ and for $j \le n$, control $\balpha^j$ is played by a proportion $q^{n,j}$ of the whole population. In particular, there is a proportion $q^{n,0}$ of players who use $\balpha^0$, which is the initial control. To be specific, we denote by $\bX^{\balpha}$ the solution of~\eqref{fo:state} with control $\balpha$. When iteration $n \ge 1$ starts, controls $(\balpha^{j})_{j=0,\dots,n-1}$ have already been generated in previous iterations. We consider the $n$ processes, $\bX^{\balpha^j}$, $j=0,\dots,n-1$, and we denote by $\mu^j_t$ the law of $X^{\balpha^j}_t$, for $t \in [0,T]$. The flow $\tilde{\bmu}^{n-1}$ of global distribution of the population is observed by each player. It is generated by the population where a proportion $q^{n-1,j}$ of players use control $\balpha^j$, $j=0,\dots,n-1$. Due to the form of the dynamics~\eqref{fo:state}, it amounts to: 
$$
    t \mapsto \tilde\mu^{n-1}_t = \sum_{j=0}^{n-1} q^{n-1,j} \mu^{j}_t.
$$
Then, $\balpha^n$ is defined as the best response against this distribution i.e., $\balpha^n$ minimizes $J(\cdot; \tilde{\bmu}^{n-1})$. 

\begin{algorithm}
\caption{Generic iterative procedure}
\label{algo:generic-iter}
    \KwData{An initial control $\balpha^0$; a number of iterations $N$; a sequence of probability distributions $(\underline{q}^{n} = (q^{n,j})_{j=0,\dots,n})_{n = 0, \dots, N}$ such that for every $n$, $\sum_{j=0}^{n} q^{n,j}=1$}
    \KwResult{$\underline{\bmu} = (\bmu^n)_{n=0,\dots,N}$ such that $\mu^n_t$ is a probability distribution for every $n$ and $t \in [0,T]$, and $\underline{\balpha} = (\balpha^n)_{n=0,\dots,N}$}
    Let $\tilde{\bmu}^{0} = \bmu^{0}$ be the distribution generated when all the players use control $\balpha^0$
    
    \For{$n=1,\dots,N$}{
    Compute $\balpha^{n}$, the best response to $\tilde{\bmu}^{n-1}$

    Let $\tilde{\bmu}^{n}$ be the flow of distributions generated by controls  $(\balpha^j)_{j=0,\dots,n}$ used respectively with proportions $\underline{q}^{n} = (q^{n,j})_{j=0,\dots,n}$ i.e., $\tilde\mu_t^{n} = \sum_{j=1}^n q^{n,j} \mu^{\balpha^j}_t$ 
    }
    Return $\tilde{\bmu}^N$ and $\underline{\balpha} = (\balpha^n)_{n=0,\dots,N}$
\end{algorithm}

Next, we present two instances of Algorithm~\ref{algo:generic-iter}.

\subsubsection{Fixed Point Algorithm}
\label{subsubsec:MFCtoMFG_iterative_fixedpoint}

Consider a sequence $(\tilde{q}^n)_n \in [0,1]$ and take $\underline{q}^{n} = (1-\tilde{q}^n, 0,\dots,0,\tilde{q}^n)$. The interpretation is as follows: at each iteration, there is a proportion $\tilde{q}^n$ of players who compute their best response to the previous distribution, and the next distribution is generated by this proportion of players plus a proportion $(1-\tilde{q}^n)$ who play $\balpha^0$. 

If $\tilde{q}^n = 1$ is constant with respect to $n$, then the algorithm boils down to Picard fixed point iterations in which, at each iteration, all the players adjust their behavior by computing a best response to the current distribution. Under strict contraction assumptions, this algorithm can be shown to converge to an MFG equilibrium as the number of iterations $N$ goes to infinity. More generally, if $\tilde{q}^n = p$ is constant with respect to $n$ for some $p \in [0,1]$, then, under strict contraction assumption, we expect these fixed point iterations to converge to a $p$-partial mean field equilibrium as defined in the previous section, see Definition~\ref{def:p-partial-eq}.

In the sequel, we choose 
\begin{equation}
\label{eq:tilde-q-k-def}
    \tilde{q}^n = 1-\prod_{i=0}^n(1-p_i). 
\end{equation}
Then at each iteration $n$ we pick a proportion $p_n$ of the players who were still playing $\balpha^0$ and we add them to the pool of players who are deviating from the socially optimal behavior and behaving in their own interest.

Now we focus on our example from Subsection~\ref{subsec:example_terminal_mean_interaction} and further discuss the properties of deviation iterations in the form of the fixed point algorithm. Since the mean field interactions only come through the terminal cost, the equilibrium is characterized by the state distribution at time $T$. We write the iterative process as
\begin{equation}
\label{eq:iter-barX-np1}
    \bar{X}^{n+1}_T =  \underbrace{\prod_{i=0}^n(1-p_i)\bar{X}^{\rMKV}_T}_{\text{followers of the social planner}} + \underbrace{\big(1- \prod_{i=0}^n(1-p_i)\big) \tilde{\bar{X}}^{n+1}_T}_{\text{non-cooperative players}},
\end{equation}
where $\tilde{\bar{X}}^{n+1}_T$ denotes the mean of the state at the terminal time that is induced by the best response $\balpha^{n+1}$ given to the environment ${\bar{\bsX}^n}$.
In order to understand the convergence results of the iterative process, we first provide the solutions of the MFC, the MFG and the $p$-partial mean field game. In this section, we will use the probabilistic approach to find the explicit solutions for the mean of the terminal states. Detailed proofs of Propositions~\ref{prop:terminal_state_mkv}, \ref{prop:terminal_state_mfg}, \ref{prop:terminal_state_best_resp} and \ref{prop:terminal_state_p-mixed-mfg} below can be found in Appendix~\ref{app:proofs_iterative_convergence}. In these propositions, we assume that $X_0=x_0\in \mathbb{R}$.

\begin{proposition}
\label{prop:terminal_state_mkv}
In the MFC problem, the mean of the state at the terminal time is given as:
\begin{equation}
\label{eq:prop:terminal_state_mkv}
    \bar X_T^{\rMKV} = \bar X_T^0 = \dfrac{2x_0}{e^T(1+(1-q)^2)+ e^{-T}(1-(1-q)^2)}.
\end{equation}
\end{proposition}

\begin{proposition}
\label{prop:terminal_state_mfg}
In the MFG problem, the mean of the state at the terminal time is given as:
\begin{equation}
    \bar X_T^{\rMFG} = \dfrac{2x_0}{e^T(1+(1-q))+ e^{-T}(1-(1-q))}.
\end{equation}
\end{proposition}

Furthermore, we find the mean of the terminal state induced by the best response given to a fixed environment explicitly.

\begin{proposition}
\label{prop:terminal_state_best_resp}
The mean of the state (for the whole population) at the terminal time that is induced by the best response, $\balpha^{n+1}$, at iteration $n$ to the environment $\bar \bsX^n$ is given as:
\begin{equation}
\label{eq:tilde-bar-Xnp1-def}
        \tilde{\bar X}_T^{n+1} = e^{-T}x_0 + \dfrac{1}{2}q\bar X_T^{n}(1-e^{-2T}).
\end{equation}
\end{proposition}

Below, we also give the $p$-partial mean field equilibrium results to show the convergence results.

\begin{proposition}
\label{prop:terminal_state_p-mixed-mfg}
The mean of the state of the non-cooperative (i.e., deviating) players at the terminal time that is induced by the $p$-partial mean field equilibrium control, $\hat \balpha^p$, is given as:
\begin{equation}
\label{eq:prop:terminal_state_p-mixed-mfg}
        {\bar X}_T^{p\rm -MFG} = \dfrac{2 x_0 + q(1-p) \bar X_T^{\rMKV}(e^T-e^{-T})}{e^{T}(1 + (1-qp) ) + e^{-T} (1 - (1-qp))} .
\end{equation}
\end{proposition}

\begin{remark}
As discussed in Remark~\ref{rem:hat:star:hat}, when $p=1$ in equation~\eqref{eq:prop:terminal_state_p-mixed-mfg}, we recover the regular MFG solution. Note that when $p=0$,  equation~\eqref{eq:prop:terminal_state_p-mixed-mfg} is for the mean of the state of a single player using the best response, which explains why it does not need to coincide with the MFC solution~\eqref{eq:prop:terminal_state_mkv}. 
\end{remark}

We now focus on the convergence properties of the iterative process. We show that its limit point must be the $p$-partial mean field equilibrium and we identify a sufficient condition that ensures the convergence of $\bar X^n_T$ as $n$ goes to infinity.

\begin{theorem}
\label{the:iterative_decaying_p} Given a sequence of proportions $p_i \in [0,1]$ for  all $i\geq0$. Let  $p^* = 1-\prod_{i=0}^\infty(1-p_i)$. Recall the definition~\eqref{eq:tilde-q-k-def} of $\tilde{q}^k$, and Algorithm~\ref{algo:generic-iter}. Let $(\bar{\bX}^n)_{n \ge 0}$ be defined by the iterations~\eqref{eq:iter-barX-np1}. Then,
    \item[(i)] If $\bar X_T^{n}$ converges, its limit $\hat{\bar{X}}_T^{p^*}$ is the $p$-partial mean field equilibrium with $p = p^*$ i.e., $\hat{\bar{X}}_T^{p^*} = (1-p^*) \bar X_T^{\rMKV} + p^* \bar X_T^{p^*\rm -MFG}$.
    \item[(ii)] If $ |\tfrac{q}{2}(1-e^{-2T})|<1$, $\bar X_T^{n}$ converges as $n\rightarrow\infty$.
\end{theorem}
The sufficient condition in point $(ii)$ requires the time horizon, namely $T$, or the interaction strength, namely $q$, to be small. This also holds for more general models, as explained below in Remark~\ref{rem:fixedpoint-general-remark}.

\begin{proof}[Proof of Theorem~\ref{the:iterative_decaying_p}]

{\bf Proof of $(i)$ (identification of the limit): }
Let us identify the limit, denoted by $\bar{X}^\infty_T$, assuming it exists. This limit satisfies:
$$
    \bar{X}^\infty_T = C_0 + C_1 \prod_{i=0}^\infty(1-p_i) + C_2 (1 - \prod_{i=0}^\infty (1-p_i)) \bar{X}^\infty_T.
$$
So:
\begin{equation*}
    \begin{aligned}
        \bar{X}^\infty_T &= \frac{C_0 + C_1 \prod_{i=0}^\infty(1-p_i)}{1 - C_2 (1 - \prod_{i=0}^\infty (1-p_i))}\\
        &=(1-p^*) \bar X_T^{\rMKV} + p^* \bar X_T^{p^*\rm -MFG}\\ 
        & = X_T^{p^*\rm -MFG}.
    \end{aligned}
\end{equation*}

{\bf Proof of $(ii)$ (convergence): } 
We will split the proof into three steps and we will use the notations: $C_0 = e^{-T} x_0 \in \RR$, $C_1 = \bar{X}^{\rMKV}_T - C_0 \in \RR$, $C_2 = \dfrac{q}{2}(1-e^{-2T}) \in \RR$, and $C_x = \frac{|C_0| + |C_1|}{1-|C_2|} \in \RR_+$.  Our assumption for this part rewrites $|C_2|<1$. We denote $Q_n = \prod_{i=0}^n(1-p_i) \in [0,1]$.

The sequence $(\bar X_T^{n})_n$ satisfies an induction formula, which can be derived as follows:
\begin{align}
    \bar{X}^{n+1}_T 
    &=  \prod_{i=0}^n(1-p_i)\bar{X}^{\rMKV}_T + \big(1- \prod_{i=0}^n(1-p_i)\big) \tilde{\bar{X}}^{n+1}_T
    \notag
    \\
    &= \prod_{i=0}^n(1-p_i)\bar{X}^{\rMKV}_T + \big(1- \prod_{i=0}^n(1-p_i)\big) [e^{-T}x_0 + \dfrac{1}{2}q\bar X_T^{n}(1-e^{-2T})]
    \notag
    \\
    &= e^{-T}x_0  + \prod_{i=0}^n(1-p_i)\left[\bar{X}^{\rMKV}_T - e^{-T}x_0 \right] + \big(1- \prod_{i=0}^n(1-p_i)\big) \dfrac{1}{2}q\bar X_T^{n}(1-e^{-2T})
    \notag
    \\
    &= C_0  + C_1 Q_n + C_2 \big(1- Q_n\big) \bar X_T^{n},
    \label{eq:proof-conv-iter-barXnp1}
\end{align}
where the first equality is by~\eqref{eq:tilde-q-k-def} and the second equality is by~\eqref{eq:tilde-bar-Xnp1-def}. 

{\bf Step 1 (uniform boundedness of $(\bar X_T^{n})_n$):} 
Notice that $\bar{X}_T^0 = \bar{X}_T^{{\rMKV}} = C_0 + C_1$ by definition of $C_0$ and $C_1$. Then, we have $|\bar{X}_T^0| \leq \frac{|C_0 + C_1|}{1-|C_2|}\leq\frac{|C_0| + |C_1|}{1-|C_2|} = C_x$ since we assume $|C_2|<1$. We then proceed by induction. Assume it is true for some $n \ge 0$. Then by~\eqref{eq:proof-conv-iter-barXnp1}, we have:
\begin{equation*}
    |\bar{X}^{n+1}_T|
    \le |C_0| + |C_1| + |C_2| |\bar X_T^{n}|
    \le \frac{(1 - |C_2|)(|C_0| + |C_1|) + |C_2|(|C_0| + |C_1|)}{1 - |C_2|} 
    = C_x,
\end{equation*}
so the bound also holds for $\bar{X}^{n+1}_T$. Hence for all $n \ge 0$, $|\bar{X}^n_T| \le C_x$. 

\textbf{Step 2 (Cauchy property for $(Q_k)_k$):}  For later use, we note that since $Q_{k}$ converges towards $1-p^*$ as $k \to +\infty$, it is Cauchy. Hence: for any $\delta>0$, there exists $N_1(\delta)$ such that for all $k \ge N_1(\delta)$ and all $m \ge 0$,  $|Q_{k+m} - Q_{k}| < \delta$.

{\bf Step 3 (Cauchy property for $(\bar{X}^n)_n$):} We want to show that this sequence is Cauchy i.e., for every $\epsilon>0$, there exists an integer $N(\epsilon)>0$ such that for any $n>N(\epsilon)$ and $m \ge 0$, $|\bar{X}^{n+m}_T - \bar{X}^{n}_T| < \epsilon$. We split this step into two sub-steps.

\textbf{Step 3 (a):} Consider two integers $n \ge 0$ and $m \ge 0$. We have: 
\begin{equation}
\begin{aligned} 
\label{eq:barX_cauchy}
    &|\bar{X}^{n+m}_T - \bar{X}^{n}_T|
    \\
    &\quad= |C_1 Q_{n+m{-1}} + C_2 \big(1- Q_{n+m{-1}}\big) \bar X_T^{n+m-1} - C_1 Q_{n{-1}} - C_2 \big(1- Q_{n{-1}}\big) \bar X_T^{n{-1}}|
    \\
    &\quad\le |C_1| |Q_{n+m{-1}} - Q_{n{-1}}| + |C_2| |\big(1- Q_{n+m{-1}}\big) \bar X_T^{n+m-1} - \big(1- Q_{n{-1}}\big) \bar X_T^{n{-1}}|
    \\
    &\quad\le |C_1| |Q_{n+m{-1}} - Q_{n{-1}}| 
    \\
    &\quad\quad + |C_2| |\big(1- Q_{n+m{-1}}\big) - \big(1- Q_{n{-1}}\big)| |\bar X_T^{n+m-1}|
     +  |C_2| \big(1- Q_{n{-1}}\big) |\bar X_T^{n+m-1} - \bar X_T^{n{-1}}|
    \\
    &\quad\le \Gamma |Q_{n+m{-1}} - Q_{n{-1}}| 
        +  |C_2| \big(1- Q_{n{-1}}\big) |\bar X_T^{n+m-1} - \bar X_T^{n{-1}}|,
\end{aligned}
\end{equation}
where we use the notation $\Gamma =   |C_1 | +  |C_2 | C_x$.
 
By induction, we have:
\begin{align}
    |\bar{X}^{n+m}_T - \bar{X}^{n}_T|
    &\le \Gamma |Q_{n+m{-1}} - Q_{n{-1}}| 
    + |C_2| (1-Q_{n-1}) |\bar X_T^{n+m-1} - \bar X_T^{n{-1}}|
    \notag 
    \\
    &\le \Gamma |Q_{n+m{-1}} - Q_{n{-1}}|\notag \\
    &+ |C_2| \Gamma |Q_{n+m-2} - Q_{n{-2}}|
    + |C_2|^2 \prod_{i=1}^2 (1-Q_{n-i}) |\bar X_T^{n+m-2} - \bar X_T^{n{-2}}|
    \notag
    \\
    &\dots
    \notag
    \\
    &\leq \underbrace{\sum_{i=1}^n  |C_2 |^{i-1} \Gamma |Q_{n+m-i} - Q_{n-i}|}_{\boldsymbol{A}} 
    + \underbrace{|C_2 |^n \prod_{i=1}^n (1-Q_{n-i})|\bar{X}_T^m-\bar{X}_T^0|}_{\boldsymbol{B}}.
    \label{eq:proof-cv-barX-upper-bound}
\end{align}

\textbf{Step 3 (b):}  We analyze the last right hand side in~\eqref{eq:proof-cv-barX-upper-bound}. 

We start with the last term in the sum, namely $\boldsymbol{B}$. Notice that $\prod_{i=1}^n (1-Q_{n-i}) \le 1$, $|\bar{X}_T^m-\bar{X}_T^0| \le 2C_x$ and $|C_2|<1$. So if $m \ge 0$ and $n \ge \frac{\epsilon/3}{\log|C_2|} \log(2C_x)$,  then  $\boldsymbol{B}$ is bounded by $\epsilon/3$. 

We the consider the first term in the sum, namely $\boldsymbol{A}$. We split it as:
$$
    \boldsymbol{A}
    = \underbrace{\Gamma \sum_{i = 1}^{N_0}  |C_2 |^{i-1}  |Q_{n+m-i} - Q_{n-i}|}_{\boldsymbol{A_1}} + \underbrace{\Gamma \sum_{i=N_0+1}^n |C_2 |^{i-1} |Q_{n+m-i} - Q_{n-i}|}_{\boldsymbol{A_2}},
$$
for $N_0 \in \{1,\dots,n\}$ to be chosen below. For $\boldsymbol{A_2}$, since $|Q_{n+m-i} - Q_{n-i}| \le 1$:
$$
    \Gamma \sum_{i=N_0+1}^n  |C_2 |^{i-1} |Q_{n+m-i} - Q_{n-i}|
    \le \Gamma \sum_{i=N_0+1}^n  |C_2 |^{i-1}
    = \Gamma \frac{ |C_2 |^{N_0+1} -  |C_2 |^{n+1}}{1- |C_2 |}
    \le \Gamma \frac{|C_2|^{N_0+1}}{1 - |C_2|}.
$$
This quantity is smaller than $\epsilon/3$ if $N_0 \ge \frac{1}{\log |C_2|} \log\left( \frac{\epsilon(1-|C_2|)}{3\Gamma}\right) - 1$. 

For $\boldsymbol{A_1}$, since $|C_2| \le 1$: 
\begin{align*}
    \Gamma \sum_{i=1}^{N_0}  |C_2 |^{i-1} |Q_{n+m-i} - Q_{n-i}|
    &\le \Gamma \sum_{k=n-N_0}^{n - 1} |Q_{k+m} - Q_{k}|,
\end{align*}
which is smaller than $\epsilon/3$ for $n \ge N_1(\epsilon / (3 \Gamma N_0)) + N_0$. Indeed, recall that \textbf{Step 2} gives the following property: if $n-N_0 \ge N_1(\delta)$, then $|Q_{k+m} - Q_{k}| < \delta$  for all $k \ge n-N_0$ and all $m$. Here we take $\delta = \epsilon / (3 \Gamma N_0)$.

Combining the terms, let $N_0 = \left\lceil \frac{1}{\log |C_2|} \log\left( \frac{\epsilon(1-|C_2|)}{3\Gamma}\right) - 1\right\rceil$, and let
\[
    N(\epsilon) = 
    \left\lceil\max\left\{ \frac{\epsilon}{3\log|C_2|} \log(2C_x), N_0 + 1, N_1\left(\frac{\epsilon}{3 \Gamma N_0}\right) + N_0 \right\} \right\rceil
\]
For any $m \ge 0$ and $n \ge N(\epsilon)$, we have $\boldsymbol{A_1} \le \frac{\epsilon}{3}$, $\boldsymbol{A_2} \le \frac{\epsilon}{3}$ and $\boldsymbol{B} \le \frac{\epsilon}{3}$, hence $|\bar{X}_T^{n+m} - \bar{X}_T^{n}| < \epsilon$. In other words, we proved that $(\bar{X}_T^{n})_n$ is a Cauchy sequence and hence it converges. 

\end{proof}

We conclude our analysis of this iterative procedure with a special case of Special case of Theorem~\ref{the:iterative_decaying_p}. 
\begin{corollary}
\label{prop:iterative_fixed_p}
When $p_i$ is constant i.e., $p_i =p$ for all $i \ge 0$,
if $\bar{X}_T^{n}$ converges, its limit $\hat{\bar X}_T$ is the mean field game solution.
\end{corollary}

In this way, we conclude that even if players use myopic adjustments, they converge to the Nash equilibrium when they are \textit{allowed} to deviate over iterations.

\begin{remark}
\label{rem:fixedpoint-general-remark}
    Above, for the sake of simplicity in the presentation we chose to show our results in the example introduced in Section~\ref{subsec:example_terminal_mean_interaction}. Similar convergence results hold for the more general model setup as introduced in Section~\ref{sec:model} and its proof is provided in Appendix~\ref{app:pmix_convergence_general} under the additional item $(i)$ in Assumption~\ref{assumption:SMP:MFG}. 
\end{remark}

\subsubsection{Fictitious Play Algorithm}
\label{subsubsec:MFCtoMFG_iterative_fictitious}

From the generic algorithm, Algorithm~\ref{algo:generic-iter}, we can also recover as a special case a Fictitious Play-type algorithm. Let $\tilde{q} \in (0,1)$ and take $\underline{q}^{n} = (q^{n,j})_{j=0,\dots,n} = \left(1-\frac{n-1}{n}\tilde{q}, \frac{1}{n}\tilde{q}, \dots, \frac{1}{n}\tilde{q}\right)$. In this setting, at iteration $n+1$, the population is composed of $n$ subpopulations of equal sizes, each playing one of the past policies. Equivalently, the population distribution can be generated by letting each player picking uniformly at random and independently of the other players one of the past policies. This is analogous to the Fictitious Play algorithm introduced, in the MFG setting, by Cardaliaguet and Hadikhanloo in~\cite{cardaliaguet2017learning}. Convergence of fictitious play scheme can be proved under monotonicity condition; see e.g.,~\cite{hadikhanloo2018learning,perrin2020fictitious,delarue2021exploration}. In our setting, we expect that the same proof technique can be applied to show convergence towards a $\tilde{q}$-partial mean field equilibrium. Indeed, since a proportion $1-\tilde{q}$ of the players is going to stick to the MFC optimum control, we can view the situation as a pure MFG with a modified cost function, and we initialize the fictitious play with the MFC optimal control.

\section{Conclusion}

In this paper, we analyze the bidirectional connections between MFGs and MFCs. The first part of the paper discusses incentivization methods to end up with the outcomes (either by having the same cost function or the same controls) with the (original) social optimum while players continue to play the non-cooperative game to find a Nash equilibrium. In this part, we also introduce and theoretically analyze a new game, which we call $\lambda$-interpolated mean field game, where each player's objective is a mixture of individual and social costs. In the second part of the paper, we focus on the instability problem of the MFC optimum (i.e., social optimum) and quantify it with the notion of Price of Instability. We further introduce and theoretically analyze a new game called $p$-partial mean field game, where a $p$-portion of the people are allowed to deviate from the social optimum. We conclude this part by analyzing a general iterative deviation process where the players have incomplete information about the game: they only see the whole distribution and they do not know the proportion of people deviating, $p$. We further analyze the convergence of this iterative decision process.

In the next steps of our work, we will work on a general setup where we can model the tragedy of the commons or common pool resource allocation problems with real life applications and propose incentives to overcome these obstacles. We also plan to add a \textit{stopping criterion} for the iterative deviation process introduced in Section~\ref{subsec:MFCtoMFG_iterative} in order to incorporate the rational decision making of individual players for deciding to change their behavior.

{\footnotesize
\bibliographystyle{siam}

}

\clearpage 

\appendix

\section{Additional Remarks on Section~\ref{subsec:MFGtoMFC_value}}

\label{app:MFGtoMFC_value}

\subsection{Remark on the Smoothness of $V$}
The reader may wonder about typical assumptions under which the function $V$ in \eqref{fo:value_function_MKV} is known to be smooth. Typically, the Hamilton-Jacobi-Bellman 
equation associated with the MFC problem 
\eqref{fo:value_MKV} has a classical solution $v$ when the coefficients are smooth (w.r.t. the measure argument) and convex (also w.r.t. the measure argument). We refer for instance to the book~\cite{CardaliaguetDelarueLasryLions}. In this latter work, convexity in the measure argument is understood in the space of measures. Another approach, considered for instance in the paper 
\cite{CCD} (see also the more recent paper 
\cite{GangboMesazros}) is to require the coefficients to be displacement convex i.e. to ask their lifts on the space of probability measures to be convex, which is in line 
with 
Remark~\ref{rem:convex}. 

Once the value function $v$ is known to be smooth, we can study the regularity of $V$. Below, we give
an outline. The first observation is that, from 
\eqref{eq:connection:v:V},
\begin{equation*}
\partial_\mu v(t,\mu)(x) = \partial_x V(t,x,\mu) + 
\int_{{\mathbb R}^d}
\partial_{\mu} V(t,x',\mu)(x) d\mu(x'), \qquad x \in \RR^d.
\end{equation*}
In turn, this says that the function 
$\bar \alpha$ in 
\eqref{fo:mkv_master_PDE}
is smooth. 
Then, 
we notice that
\eqref{fo:mkv_master_PDE}
is in fact a linear PDE on the space of probability measures driven by  the generator of the McKean-Vlasov SDE:
\begin{equation*}
dX_t = \bar{\alpha}\bigl(t,X_t,{\mathcal L}(X_t) \bigr) dt + dB_t, \quad t \in [0,T]. 
\end{equation*}
We then refer to 
\cite{BuckdahnLiRainerPeng} for a probabilistic approach to the regularity of
the solution to the linear PDE associated with the above equation.

\section{Additional Remarks on Section~\ref{subsec:MFGtoMFC_control}}
\label{app:MFGtoMFC_control}
\subsection{Remarks about Uniqueness of the Solution to the System 
\eqref{fo:mkv_FBSDE} (and to~\eqref{fo:mkv_fbsde_lambda})}
\label{app:extra-remark-uniqueness-FBSDE}

\begin{remark}
The reader may object that requiring the system 
\eqref{fo:mkv_FBSDE}
to be uniquely solvable is a very strong assumption. 
This is obviously true. 
However, we can easily deduce from simple examples of MFC problem
that this assumption is appropriate. Indeed, there are examples for which the minimizer is unique but the related first-order condition (as given by the stochastic maximum principle) has multiple solutions. 
For instance,
so is the case if we choose 
$d=1$, $T=1$ and
\begin{equation*}
J^{\rMKV}(\balpha) 
= \frac12 {\mathbb E} \int_0^1
\vert \alpha_t \vert^2 dt 
+ \frac12 {\mathbb E}
\bigl[ \vert X_1 \vert^2 \bigr]
+ G \bigl( {\mathbb E}[X_1] \bigr), 
\end{equation*}
for a function $G : {\mathbb R} \rightarrow {\mathbb R}$. 
In this example, we can prove that the minimal trajectories are characterized by ${\mathbb E}[X_1]$ and that it is a global minimizer of 
the function $\beta \mapsto \bar{x}^2 + G(\bar{x})$. As for the related stochastic maximum principle, its solutions are also characterized by 
${\mathbb E}[X_1]$, but it is just asked to be a critical point of $\beta \mapsto \bar{x}^2 + G(\bar{x})$ (i.e., a point where the derivative vanishes). 
Of course, we may have a unique global minimizer but several critical points. 
The issue in Proposition 
\ref{prop:MKV:MFG:lambda=1} is exactly the same: if the system 
\eqref{fo:mkv_FBSDE} had multiple solutions, some of them might not be minimizers of 
$J^{\rMKV}(\cdot)$. 
\end{remark}

\subsection{Finite Player Setup for $\lambda$-Interpolated Mean Field Game}

For reader's interest, we give some intuition about the particle interpretation of a
$\lambda$-interpolated mean field equilibrium. 
We propose the following construction.
Consider a population of 
$N^2$ players, indexed by 
pairs $(i,j) \in \{1,\cdots,N\}^2$, together with an i.i.d. collection $(X_0^{i,j},{\boldsymbol B}^{i,j})$ of initial conditions and Brownian motions. Intuitively, 
$i$ denotes a subgroup
of $N$ players within the population (to which player 
$(i,j)$ belongs) and 
$j$ denotes the label of 
the player within the population 
$i$. With each group
$i$, assign a bounded feedback function 
${\beta}^i : [0,T] \times ({\mathbb R}^d)^N \rightarrow {\mathbb R}^d$ that is symmetric in the last 
$(N-1)$-coordinates 
and assume that player 
$(i,j)$ obeys the (implictly well-posed) dynamics
\begin{equation*}
dX_t^{i,j} = \beta^i\bigl(t,X_t^{i,j},
{\mathbb X}_t^{i,-j}\bigr)
dt + dB_t^{i,j},
\end{equation*}
with $X_0^{i,j}$
as initial functional. 
Above, ${\mathbb X}_t^{i,-j}$
is the $(N-1)$-tuple 
$(X_t^{i,1},\cdots,X_t^{i,j-1},X_t^{i,j+1},\cdots,X_t^{i,N})$.

Next, define 
the (marginal) empirical measure of group $i$ as
\begin{equation*}
\overline{\mu}_t^{N,i} := \frac1N \sum_{j=1}^N 
\delta_{X_t^{i,j}},
\end{equation*}
whilst the marginal empirical measure of the whole population 
is
\begin{equation*}
\overline{\mu}_t^N :=
\frac1N \sum_{i=1}^N 
\mu_t^{N,i}. 
\end{equation*}
The cost to group 
$i$ is then defined by 
\begin{equation*}
\begin{split}
J^{\lambda,i}\bigl(\bbeta^1,\cdots,{\bbeta}^N\bigr) 
&= \frac{1}{2N} \sum_{j=1}^N {\mathbb E}
\int_0^T \vert \beta^i(t,{X}_t^{i,j},{\mathbb X}_t^{i,-j}) \vert^2 dt
\\
&\hspace{15pt} + \frac{1-\lambda}{N} \sum_{j=1}^N 
{\mathbb E}
\biggl[ \int_0^T 
f_0(X_t^{i,j},\overline{\mu}_t^N)
dt 
+ 
g(X_T^{i,j},\overline{\mu}_T^N) 
\biggr]
\\
&\hspace{15pt} 
+ \frac{\lambda}N
{\mathbb E}
\biggl[
\int_0^T 
\sum_{j=1}^N 
f_0(X_t^{i,j},\overline{\mu}_t^{N,i})
dt
+
g(X_T^{i,j},\overline{\mu}_T^{N,i}) 
\biggr].
\end{split}
\end{equation*}
Intuitively, 
$\lambda$-interpolated equilibria correspond to 
Nash equilibria between 
the $N$ 
populations, or equivalently 
between the 
$N$
feedback
functions $\bbeta^1,\cdots,\bbeta^N$.

We conjecture the following statement, whose proof should follow from standard approximation results in MFG/MFC theory, e.g.,~\cite[Chapter 6]{CarmonaDelarue_book_II}

\begin{proposition}
Assume that 
the coefficients $f_0$ and 
$g$ are bounded and Lipschitz continuous within the two arguments (using the $2$-Wasserstein distance for the measure variable) 
and  that 
$J^{\lambda,{\rm MF}}$ has 
a minimizer $\balpha^{\lambda}$ in feedback form i.e., 
the optimal trajectory has the form 
\begin{equation*}
dX_t = \alpha^{\lambda}(t,X_t) dt + 
dB_t, \qquad t \in [0,T]. 
\end{equation*}
Assume also that ${\boldsymbol \alpha}^{\lambda}$
is at most of linear and Lipschitz continuous in 
$x$, uniformly in $t$. 

Then,
for any constant $C>0$,
there exists a sequence 
$(\varepsilon_N)_{N \geq 1}$, 
converging to $0$, 
such that 
the tuple 
$(\balpha^{\lambda},\cdots,\balpha^{\lambda})$
(with each 
$\balpha^{\lambda}$ being understood 
as $(t,x_1,\cdots,x_N) \mapsto 
\alpha^{\lambda}(t,x_1)$)
is an $\varepsilon_N$-Nash equilibrium 
for the cost functions
$J^{\lambda,1},\cdots,J^{\lambda,N}$ within the class of 
feedback functions 
of the form
\begin{equation*}
    (t,x_1,\cdots,x_n) 
    \mapsto \beta\Bigl(t,x_1,\frac1{N-1} \sum_{j =2}^N \delta_{x_j}\Bigr)
\end{equation*}
for $\beta$ being $C$-Lipschitz continuous in 
$(x,\mu)$ (using the 
$2$-Wasserstein distance
on the space of probability measures) and 
with $\vert \beta(t,0,\delta_0) \vert \leq C$. Namely,
for any 
$i \in \{1,\cdots,N\}$ and for any other control 
$\bbeta : [0,T] \times ({\mathbb R}^d)^N \rightarrow {\mathbb R}^d$ induced by $\beta$ in the latter class, 
\begin{equation*}
J^{\lambda,i}(\balpha^{\lambda},\cdots,\balpha^{\lambda},\bbeta,\balpha^{\lambda},\cdots) 
\geq
J^{\lambda,i}(\balpha^{\lambda},\cdots,\balpha^{\lambda},\balpha^{\lambda},\balpha^{\lambda},\cdots) 
- \varepsilon_N. 
\end{equation*}
\end{proposition}

\begin{remark}
The statement is given for reader's interest, as it clearly shows that an $\lambda$-interpolated mean field game model is in fact a game between
an 
infinite number of infinitely large population. Still, one may think of many possible extensions. For instance, one could allow the feedback functions
${\boldsymbol \beta}$ to depend on the states of other classes, which would be a natural extension. Here, our formulation is
simple 
since the players only observe their own class. It is however very similar to earlier results in the theory of MFG problems (in which players only observe their own state), hence our choice to stick to this simpler setting.  

Also, it is an interesting question to see whether 
the sequence 
$(\varepsilon)_{N \geq 1}$
in the statement can be chosen independently of the constant $C$. Here, 
the Lipschitz property 
of $\beta$ allows us to get an explicit rate in the propagation of chaos for the corresponding class. Such a difficulty is proper to this model and does not manifest in standard mean field games
(as classes reduce to one individual only). 

Last, it is worth observing that 
$\balpha^{\lambda}$ may depend on the law of the optimal trajectory, but the dependence is implicitly encoded in the additional time variable. This is a common feature in MFG. 
\end{remark}

\section{Proofs of Section \ref{subsubsec:MFCtoMFG_completeinfo_pmixed}}

\label{app:p_mixed}

\subsection{Proof of Proposition~\ref{prop:lipschitz:property:p}}
\begin{proof}
We notice that (the complete version of) Proposition \ref{prop:existence:p-mixed:eq}
holds true under the standing assumption. Therefore, 
we can duplicate the proof of \eqref{eq:monotonicity:1}. We get, for two different
values $p,'p \in [0,1]$, 
\begin{equation*}
\begin{split}
&{\mathbb E}
\int_0^T \vert \hat{\alpha}_t^{p}
- \hat{\alpha}_t^{p'} \vert^2 
dt
\\
&\leq  \int_0^T \biggl[
\int_{{\mathbb R}^d} 
\bigl[ f_0\bigl( x, 
(1-p') \mu_t^{\rMKV}
+ p' \mu_t^{\hat{\balpha}^{p'}} \bigr) 
- f_0\bigl( x, 
(1-p) \mu_t^{\rMKV}
+ p 
\mu_t^{\hat{\balpha}^{p}}
\bigr) 
\bigr] 
d \bigl( \mu_t^{\hat{\balpha}^{p}}
- 
\mu_t^{\hat{\balpha}^{p'}}
\bigr) (x) 
\biggr] dt
\\
&\hspace{15pt}
+
\int_{{\mathbb R}^d} 
\bigl[ g\bigl( x, 
(1-p') \mu_T^{\rMKV}
+ p' \mu_T^{\hat{\balpha}^{p'}} \bigr) 
- g\bigl( x, 
(1-p) \mu_T^{\rMKV}
+ p \mu_T^{\hat{\balpha}^{p}} \bigr) 
\bigr] 
d  \bigl( \mu_T^{\hat{\balpha}^{p}}
- 
\mu_T^{\hat{\balpha}^{p'}}
\bigr)
 (x) .
\end{split}
\end{equation*}
Writing $f_0( x, 
(1-p') \mu_t^{\rMKV}
+ p' \mu_t^{\hat{\balpha}^{p'}})$ in the form 
\begin{equation*}
\begin{split}
&f_0 \bigl( x, 
(1-p') \mu_t^{\rMKV}
+ p' \mu_t^{\hat{\balpha}^{p'}} \bigr) 
\\
&= \Bigl[ f_0 \bigl( x, 
(1-p') \mu_t^{\rMKV}
+ p' \mu_t^{\hat{\balpha}^{p'}} \bigr)
- 
f_0 \bigl( x, 
(1-p) \mu_t^{\rMKV}
+ p \mu_t^{\hat{\balpha}^{p'}} \bigr)
\Bigr] 
+
f_0 \bigl( x, 
(1-p) \mu_t^{\rMKV}
+ p \mu_t^{\hat{\balpha}^{p'}} \bigr),
\end{split}
\end{equation*}
and similarly for $g$, 
and then using the monotonocity property and following 
\eqref{eq:second:order:g}, we easily obtain 
\begin{equation*}
\begin{split}
&{\mathbb E}
\int_0^T \vert \hat{\alpha}_t^{p}
- \hat{\alpha}_t^{p'} \vert^2 
dt
\leq C \vert p-p' \vert
\sup_{0 \leq t \leq T} {\mathbb E} \bigl[ \vert X_t^{\hat{\balpha}^p} - X_t^{\hat{\balpha}^{p'}} \vert^2\bigr],
\end{split}
\end{equation*}
from which the Lipschitz property easily follows. 

The second part of the statement is a direct consequence of the latter computation, noticing that 
\begin{equation*}
\begin{split}
&{\mathbb E}
\int_0^T \vert \hat{\alpha}_t^{p}
- \hat{\alpha}_t^{p'} \vert^2 
dt
\leq C \vert p-p' \vert^2,
\end{split}
\end{equation*}
and using item (c) in the consequences of 
Assumption \ref{assumption:initial-model}.
\end{proof}

\subsection{Proof of Proposition~\ref{prop:ode_mfg}}
\label{app:proof_ode_mfg}
In Subsection~\ref{subsec:example_terminal_mean_interaction} proofs, we make use of the Hamilton-Jacobi-Bellman (HJB) and Kolmogorov-Fokker-Planck (KFP) equations as in~\cite{Bensoussan_Book} in order to calculate the cost functions directly. 

\begin{proof}
We start by recalling the Hamiltonian and its minimizer: 
\begin{equation}
    \label{eq:Hamiltonian_analytical}
    \begin{aligned}
    H(t, x, m , y, \alpha) =& \alpha y+ \dfrac{1}{2}x^2 + \dfrac{1}{2}\alpha^2\\
    \alpha^* :=& \argmin_{\alpha} H(t, x, m, y, \alpha) = -y    
    \end{aligned}
\end{equation}
where $m$ denotes the state distribution. Realize that since in the running cost and the drift there is no dependence on the distribution, the Hamiltonian also does not depend on the distribution. 
We denote with $H^*(t, x, m, y) := H(t, x, m, y, \alpha^*) = \dfrac{1}{2} x^2 -\dfrac{1}{2}y^2$ the optimal Hamiltonian.
Then the Hamilton-Jacobi-Bellman (HJB) equation becomes:
\begin{equation}
\label{eq:HJB_mfg_before_ansatz}
    -\partial_t u(t,x) = \dfrac{1}{2} \partial^2_{xx} u(t,x) + H^*(t, x, m, \partial_x u(t,x)),\ u(T,x) = \dfrac{1}{2}(x-q\bar X_T)^2.
\end{equation}
We introduce the following ansatz:
\begin{equation}
\label{eq:HJB_KFP_ansatz}
    u(t,x) = \dfrac{1}{2} \eta_t x^2 + r_t x + s_t.
\end{equation}
Then, we have $\partial_x u(t, x) = \eta_t x + r_t$ and $\partial_{xx}^2 u(t,x) = r_t$. By plugging \eqref{eq:HJB_KFP_ansatz} into \eqref{eq:HJB_mfg_before_ansatz}, we obtain the following differential equations:
\begin{equation}
    \begin{aligned}
        \dot{\eta}_t &= \eta_t^2-1 && \eta_T=1,\\
        \dot{r}_t    &= \eta_t r_t && r_T=q\bar X_T,\\
        \dot{s}_t    &= \dfrac{1}{2} (r_t^2-\eta_t) && s_T=\dfrac{1}{2}q^2\bar X_T^2.
    \end{aligned}
\end{equation}
Our next step is to write the KFP equation
\begin{equation}
\label{eq:KFP_mfg_before_ansatz}
    \partial_t m(t, x) - \dfrac{1}{2} \partial^2_{xx}m(t, x) + \nabla_x \big(-m(t, x) \partial_x u(t,x) \big)=0
\end{equation}

We conclude by finding the dynamics of $\bar X_t$ can be found by using \eqref{eq:KFP_mfg_before_ansatz} and \eqref{eq:HJB_KFP_ansatz} as follows:
\begin{equation}
\begin{aligned}
    \dfrac{d\bar X_t}{dt}&= \dfrac{d}{dt} \int_{\mathbb R} x m(t, x)dx =\int_{\mathbb R} x \dfrac{\partial m(t, x)}{\partial t}dx\\[2mm]
    &=\int_{\mathbb R} x\Big(\dfrac{1}{2} \partial^2_{xx}m(t, x) - \nabla_x \big(-m(t, x) (\eta_t x+r_t) \big)\Big)dx\\[2mm]
    &= \int_{\mathbb R} x \dfrac{\partial m(t, x)}{\partial t}\Big(\eta_t x+r_t\Big) dx +  \int_{\mathbb R} x m(t, x)\eta_t dx\\[2mm]
    & = -(\eta_t \bar X_t+r_t)
\end{aligned}
\end{equation}
with $\bar X_0 = x_0$.

After realizing $\eta_t=1$ for all $t\in[0,1]$, we end up with forward backward ordinary differential equation system and the MFG cost can be directly computed using the value function:
\begin{equation}
\begin{aligned}
    \hat J_1 = \mathbb{E}\Big[u(0,x_0)\Big] = \dfrac{1}{2} x_0^2 +r_0 x_0 + s_0.
\end{aligned}
\end{equation}

\end{proof}

\subsection{Proof of Proposition~\ref{prop:ode_mfc}}
\label{app:proof_ode_mfc}

\begin{proof}

In the MFC, the only difference with respect to the proof of~\ref{prop:ode_mfg} from the MFG will be the addition of Gateaux differentials. The reason behind this is that since the players are identical, when one player deviates everyone is assumed to deviate in the same way and the distribution cannot be assumed to stay constant. Differently from the MFG, the HJB equation becomes
(for details, see e.g.,~\cite{Bensoussan_Book}):
\begin{equation}
    \begin{aligned}
        -\partial_t u(t,x) &= \dfrac{1}{2} \partial^2_{xx} u(t,x) + H^*(t, x, m, \partial_x u(t,x)) +\int_{\mathbb R} \dfrac{\delta H^*}{\delta m}(t, \xi, m, \partial_\xi u(t, \xi))m(t, \xi)d\xi,\\[2mm] 
        &\hskip5.5cm u(T,x) =g(x, m) + \int_{\mathbb R} \dfrac{\delta  g}{\delta  m}(x, m) m(T, \xi)d\xi. 
    \end{aligned}
\end{equation}
Since we have
\begin{equation}
    \begin{aligned}
        & \dfrac{\delta  H^*}{\delta m}(t, \xi, m, \partial_\xi u(t, \xi)) =0,\\[2mm]
        \hbox{and } &\int_{\mathbb R} \dfrac{\delta g}{\delta m}(x, m) m(T, \xi)d\xi = -q(1-q)\bar X_T x,
    \end{aligned}
\end{equation}
where $\bar X_T =\int_{\RR} \xi m(T, \xi) d\xi$, we can write the HJB equation as
\begin{equation}
\label{eq:HJB_mfc_before_ansatz}
    -\partial_t u(t,x) = \dfrac{1}{2} \partial^2_{xx} u(t,x) + H^*(t, x, m, \partial_x u(t,x)),\ u(T,x) = \dfrac{1}{2}(x-q\bar m)^2 -q(1-q)\bar X_T x.
\end{equation}
By using the same ansatz form as in \eqref{eq:HJB_KFP_ansatz}, we end up with the following differential equation system:
\begin{equation}
    \begin{aligned}
        \dot{\eta}_t &= \eta_t^2-1 && \eta_T=1\\
        \dot{r}_t    &= \eta_t r_t && r_T=-2q\bar X_T+q^2\bar X_T\\
        \dot{s}_t    &= \dfrac{1}{2} (r_t^2-\eta_t) && s_T=\dfrac{1}{2}q^2\bar X_T^2
    \end{aligned}
\end{equation}

The KFP equation and in conclusion the dynamics of $\bar X$ stays the same, in other words, we have
\begin{equation}
    \dot{\bar X}_t = -(\eta_t\bar X_t+r_t) \qquad \bar X_0 = x_0.
\end{equation}

After realizing $\eta_t=1$ for all $t\in[0,1]$, we end up with forward backward ordinary differential equation system and the MFC cost by using the ansatz as follows (for details, see e.g.,~\cite{lauriere2021numerical}):
\begin{equation}
\begin{aligned}
     J_0^* &= \mathbb{E}\Big[u(0,x_0) - \int_{\mathbb R} \dfrac{\delta g}{\delta m}(x, m) m(T, \xi)d\xi\Big] \\
    & = \dfrac{1}{2}x_0^2 +r_0 x_0 + s_0 + (1-q)q\bar X_T^2.
\end{aligned}
\end{equation}
\end{proof}

\subsection{Proof of Proposition~\ref{prop:ode_p_mixed_mfg}}
\label{app:proof_ode_p_mixed_mfg}
\begin{proof}
In the $p$-partial mean field game only the terminal cost changes. Therefore, this changes the terminal condition of the HJB equation. The new terminal condition for the HJB equation can be written as
\begin{equation}
    u(T,x) = \dfrac{1}{2}\Big(x - q(p\bar X_T + (1-p)\bar X_T^{\rMKV})\Big)^2,
\end{equation}
where $\bar X^{\rMKV}_T$ is the mean of the state at the terminal time in the MFC problem and it is taken exogenously. By using the same ansatz form as in \eqref{eq:HJB_KFP_ansatz} and following the same ideas as in the proofs of Propositions~\ref{prop:ode_mfg} and \ref{prop:ode_mfc}, we conclude that the forward backward ordinary differential equation system that characterizes the $p$-partial mean field equilibrium is given as
\begin{equation}
    \begin{aligned}
        \dot{\eta}_t &= \eta_t^2-1 && \eta_T=1\\
        \dot{r}_t    &= \eta_t r_t && r_T=-qp\bar X_T - q(1-p) \bar X_T^{\rMKV}\\
        \dot{s}_t    &= \dfrac{1}{2} (r_t^2-\eta_t) && s_T=\dfrac{1}{2}\big(q p \bar X_T + q(1-p)\bar X_T^{\rMKV}\big)^2\\
        \dot{\bar X}_t &= -(\eta_t\bar X_t+r_t) \qquad\qquad&&\bar X_0 = x_0.
    \end{aligned}
\end{equation}

After realizing $\eta_t=1$ for all $t\in[0,1]$, we end up with forward backward ordinary differential equation system and the $p$-partial mean field game cost can be directly computed using the value function:
\begin{equation}
\begin{aligned}
    \hat J_p = \mathbb{E}\Big[u(0,x_0)\Big] = \dfrac{1}{2} x_0^2 +r_0 x_0 + s_0.
\end{aligned}
\end{equation}
\end{proof}

\section{Proofs of Section~\ref{subsec:MFCtoMFG_iterative}}
\label{app:proofs_iterative_convergence}

\subsection{Proof of Proposition~\ref{prop:terminal_state_mkv}}
\begin{proof}

The MFC optimum can be characterized as the solution of the following FBSDE:
\begin{equation*}
\begin{cases}
&dX_{t} = - Y_{t}dt + dB_{t}
\\
&dY_{t} = - X_{t} dt + Z_{t} dB_{t}
\\
&Y_{T} = X_{T} - q \bar X_{T} - q \tilde{\mathbb E} \bigl[ \tilde{X}_{T} - q \bar X_{T} \bigr]
= X_{T} - (2q-q^2) \bar{X}_{T}.
\end{cases}
\end{equation*}
The optimal control is given by $\alpha^{\rMKV}_t=-Y_t$ for $0\le t\le T$. Taking expectations on both sides of the first two equations we find that
the system for the means is:
\begin{equation*}
\begin{cases}
&d \bar X_{t} = - \bar Y_{t} dt
\\
&d \bar Y_{t} = - \bar X_{t} dt, \quad \bar Y_{T} = (1-q)^2 \bar X_{T}.
\end{cases}
\end{equation*}
We look for a solution of the form:
\begin{equation*}
\bar Y_{t} = \eta_{t} \bar X_{t}, \quad t \in [0,T].
\end{equation*}
The deterministic function of time $(\eta_{t})_{0 \leq t \leq T}$ must solve the Riccati equation:
\begin{equation*}
\label{eq:riccati_4.2.2}
\dot{\eta}_{t} - \eta_{t}^2 =-1, \quad t \in [0,T] \ ; \quad \eta_{T} = (1-q)^2.
\end{equation*}
The solution of this Riccati equation is given by:
\begin{equation*}
\begin{split}
\eta_{t} &= \frac{
-e^{2(T-t)}+1 - (1-q)^2 (e^{2(T-t)}+1) }{-(1+
e^{2(T-t)})-(1-q)^2 (e^{2(T-t)}-1)
},
\end{split}
\end{equation*}
and as a result, we find:
\begin{equation*}
\bar X_{T} = x_{0} e^{- \int_{0}^T \eta_{t} dt}.
\end{equation*}
To compute the integral in the right hand side, we perform the change of variable: 
\begin{equation*}
u = e^{2(T-t)},
\end{equation*}
so that 
\begin{equation*}
T-t = \frac12 \ln(u) \ , \quad dt = - \frac{du}{2u},
\end{equation*}
and
\begin{equation*}
\begin{split}
\int_{0}^T \eta_{t} dt 
 &=   \ln\Bigl( e^{2T}
   (1 + (1-q)^2 ) 
 +
  (1 - (1-q)^2 )\Bigr)
  - T- \ln(2).
\end{split}
\end{equation*}
Finally, we conclude
\begin{equation}
\label{fo:X_bar_of_T}
\bar X_{T}^{\rMKV} = \frac{2x_{0}}{  e^T
   (1 + (1-q)^2 ) +  e^{-T}(1 - (1-q)^2) }.
\end{equation}

\end{proof}

\subsection{Proof of Proposition~\ref{prop:terminal_state_mfg}}
\begin{proof}
When we look at the MFG case, we end up with the same Riccati equation~\eqref{eq:riccati_4.2.2} with a different terminal condition. In the MFC, we had $\eta_T = (1-q)^2$; on the other hand, in the MFG, we have $\eta_T = (1-q)$. Then the proof follows the same steps to conclude:
\begin{equation}
\bar X_{T}^{\rMFG} = \frac{2x_{0}}{  e^T
   (1 + (1-q) ) +  e^{-T}(1 - (1-q)) }.
\end{equation}

\end{proof}

\subsection{Proof of Proposition~\ref{prop:terminal_state_best_resp}}
\begin{proof}
Since the mean field interaction enters the model only through the mean of $\mu_T$, we capture the environment by a number $\mu$ instead of an entire flow of probability measures.
When we replace the mean field interaction by a fixed parameter $\mu$, the optimal control problem becomes the minimization of:
\begin{equation*}
J^{\mu}(\boldsymbol \alpha) = {\mathbb E}
\biggl[ \int_{0}^T \frac12 (X_{t}^2 + \alpha_{t}^2) dt + \frac12 (X_{T} - q  \mu)^2
\biggr],
\end{equation*}
under the constraint $dX_{t} = \alpha_{t} dt + dB_{t}$, for  $\mu\in\mathbb{R}$ is fixed.
In the present situation, the FBSDE becomes:
\begin{equation*}
\begin{cases}
&dX_{t} = -Y_{t} dt + dB_{t}, \\
& dY_{t} = -X_{t}dt + Z_t dB_t, \quad Y_{T} = X_{T} - q\mu.
\end{cases}
\end{equation*}
and as before, the optimal control (i.e., the best response to the fixed environment $\mu$) is $\alpha_t=-Y_t$.
We look for a solution in the form:
\begin{equation*}
Y_{t} = X_{t} + \varphi_{t},
\end{equation*}
for some deterministic function of time
$(\varphi_{t})_{0 \leq t \leq T}$. The latter must solve:
\begin{equation*}
\dot{\varphi}_{t} - \varphi_{t} = 0, \quad t \in [0,T] \ ; \quad \varphi_{T} = - q \mu,
\end{equation*}
which implies:
\begin{equation*}
\varphi_{t} = - q\mu e^{t-T}, \quad t \in [0,T].
\end{equation*}
Finally, since:
\begin{equation*}
d\bar{X}_{t} = - \bigl( \bar{X}_{t} + \varphi_{t} \bigr) dt, \quad t \in [0,T],
\end{equation*}
we get:
\begin{equation*}
\begin{split}
\bar{X}_{T} &= e^{-T} x_{0} - e^{-T} \int_{0}^T e^s \varphi_{s} ds
\\
&=  e^{-T} x_{0} + \frac12  q \mu \bigl( 1 - e^{-2T}\bigr).
\end{split}
\end{equation*}
In this way, given an environment $\bar X^n_T$, the mean of the state at the terminal time that is induced by the best response $\balpha^{n+1}$ is found as 
\begin{equation}
    \tilde{\bar X}_T^{n+1} = e^{-T}x_0 + \dfrac{1}{2}q\bar X_T^{n}(1-e^{-2T}).
\end{equation}
\end{proof}

\subsection{Proof of Proposition~\ref{prop:terminal_state_p-mixed-mfg}}
\begin{proof}
The $p$-partial mean field equilibrium can be characterized by the solution of the FBSDE:

\begin{equation*}
\begin{cases}
&dX_{t} = - Y_{t}dt + dB_{t}
\\
&dY_{t} = - X_{t} dt + Z_{t} dB_{t}
\\
&Y_{T} = X_{T} - q p \bar X_{T} - q (1-p) \bar X_T^{\rMKV}
\end{cases}
\end{equation*}
and the equilibrium control is given by $\hat{\alpha}^p_t=-Y_t$ for $0\le t\le T$. Indeed, at terminal time, the mean of the states for the whole population is: $p \bar X_{T} - (1-p) \bar X_T^{\rMKV}$. Taking expectations on both sides of the first two equations, we find that
the system for the mean trajectories is:
\begin{equation*}
\begin{cases}
&d \bar X_{t} = - \bar Y_{t} dt
\\
&d \bar Y_{t} = - \bar X_{t} dt, \quad \bar Y_{T} = (1-qp)\bar X_T - q(1-p)\bar X_T^{\rMKV}.
\end{cases}
\end{equation*}
We look for a solution of the form:
\begin{equation*}
\bar Y_{t} = \eta_{t} \bar X_{t} + r_t, \quad t \in [0,T].
\end{equation*}
The deterministic functions of time $(\eta_{t}, r_t)_{0 \leq t \leq T}$ must solve the following Riccati equation and first order ordinary differential equation, respectively:
\begin{equation*}
\begin{aligned}
\dot{\eta}_{t} - \eta_{t}^2 &=-1, \quad &&\eta_{T} = (1-qp),\\
\dot{r}_t -\eta_t r_t &=0,  &&r_T = -q(1-p)\bar X_T^{\rMKV}.
\end{aligned}
\end{equation*}
The solution of this Riccati equation is given by:
\begin{equation*}
\eta_{t} = \frac{
-e^{2(T-t)}+1 - (1-qp) (e^{2(T-t)}+1) }{-(1+e^{2(T-t)})-(1-qp) (e^{2(T-t)}-1)
},
\end{equation*}
and the solution of the ordinary differential equation is given by:
\begin{equation*}
r_{t} = -q(1-p)\bar X_T^{\rMKV} e^{\int_t^T-\eta_tdt}.
\end{equation*}
From here, we find that:
\begin{equation}
\begin{aligned}
\label{eq:explicit_barX_T_p-mixed-mfg}
    \bar X_T &= \Big(e^{\int_0^T -\eta(s)ds}\Big)\Big(\int_0^T -e^{\int_0^t \eta_sds} r_tdt+x_0\Big)\\
    &=\Big(e^{\int_0^T -\eta(s)ds}\Big)\Big(q(1-p)\bar X_T^{\rMKV}\int_0^T e^{\int_0^t \eta_sds}e^{\int_t^T -\eta_sds} dt+x_0\Big).
\end{aligned}
\end{equation}
By using 
\begin{equation}
    \begin{aligned}
        \int_{0}^T \eta_{s} ds &=  \frac{1}{2}\ln(e^{-2T}) +\ln\Big(e^T \big(\cosh(T)+ (1-pq) \sinh(T)\big)\Big),\\[2mm]
        \int_{0}^t \eta_{s} ds &=   \frac{1}{2}\ln(e^{2t}) -\ln\Big(e^{2t} (1-(1-pq)) +e^{2T} (1+(1-pq))\Big)\\[2mm]
        &\hskip4cm+\ln\Big(2e^T \big(\cosh(T)+ (1-pq) \sinh(T)\big)\Big),\\[2mm]
        \int_{t}^T \eta_{s} ds &=  \frac{1}{2}\Big(-\ln(4e^{2t})+\ln(e^{-2T})\Big) + \ln\Big(e^{2t} (1-(1-pq)) +e^{2T} (1+(1-pq))\Big),
    \end{aligned}
\end{equation}
we conclude that
\begin{equation}
\label{eq:openform_barX_T_p-mixed-mfg}
    \bar X_T^{p-\rm MFG} = \dfrac{2 x_0 + q(1-p) \bar X_T^{\rMKV}(e^T-e^{-T})}{e^{T}(1 + (1-qp) ) + e^{-T} (1 - (1-qp)}.
\end{equation}

\end{proof}

\subsection{Convergence Proof for the General Model}
\label{app:pmix_convergence_general}
\begin{proof}
For completeness, we address Remark \ref{rem:fixedpoint-general-remark}
and
 give the convergence proof for a general model as given in Section~\ref{sec:model}.

In the general model, the analogue of~\eqref{eq:iter-barX-np1} would be 
\begin{equation*} 
\mu^{n+1}_t = Q_n \mu_t^{\rm MFC} + (1-Q_n) \tilde{\mu}_t^{n+1}, 
\end{equation*} 
where $(\tilde{\mu}^{n+1}_t)_{0 \le t \le T}$ solves the optimal control problem in environment $({\mu}^n_t)_{0 \leq t \leq T}$, namely 
$\tilde{\mu}^{n+1}_t$ is the law of $\tilde{X}_t^{n+1}$, where
\begin{equation}
\label{eq:FGMR:0}
\begin{split}
& d \tilde{X}_t^{n+1} = - \tilde{Y}_t^{n+1} dt + \sigma dB_t, 
 \\
& d \tilde{Y}_t^{n+1} =- \partial_x f_0(\tilde{X}_t^{n+1},\mu_t^{n}) dt + Z_t^{n+1} dB_t, 
 \\
& \tilde Y_T^{n+1} = \partial_x g(\tilde{X}_T^{n+1},\mu_T^n). 
\end{split} 
\end{equation} 
If we take item $(i$) in Assumption~\ref{assumption:SMP:MFG} for granted, 
it is straightforward to derive a uniform bound (in $L^\infty$) for the processes $(\tilde{Y}_t^{n})_{0 \le t \le T}$. We then deduce that the functions 
\begin{equation*} 
t \in [0,T] \mapsto \mu_t^n \in {\mathcal C}([0,T],{\mathcal P}_2({\mathbb R}^d)), \quad n \geq 0, 
\end{equation*}
form a relatively compact set, which is analogue to the \textbf{Step 1} of the proof of~(ii) of Theorem~\ref{the:iterative_decaying_p}. 

\textbf{Step 2} is the same as in the proof (ii) of Theorem~\ref{the:iterative_decaying_p} and we focus on \textbf{Step 3}:

We first repeat \eqref{eq:barX_cauchy}. We write (with $W_1$ being the 1-Wasserstein distance) 
\begin{equation*}
\begin{split}
&W_1\bigl( \mu_t^{n+m}, \mu_t^n \bigr) \leq C_1 \vert Q_{n-1} - Q_{n+m-1} \vert 
\\
&\hspace{15pt} + \sup_{\| f \|_{1,\infty} \leq 1}
\int_{{\mathbb R}^d} f(x) d\Bigl( (1-Q_{n-1})  \tilde{\mu}_t^{n}  -  (1-Q_{m+n-1}) \tilde{\mu}_t^{m+n} \Bigr)(x),
\end{split}
\end{equation*} 
 where $C_1$ is $\sup_{0 \leq t \leq T} \int_{{\mathbb R}^d} \vert x \vert d \mu_t^{\rm MFC}(x)$. 
 
 Then, denoting by $C_2$ the quantity $\sup_{n \geq 0} 
 \sup_{0 \leq t \leq T} \int_{{\mathbb R}^d} \vert x \vert d \mu^{n}_t(x)$, we obtain 
\begin{equation*}
\begin{split}
    W_1\bigl( \mu_t^{n+m}, \mu_t^n \bigr) 
    &\leq (C_1 + C_2) \vert Q_{m+n-1} - Q_{n-1} \vert 
    \\
    &\hspace{15pt} + (1-Q_{n-1}) \sup_{\| f \|_{1,\infty} \leq 1}
    \int_{{\mathbb R}^d} f(x) d\Bigl(   \tilde{\mu}_t^{n}  -  \tilde{\mu}_t^{m+n}
    \Bigr)(x),
\end{split}
\end{equation*}  
and it follows, 
\begin{equation}
\label{eq:FGMR:1}
\begin{split}
    W_1\bigl( \mu_t^{n+m}, \mu_t^n \bigr) 
    &\leq (C_1 + C_2) \vert Q_{m+n-1} - Q_{m-1} \vert 
    \\
    &\hspace{15pt} + (1-Q_{n-1}) {\mathbb E} \bigl[ \vert \tilde{X}_t^{n+m} - \tilde{X}_t^{n} \vert \bigr]. 
\end{split}
\end{equation}  
Next, there are several strategies to estimate 
\begin{equation*}
    {\mathbb E} \bigl[ \vert \tilde{X}_t^{n+m} - \tilde{X}_t^{n} \vert \bigr].
\end{equation*} 
One way (the most natural one under the standing assumption) is to use the non-degeneracy of $\sigma$ and next to write 
\begin{equation*} 
\tilde Y_t^{n} = \theta_t^n(t,\tilde X_t^n),
\end{equation*} 
where $\theta_t^n$ is the solution of the (system of) PDE:
\begin{equation*}  
    \partial_t \theta^n_t(x) + \tfrac12 {\rm Trace} \Bigl( \sigma \sigma^\top D^2_x \theta^n_t(x) \Bigr) - D_x \theta^n_t(x) \theta^n_t(x) + \partial_x f_0(x,\mu_t^{n-1}) = 0,
\end{equation*} 
with the boundary condition $\theta_T^n(x) = g(x,\mu_T^{n-1})$. By~\cite{delarue2002existence}, we have a bound for 
$D_x \theta^n$ that is independent of $n$. It is then possible to deduce that 
\begin{equation*} 
    \sup_{0 \leq t \leq T} 
    {\mathbb E} \bigl[ \vert \tilde{X}_t^{n+m} - \tilde{X}_t^{n} \vert \bigr]
    \leq 
    \Gamma \sup_{0 \leq t \leq T} {\mathbb E} \bigl[  \vert \theta^{n} (t,\tilde{X}_t^{n+m}) - \tilde{Y}_t^{n} \vert \bigr]. 
\end{equation*} 
In order to control the right-hand side, it suffices to expand 
\begin{equation*} 
    d_t \theta^{n}(t,\tilde{X}_t^{n+m})
\end{equation*} 
by means of It\^o's formula. This gives
\begin{equation*}
\begin{split} 
    &d_t \bigl( \theta^{n}(t,\tilde{X}_t^{n+m}) - \tilde{Y}_t^{n} \bigr) 
    \\
    &= D_x\theta^{n}(t,\tilde{X}_t^{n+m})\bigl( \theta^{n}(t,\tilde{X}_t^{n+m})  - \tilde{Y}_t^{n} \bigr) dt
    \\
    &\hspace{15pt} 
    +\Bigl( \partial_x f_0\bigl(  \tilde{X}_t^{n+m},\mu_t^{n+m-1}\bigr) 
     -
     \partial_x f_0\bigl(  \tilde{X}_t^{n+m},\mu_t^{n-1}\bigr) \Bigr) dt + dm_t, 
 \end{split}
\end{equation*} 
where $(m_t)_{0 \leq t \leq T}$ is a martingale whose form does not matter. 
At terminal time $T$, we have
$
    \theta^{n}(T,\tilde{X}_T^{n+m}) - \tilde{Y}_T^{n} 
    =
    \partial_x g(\tilde{X}_T^{n+m},\mu_T^{n+m-1}) - 
\partial_x g(\tilde{X}_T^{n+m},\mu_T^{n-1})$. 
Fixing $t \in [0,T]$, integrating between $t$ and $T$ and taking conditional
expectation given ${\mathcal F}_t$, we obtain 
\begin{equation*}
\begin{split} 
    \bigl\vert \theta^{n}(t,\tilde{X}_t^{n+m}) - \tilde{Y}_t^{n} 
    \bigr\vert 
    &\leq 
    c_0 W_1(\mu_T^{n-1},\mu_T^{n+m-1}) 
    + 
    c_0
    \int_t^T W_1(\mu_s^{n-1},\mu_s^{n+m-1}) ds
    \\
    &\hspace{15pt}
    + \Gamma {\mathbb E} \biggl[ \int_t^T 
    \bigl\vert \theta^{n}(s,\tilde{X}_s^{n+m})  - \tilde{Y}_s^{n} \bigr\vert ds 
    \, \vert \, {\mathcal F}_t \biggr]. 
 \end{split}
\end{equation*} 
Here, $c_0$ is the Lispchitz constant of the coefficients $f_0$ and $g$ in the measure argument, while $\Gamma$ is the bound for the gradient of 
$\theta^{n}$. Importantly, $\Gamma$ does not on depend on $c_0$.
Taking expectation on both sides and applying Gronwall's lemma, we obtain 
\begin{equation*} 
\begin{split}
    {\mathbb E}
    \bigl\vert \theta^{n}(t,\tilde{X}_t^{n+m}) - \tilde{Y}_t^{n} 
    \bigr\vert
    \leq c_0 (1+T)  \exp(\Gamma T) \sup_{0 \leq s \leq T} W_1(\mu_s^{n-1},\mu_s^{n+m-1}). 
\end{split}
\end{equation*} 
Back to
\eqref{eq:FGMR:1}, we get  
\begin{equation*}
\begin{split}
    W_1\bigl( \mu_t^{n+m}, \mu_t^n \bigr) &\leq (C_1 + C_2) \vert Q_{m+n-1} - Q_{m-1} \vert 
    \\
    &\hspace{15pt} +  c_0 (1+T)  \exp(\Gamma T) \sup_{0 \leq s \leq T} W_1(\mu_s^{n-1},\mu_s^{n+m-1}), 
\end{split}
\end{equation*}  
and we conclude as in the LQ case provided $c_0(1+T) \exp(\Gamma T)$ is small enough. 

It then remains to identify the limit $(\mu^\infty_t)_{0 \leq t \leq T}$. 
Following 
\eqref{eq:FGMR:0}, we solve
\begin{equation*}
\begin{split}
    & d \tilde{X}_t = - \tilde{Y}_t dt + \sigma dB_t, 
     \\
    & d \tilde{Y}_t =- \partial_x f_0(\tilde{X}_t,\mu_t^{\infty}) dt + Z_t^{n+1} dB_t, 
     \\
    & \tilde Y_T = \partial_x g(\tilde{X}_T\mu_T^\infty). 
\end{split} 
\end{equation*} 
We can prove, as above, that 
\begin{equation*} 
    \lim_{n \rightarrow \infty} {\mathbb E} \bigl[ \vert \tilde{X}_t^n - \tilde{X}_t \vert \bigr]=  0. 
\end{equation*} 
Recalling that 
\begin{equation*} 
    \mu^{n+1}_t = Q_n \mu_t^{\rm MFC} + (1-Q_n) \tilde{\mu}_t^{n+1} = Q_n \mu_t^{\rm MFC} + (1-Q_n) {\mathcal L} \bigl( \tilde{X}_t^{n+1} \bigr),
\end{equation*} 
we deduce that 
\begin{equation*}
    \mu_t^{\infty} = \bigl( 1 - p^*\bigr) \mu_t^{\rm MFC}
    + p^* {\mathcal L} \bigl( \tilde{X}_t \bigr), 
\end{equation*}
i.e., we have found a $p^*$-mixed equilibrium. 
\end{proof}

\newpage

\section{Notations}

\begin{table}[H]
\small
    \centering
    \begin{tabular}{p{1.5in}p{4.25in}}
         $J^{\bmu}(\balpha)= J(\balpha;\bmu)$ & Expected cost of the player using control $\balpha$ given population distribution $\bmu$ \\[2mm]
         $\hat{\alpha}(t,x,\mu,y)$& Minimizer of Hamiltonian $H(t, x, \mu, y, \alpha)$ as in \eqref{eq:hat-alpha-def}\\[2mm]
         $\balpha^{{\rMKV}}$& MFC optimal control  \\[2mm]
         $\balpha^{\rMFG}$& MFG equilibrium control  \\[2mm]
         $\bX^{\balpha}$& State trajectory controlled by $\balpha$  \\[2mm]
         $\bX^{{\rMKV}}$& State trajectory controlled by $\balpha^{{\rMKV}}$ i.e., $\bX^{\balpha^{{\rMKV}}}$ \\[2mm]
         $\bX^{\rMFG}$& State trajectory controlled by $\balpha^{\rMFG}$ i.e., $\bX^{\balpha^{\rMFG}}$  \\[2mm]
         $\bmu^{\balpha}$& Law of state trajectory (i.e. mean field) controlled by $\balpha$  \\[2mm]
         $\bmu^{{\rMKV}}$& Law of state trajectory (i.e. mean field) controlled by $\balpha^{{\rMKV}}$  \\[2mm]
         $\bmu^{\rMFG}$& Law of state trajectory (i.e. mean field) controlled by $\balpha^{\rMFG}$  \\[2mm]
         $J^{\lambda,{\rm MF}}( \balpha ;{\boldsymbol \mu} )$& Expected $\lambda$-interpolated mean field cost of the player using control $\balpha$ given population distribution $\bmu$ \\[2mm]
         $\balpha^\lambda$& $\lambda$-interpolated mean field equilibrium control \\[2mm]
         $\bX^{\lambda}$& State trajectory controlled by $\balpha^{\lambda}$  \\[2mm]
         $\bmu^{\lambda}$& Law of state trajectory (i.e., mean field) controlled by $\balpha^{\lambda}$ \\[2mm]
         $\hat{\balpha}^p$& $p$-partial mean field equilibrium control\\[2mm] 
         $\hat{J}_p $ & Expected cost of the non-cooperative player at the equilibrium when a proportion $p$ of the population is deviating \\[2mm] 
         $J^*_p $ & Expected cost of the player continuing to follow the social planner at the equilibrium when a proportion $p$ of the population is deviating \\[2mm] 
         $\hat{J}_1 $ & Expected cost of representative player in MFG equilibrium \\[2mm] 
         $J^* = J^{\rMKV}(\balpha^{\rMKV})=J^*_0$ & Expected cost of a player under social planner's regulation (MFC optimum cost) \\[2mm]
    \end{tabular}
    \label{tab:notations_ch2_3}
\end{table}

In Section 4, when $\balpha^{\rMKV}$ and $\bmu^{\rMKV}$ notations are used, they represent the original MFC optimum control and mean field controlled by the original MFC optimum control.

\end{document}